\documentclass[reprint,aps,showpacs,prl,nofootinbib]{revtex4-1}
\usepackage{color}
\usepackage{newlfont}
\usepackage{graphicx}
\usepackage{amssymb}
\usepackage{amsmath}
\usepackage{lineno}

\usepackage[dvipsnames]{xcolor}

\definecolor{brown}{gray}{0.6}

\usepackage[caption=false]{subfig}
  
\usepackage{bm}
\usepackage{verbatim}
\usepackage[latin1]{inputenc}
\usepackage{bbm}
\usepackage{multirow}
\usepackage[thinlines]{easytable}
\usepackage{braket}
\usepackage{amsthm}
\usepackage{bbold}
\usepackage{subfig}
\usepackage{mathrsfs}
\usepackage[varg]{txfonts} 
\usepackage{amsfonts}
\usepackage{eqnarray}

\newcommand{\be}{\begin{equation}}
\newcommand{\ee}{\end{equation}}
\newcommand{\beq}{\begin{eqnarray}}
\newcommand{\eeq}{\end{eqnarray}}
\newcommand{\ba}{\begin{array}}
\newcommand{\ea}{\end{array}}

\DeclareMathOperator{\Tr}{Tr}

\newtheorem{result}{Result}
\newtheorem*{result2}{Result 2}

\newtheorem{corolarios}{Corollary}

\newtheorem{lemas}{Lemma}
\begin{document}

\title{On the consistency of measurement protocols for quantum processes fluctuations}
\author{Thales A. B. Pinto Silva}
\email{pinto\_silva@campus.technion.ac.il}
\author{David Gelbwaser-Klimovsky}
\affiliation{Schulich Faculty of Chemistry and Helen Diller Quantum Center, Technion-Israel Institute of Technology, Haifa 3200003, Israel}

\begin{abstract}
Quantum fluctuations are fundamental to quantum technologies, affecting quantum computing, sensing, cryptography, and thermodynamics. A paradigmatic example is quantum work, which, for a thermally isolated driven system, is identified with the variation of its energy. Although quantum mechanics provides precise rules for measuring energy and other observables at individual instants of time, it does not provide a standard framework for characterizing the statistics of their variations between two times. This ambiguity has led to competing protocols that can assign different work distributions--and, more generally, different fluctuation statistics--to the same physical process, with consequences for both foundational physics and quantum technologies. In this work, we propose four fundamental criteria that any consistent protocol for measuring energy variations must satisfy, grounded in conservation laws, state independence of the measurement apparatus, the no-signaling principle, and constraints on physical reality. We prove that these criteria uniquely select the two-time quantum observable protocol. We then show that this conclusion extends beyond work and energy to the variation of arbitrary physical observables, including charge, particle number, and momentum. This result has the potential to establish the foundations for measurements of quantum processes, possibly resolving ambiguities in quantum fluctuation measurements. Moreover, it enables the extension of quantum information concepts, such as entanglement and Bell's inequalities, to processes rather than instantaneous observables.
\end{abstract}

\maketitle

\section{Introduction}

Quantum fluctuations play an increasingly critical role in emerging devices~\cite{Janine2024}, challenging the performance of quantum computers~\cite{Arute2019,Neill2017,Lopez2024}, sensors~\cite{Zhao2011}, and cryptographic systems~\cite{Bozzio2022}. Central to controlling and predicting these fluctuations is the ability to accurately characterize the \emph{variation} of physical quantities (VPQs). For instance, in superconducting quantum computers, minimizing the impact of charge fluctuation is essential for ensuring reliable qubit performance \cite{Arute2019,Neill2017}. Likewise, in quantum thermodynamics, the statistical characterization of VPQs, such as variation of energy and particle number, is vital for refining fluctuation theorems and extending thermodynamic and conservation laws into the quantum realm \cite{Zhang2024,Janine2024,Campisi2011,Talkner2007,Hovhannisyan2024,Hanggi2017,Loveridge2011,Gisin2018,Aharonov2016,Aamir2025}. Remarkably, despite significant advancements, \emph{there remains no universal and standard framework for characterizing the statistics of VPQs in quantum mechanics}.

A paradigmatic instance of this problem is the long-standing question of how to define and measure work, and its fluctuations, in quantum mechanics~\cite{Campisi2011,Talkner2007}. Consider the elementary setting of a thermally isolated quantum system $\mathcal{S}$ evolving unitarily according to $U_t$ during the time interval $[0,t]$, while its Hamiltonian is changed from $H(0)$ to $H(t)$. Since the system is isolated and no heat is exchanged, the work performed on the system is identified with its change in energy. The task is then to characterize the statistics of this energy variation. Quantum mechanics provides an unambiguous prescription for the statistics of the energy at any single time. If ${\ket{e_k(t)}}$ denotes the eigenbasis of $H(t)$, the probability of obtaining the outcome $e_k(t)$ at time $t$ is
\begin{equation}
p_k(t|\rho)=\bra{e_k(t)}U_t\rho U_t^\dagger\ket{e_k(t)} ,
\end{equation}
for an arbitrary initial state $\rho$. This fully specifies the statistics of $H(t)$ at that time. However, quantum mechanics does not provide an equally standard prescription for the statistics of the \emph{variation} of $H$ over the interval $[0,t]$. The reason is that such a variation is not a single-time property: it is a quantity relating two instants of time, and any attempt to infer it through measurements may itself disturb the subsequent dynamics.

A common approach to measure work in this context is the two-point measurement (TPM) protocol~\cite{Campisi2011,Talkner2007,Roncaglia2014,Perarnau2017,Baumer2018,Batalhao2014}. In this protocol, the system is first projectively measured in the eigenbasis of $H(0)$, yielding an outcome $e_n(0)$. The post-measurement state then evolves under $U_t$ until the final time $t$, where a second projective measurement of $H(t)$ is performed, yielding $e_m(t)$. In each run, the work value is assigned as
\begin{equation}
w_{mn}=e_m(t)-e_n(0),
\end{equation}
and the work distribution is reconstructed by repeating the protocol over many runs. In this sense, the TPM protocol operationally defines work as the difference between two measured energy outcomes.

Although intuitive, the TPM approach has several drawbacks. Perhaps the most important one is that the initial measurement collapses any superpositions in $\rho$ with respect to the eigenbasis of $H(0)$. This can lead TPM to fail to respect conservation laws \cite{Perarnau2017,Baumer2018,Silva2021}\footnote{We later exemplify cases of such conservation disturbances in section ``Trapped ion case study''}. These and other limitations~\cite{Silva2024,Silva2023} have prompted the development of alternative methodologies for describing the statistics of work, including the well-known full-counting statistics~\cite{Esposito2009}, quasi-probabilities~\cite{Gherardini2024}, Gaussian pointers~\cite{Talkner2016}, and two-time observables \cite{Silva2021,Silva2023,Silva2024,Maquedano2024,Lindblad1983,Allahverdyan2005}, among others~\cite{Baumer2018}. These approaches, however, often produce conflicting results, raising critical concerns for both foundational physics and practical applications. Indeed, because they provide contradictory predictions, it is unclear which approach, if any, accurately reflects the underlying physics. Also, in the absence of a standardized protocol, physical conclusions risk being influenced more by the measurement procedure chosen itself than by the process under investigation. Furthermore, the lack of a reliable standard has critical implications for quantum technologies. Predictions must be accurate and consistent across diverse scenarios to ensure quantum devices' functionality and resource efficiency--faulty or overly demanding predictions can undermine their practical feasibility. 

Without a universally accepted standard methodology, grounding protocols for measuring work in fundamental principles of physics offers a clear and reliable criterion to ensure their consistency and applicability across diverse scenarios~\cite{Hovhannisyan2024,Brandao2013,Brandao2015,Silva2024,Perarnau2017,Guryanova2016}. This was the strategy pursued in earlier works, where several physically motivated requirements were proposed for quantum work-measurement protocols. In particular, it was argued that a consistent protocol should reproduce the appropriate classical behavior in the classical limit, satisfy average conservation laws, and be implemented by an apparatus that is independent of the initial state of the system. Under one formulation of these assumptions~\cite{Perarnau2017}, a no-go theorem was obtained, showing that no measurement protocol could satisfy all the \emph{proposed} requirements simultaneously. In Ref.~\cite{Silva2024}, we revisited this conclusion. We demonstrated that the classicality criterion used in the no-go theorem does not properly capture the quantum-to-classical transition. By replacing it with a necessary condition for classicality, we showed that work can consistently be represented by a two-time quantum observable, namely by the difference between the Heisenberg-evolved final Hamiltonian and the initial Hamiltonian. Moreover, among the well-known work-measurement schemes reviewed in Ref.~\cite{Baumer2018}, this two-time observable protocol was the only one satisfying the proposed criteria.

However, Ref.~\cite{Silva2024} left two important questions open. First, it did not prove uniqueness among all possible protocols: it compared the two-time observable protocol with the main existing approaches in the literature. Therefore, it remained logically possible that another, yet unidentified, protocol could satisfy the same physical requirements while predicting different fluctuation statistics. Second, the discussion was formulated primarily in the language of work and energy variations. Yet the underlying issue is more general. The same conceptual tension arises whenever one attempts to assign statistics to the variation of a quantum observable between two times, such as angular momentum, linear momentum, particle number, charge, or current. Thus, the ambiguity in quantum work statistics is not an isolated peculiarity of thermodynamics, but rather a manifestation of a broader structural problem in quantum theory: how to consistently define the statistics of variations of physical quantities.

Motivated by this open problem, we propose four fundamental properties that any consistent protocol for measuring variation of energy should satisfy. Besides adapting previously proposed principles that enforce conservation laws and a state-independent measurement apparatus, we impose reality and locality constraints inspired by the EPR argument~\cite{EPR1935}. These constraints are meant to play the same foundational role here: they restrict admissible measurement protocols by requiring compatibility with no-signaling and with the existence of definite values whenever such values can be inferred without disturbance.. By assuming these criteria, we prove that only a single measurement protocol meets all principles simultaneously: the two-times quantum observable protocol~\cite{Silva2021,Silva2023,Silva2024,Maquedano2024,Lindblad1983,Allahverdyan2005,Bochkov1977}. We then show that the result is not restricted to work or to variations of energy. Owing to the general structure of the proof, the same reasoning applies to arbitrary VPQs. The two-time quantum observable protocol therefore emerges as the unique standard for assigning statistics to variations of physical quantities in quantum mechanics, while simultaneously respecting conservation laws, locality through no-signaling, and the required constraints on physical reality.

\section{Setup and measurement protocols}

In our framework, we consider a general system $\Omega$, comprising arbitrary subsystems, prepared in an arbitrary quantum state $\rho$ acting on a Hilbert space $\mathcal{H}$. $\Omega$ evolves from time $0$ to $t$ under an arbitrary unitary operator $U$, with a countable eigenbasis. We focus on describing the measurement of the variation of \emph{a specific part} of the total energy of the whole system $\Omega$, described by a time-independent Hermitian operator $H_1$ acting on $\mathcal{H}$.

This notation is deliberately general. A standard example is provided by the system-environment setting underlying quantum open-system dynamics. In this case, $\Omega=S+E$ is composed of a system $S$ and an environment $E$, with total Hamiltonian
\begin{equation}
H = H_S \otimes \mathbbm{1}_E + \mathbbm{1}_S \otimes H_E + V_{SE},
\end{equation}
where $H_S$ and $H_E$ are the local Hamiltonians and $V_{SE}$ is their interaction. One may then be interested not in the variation of the total energy of $\Omega$, but rather in the variation of the energy associated with the subsystem $S$, represented on the full Hilbert space by
\begin{equation}
H_1 = H_S \otimes \mathbbm{1}_E .
\end{equation}
The same formalism also applies to quantum-information settings. For example, in a register of interacting qubits, the global dynamics may be generated by a sequence of gates $U_1,U_2,\ldots$, while the quantity of interest may be the variation of the energy of a particular qubit, $H_1=H_{q_1}$, or of a collection of qubits. Similarly, in many-body systems, $H_1$ may describe the energy of a local region, a single particle, or a selected interaction term. Thus, the role of $H_1$ is to identify the physical contribution whose variation is being probed, while $U$ specifies the quantum process during which this variation takes place.

Although we phrase the discussion in terms of energy variations, no step of the construction depends on this choice. The same framework applies to the variation of any Hermitian observable, such as particle number, charge, linear momentum, or angular momentum. Moreover, by allowing $H_1$ to depend explicitly on time, the formalism also covers the usual driven scenario of quantum thermodynamics, where the variation of energy of a thermally isolated system is connected with work.

We first clarify what we mean by a measurement protocol. In quantum mechanics, measurements can be generally described by positive operator-valued measures (POVMs)~\cite{Nielsen2010,Heinosaari2012,Busch1996}. Following this perspective, we define a \emph{measurement protocol} $\mathbbm{M}$ \cite{Perarnau2017} to measure the \emph{variation} of any arbitrary energy operator $H_1$ under an arbitrary evolution $U$. For the protocol $\mathbbm{M}$ and for each $(H_1, U)$ pair, the set $\mathbbm{M}(H_1, U) = \{M(z, H_1, U)\}$ defines a POVM whose operators satisfy $\int_{-\infty}^{\infty} dz  M(z, H_1, U) = \mathbbm{1}$ and $M(z, H_1, U) \geq 0$. In this framework, the probability density of observing a variation $z$ of $H_1$ under $U$ is $\wp(z, H_1, U, \rho) = \Tr [M(z, H_1, U) \rho]$, for the initial state $\rho$.  By the POVM properties, it follows that $\int_{-\infty}^{\infty} dz  \wp(z, H_1, U, \rho) = 1$ and $\wp(z, H_1, U, \rho)\geq 0$. This approach aims to keep the protocol as general as possible, so that $\mathbbm{M}$ applies universally. For now, we assume that $\mathbbm{M}$ can depend on infinitely many different variables. The only restriction we impose is that, among these, it necessarily depend on $U$ and $H_1$.

An important example of a measurement protocol is the two-times observables (OBS) protocol, denoted by $\mathbbm{M}{\text{\tiny OBS}}$. This protocol operates as follows: for any given energy operator $H_1$ and unitary evolution $U$ over time $t$, the variation of energy is defined by the two-time quantum observable \cite{ Silva2021, Silva2023,Silva2024,Lindblad1983,Bochkov1977}: 
\begin{equation} 
\Delta (H_1,U) = U^{\dagger} H_1 U - H_1, \label{TTOBS} 
\end{equation}
representing the difference between the Heisenberg picture operators $U^{\dagger} H_1 U$ and $H_1$ at times $t$ and $0$. $\Delta (H_1, U)$ is a Hermitian operator, with eigenvalues $\{\delta_j(H_1, U)\}$ and corresponding eigenvectors $\{\ket{\delta_j(H_1, U)}\}$, as expressed in the decomposition $\Delta (H_1, U) = \sum_{j} \delta_j(H_1, U) \ket{\delta_j(H_1, U)} \bra{\delta_j(H_1, U)}$. The probability of finding a specific eigenvalue $\delta_j(H_1, U)$ is $p_j(H_1, U,\rho) = \Tr[P_{j}(H_1, U) \rho]$, where $P_{j}(H_1, U) = \ket{\delta_j(H_1, U)} \bra{\delta_j(H_1, U)}$ and $\rho$ is the initial state. Accordingly, the OBS protocol $\mathbbm{M}{\text{\tiny OBS}}$ is defined by POVMs $\mathbbm{M}{\text{\tiny OBS}}(H_1, U) = \{M_{\text{\tiny OBS}}(z, H_1, U)\}$, where each element is given by $M_{\text{\tiny OBS}}(z, H_1, U) = \sum_{j} \delta^{\text{\tiny \textbf{D}}}[z - \delta_j(H_1, U)] P_{j}(H_1, U)$, and $\delta^{\text{\tiny \textbf{D}}}$ is the Dirac's delta~ \cite{Silva2021, Silva2023,Silva2024}. Consequently, the probability density 
\begin{equation}
    \wp_{\text{\tiny OBS}}(z, H_1, U, \rho) = \Tr[M_{\text{\tiny OBS}}(z, H_1, U) \rho]\label{obsdistrib}
\end{equation}
provides the probability density of measuring a variation of $H_1$ equal to $z$ under evolution $U$. 
To our knowledge, Bochkov and Kuzovlev were the first one to consider this protocol in their seminal work~\cite{Bochkov1977}. There, they proposed their famous fluctuation theorem and extended their framework to quantum mechanics considering the OBS protocol. Later, Lindblad also considered such observable in Ref.~\cite{Lindblad1983}.

Another widely-used example is TPM protocol $\mathbbm{M}_{\text{\tiny TPM}}$, briefly described in the introduction. When considering the TPM protocol for any $(H_1, U)$ pair, we define~\cite{Roncaglia2014,Perarnau2017} $\mathbbm{M}_{\text{\tiny TPM}}(H_1, U) = \{M_{\text{\tiny TPM}}(z, H_1, U)\}$, where $M_{\text{\tiny TPM}}(z, H_1, U) = \sum_{jk} \delta^{\text{\tiny \textbf{D}}}[z - (e_j - e_k)] |\bra{e_j} U \ket{e_k}|^2 \ket{e_k} \bra{e_k}$, with $\ket{e_j}$ and $\ket{e_k}$ eigenvectors of $H_1$ having eigenvalues $e_j$ and $e_k$, respectively. The probability density is thus given by 

\begin{equation}
    \wp_{\text{\tiny TPM}}(z, H_1, U, \rho) = \Tr[M_{\text{\tiny TPM}}(z, H_1, U) \rho]\label{eq:TPMprobgeneral}
\end{equation}.

\section{Consistency physical principles: CRIN conditions}
\subsection{Motivation: a goal-oriented task}

We now motivate the physical principles that will be used to constrain measurement protocols for energy variations. Our goal is not to provide an exhaustive axiomatization of all possible requirements that such protocols should satisfy. This would be too broad a task and would possibly depend on the physical context under consideration. Instead, we seek a small set of general consistency requirements that are independently motivated and sufficiently restrictive to test whether a proposed protocol can be regarded as universal.

Previous works have shown that even a limited number of assumptions can impose strong constraints on quantum work-measurement schemes~\cite{Perarnau2017,Hovhannisyan2024,Silva2024}. Three ideas have played a central role in this discussion: compatibility with conservation laws, independence of the measurement scheme from the initial state, and agreement with the classical limit. The first two have a direct operational meaning. If two parts of an isolated system exchange energy while their sum is conserved, a measurement protocol should not predict a statistically detectable violation of this balance. Similarly, a protocol should not require prior knowledge of the initial state $\rho$ in order to define the apparatus; otherwise, the measurement scheme would be tuned to the preparation whose variation it is supposed to characterize.

Compatibility with the classical limit is also physically natural, but it is more difficult to use as a decisive criterion. Different areas of physics employ different notions of classicality. For instance, in some thermodynamic settings the classical regime is associated with the absence of coherences in a preferred energy basis \cite{Jarzynski2015,Perarnau2017}, whereas in quantum optics coherent states are often regarded as classical~\cite{Cohen2020}. In Ref.~\cite{Silva2024}, we sought to address this difficulty by formulating necessary, though not sufficient, conditions for a classical limit: in any genuinely classical regime, the effects of noncommutativity should become negligible and the quantum statistics should converge to an appropriate classical description. This was enough to show that the OBS protocol can satisfy the previously proposed consistency requirements, but it does not by itself establish \emph{uniqueness} among all possible POVM protocols.

This motivates looking for additional principles that do not rely exclusively on a specific notion of classicality. A natural source of guidance is the foundational concepts of reality and locality. Historically, these notions have played a special role in identifying what one should demand from a physical theory. In the EPR argument~\cite{EPR1935}, Einstein, Podolsky and Rosen used precisely these ideas to challenge the completeness of quantum mechanics: if a physical quantity can be predicted with certainty without disturbing the system, then it should correspond to an element of physical reality, and operations performed on a distant system should not affect the local physical description. Bell's theorem~\cite{Bell1964} later showed that this classical combination of assumptions cannot be maintained as a universal hidden-variable explanation of quantum predictions. These results and consequent discussions are nowadays in the core of advances in quantum information theory \cite{NobelPhysics2022}. Therefore, rather than making these ideas irrelevant, this history shows that reality and locality are powerful principles for testing the consistency of physical descriptions, provided they are formulated in a way that is compatible with quantum mechanics.

This is the sense in which we use them here. We do not assume EPR realism, hidden variables, or Bell locality. Instead, we extract two operational requirements appropriate for measurement protocols of variations. The first is that, whenever quantum mechanics itself assigns sharp values to the relevant observable at the initial and final times, the corresponding variation should also be assigned a sharp value. The second is that the statistics assigned to a local variation should not be affected by an independent operation performed on an unrelated auxiliary system. These considerations motivate the additional principles introduced below.

\subsection{Definition of the CRIN conditions}
We now precisely define the four fundamental conditions that we expect any consistent protocol $\mathbbm{M}$ for measuring energy variations to satisfy. For didactic purposes, we illustrated them in Fig. \ref{conditionsfig}.
\begin{figure}
    \centering
    \includegraphics[width=1\linewidth]{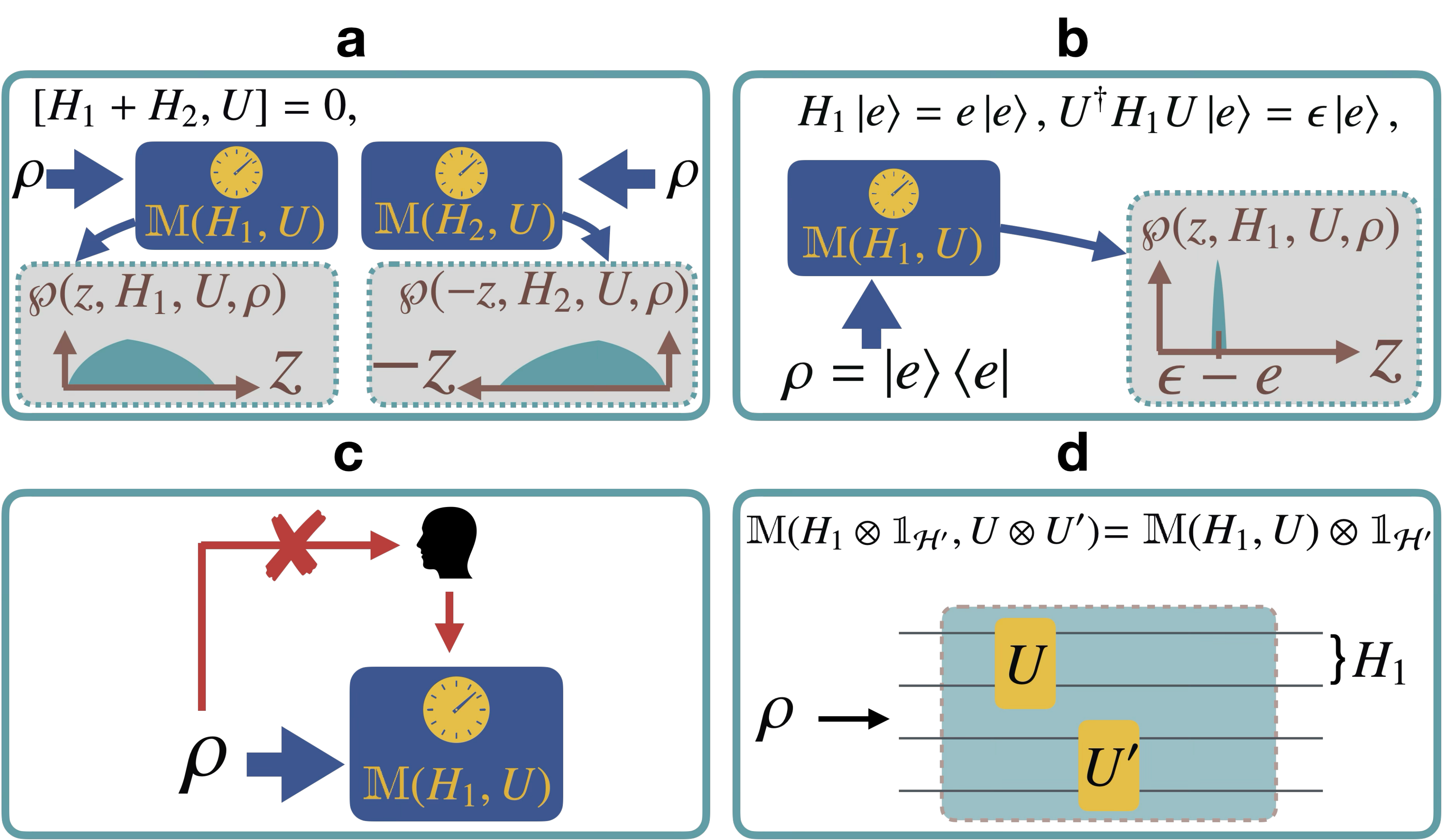}
    \caption{\textbf{Schematic illustration of the four key conditions for a consistent measurement protocol of VPQs}. \textbf{(a)} Energy Conservation (Condition 1): If $[H_1 + H_2, U] = 0$, the total energy $H_1 + H_2$ must be conserved. This means the measurement protocols $\mathbbm{M}(H_1,U)$ and $\mathbbm{M}(H_2,U)$ should satisfy the relation $\wp(z, H_1, U, \rho) = \wp(-z, H_2, U, \rho)$ for the same initial state $\rho$. \textbf{(b)} Reality Condition (Condition 2): If the whole system is in an eigenstate of the energy $H_1$ at times 0 and $t$, with eigenvalues $e$ and $\epsilon$, the measurement distribution for the variation of $H_1$ should collapse to a delta function $\delta^{\text{\tiny \textbf{D}}}[z - (\epsilon - e)]$, reflecting a well-defined energy difference. \textbf{(c)} Independence of the state (Condition 3): The measurement protocol must not depend on the initial state $\rho$, ensuring that the apparatus is set independently of prior knowledge about the system. \textbf{(d)} No-signaling (Condition 4): The measurement protocol must ensure that a local evolution $U'$ on one part of a bipartite system does not affect the energy variation statistics of $H_1$ in the other part, preserving the no-signaling principle during the evolution $U \otimes U'$.}
    \label{conditionsfig}
\end{figure}
\begin{enumerate}    
    \item \textbf{Conservation laws}: \emph{For any preparation $\rho$, unitary evolution $U$, and energy operators $H_1$ and $H_2$ representing parts of the energy of any system $\Omega$, if $[H_1+H_2,U]=0$, then, for any $z$, $\Tr [M(z, H_1,U)\rho]=\wp(z,H_1,U,\rho)=\wp(-z,H_2,U,\rho)=\Tr [M(-z, H_2,U)\rho]$}. In other words, if the sum of energies $H_1+H_2$ is conserved under $U$, then the probability of $H_1$ of increasing an amount $z$ must equal the probability of $H_2$ of \emph{decreasing} the same amount. This condition ensures that conservation laws hold at the level of probability distributions instead of just on average. In light of the Wigner-Araki-Yanase (WAY) theorem~\cite{Loveridge2011,Gisin2018}, it is interesting to note that this condition allows for the possibility that the measurement process may, in specific rounds, disturb the subsystem's energy or even not be repeatable. However, even with such single-shot disturbances, we assume that conservation laws remain preserved, ensuring no statistically detectable violation in the energy balance (see Fig. \ref{conditionsfig}\textbf{a}).
    \item \textbf{Reality}:  \emph{Consider  any system $\Omega$, operator $H_1$ and evolution $U$. If the initial state is $\rho=\ket{e}\bra{e}$ such that $\ket{e}$ is an eigenvector of both $H_{1}$ and $U^{\dagger}H_1 U$ with respective eigenvalues $e$ and $\epsilon$, then the POVM must result in the probabilities $\wp(z,H_1,U,\rho)=\delta^{\text{\tiny \textbf{D}}}[z-(\epsilon-e)]$ for this specific $\rho=\ket{e}\bra{e}$.} In other words, if $H_1$ is well-defined (or real, in an EPR-sense~\cite{EPR1935}) at both the start and end of the process, the variation should be precisely the difference between the initial and final eigenvalues (see Fig. \ref{conditionsfig}\textbf{b}). 
    \item \textbf{ Independence of the initial state}: \emph{For any system $\Omega$, operators $H_1$, and evolution operators $U$, the elements of the POVM ${M(z, H_1, U)}$ must not depend on the initial state $\rho$}. This condition is generally satisfied in quantum mechanics \cite{Nielsen2010,Busch1996,Heinosaari2012} and ensures that the measurement apparatus is not adjusted based on the initial preparation of the system~\cite{Hovhannisyan2024,Baumer2018,Perarnau2017,Silva2024} (see Fig. \ref{conditionsfig}\textbf{c}). Notably, this condition can also be formulated in a weaker form, based on a linearity requirement: the probabilities generated by the protocol must depend linearly on the initial state $\rho$. As shown in Ref.~\cite{Perarnau2017}, this is mathematically equivalent to requiring a state-independent POVM $\mathbbm{M}(H_1,U)$.
    \item \textbf{No-signaling}: \emph{Consider a system $\Omega$ evolving under an arbitrary bipartite unitary evolution $U\otimes U'$ acting on a bipartite Hilbert space $\mathcal{H}\otimes \mathcal{H}'$. For any such system and every energy operator $H_1\otimes \mathbbm{1}_{\mathcal{H}'}$ acting locally on $\mathcal{H}$, $\mathbbm{M}$ is such that its POVM elements satisfy $M(z,H_1\otimes \mathbbm{1}_{\mathcal{H}'},U\otimes U')=M(z,H_1,U)\otimes \mathbbm{1}_{\mathcal{H}'}$ for every $z$.} In other words, local \emph{statistics} of $H_1$ should remain unaffected by changes ($U'$) in another subspace, ensuring that no statistically detectable information is transmitted between different subspaces via the measurement process \cite{Peres2004,Horodecki2019} (see Fig. \ref{conditionsfig}\textbf{d}). This condition should be crucial in setups where additional auxiliary subsystems are introduced to assist measurements~\cite{Talkner2016,Mohammady2019}.
\end{enumerate} 
We call these the CRIN conditions and the measurement protocols that satisfies them as CRIN protocols.

\section{Main results}
\label{sec:mainresults}

We are now in position to establish the main result of this work:

\begin{result} The OBS protocol is the only CRIN protocol. \end{result}

We here prove this result for the case in which $\Delta(H_1, U)$ has discrete, non-degenerate basis. The more general case is proved in Appendix B.

\begin{proof}
We begin by proving that the OBS protocol is a CRIN protocol. Since the POVM $\mathbbm{M}_{\text{\tiny OBS}}(H_1,U)$ does not depend on the initial state $\rho$ for any pair $(H_1, U)$, the condition 3 is satisfied. Second, for two energy operators $H_1$ and $H_2$ conserved under evolution $U$ (i.e., $[H_1 + H_2, U] = 0$), it follows that $\Delta(H_1, U) = -\Delta(H_2, U)$. This implies the eigenvectors of both operators coincide, $\ket{\delta_j(H_1, U)} = \ket{\delta_j(H_2, U)}$, with eigenvalues related by $\delta_j(H_1, U) = -\delta_j(H_2, U)$. Consequently, for any $z$, $M_{\text{\tiny OBS}}(z, H_1, U) = M_{\text{\tiny OBS}}(-z, H_2, U)$, ensuring condition 1. Third, if $\rho_1 = \ket{e_1}\bra{e_1}$ is an eigenstate of both $H_1$ and $U^\dagger H_1 U$, with eigenvalues $e_1$ and $\epsilon_1$, then $\ket{e_1}$ is also an eigenstate of $\Delta(H_1, U)$ with eigenvalue $\epsilon_1 - e_1$. Thus, the probability distribution $\wp_{\text{\tiny OBS}}(z, H_1, U, \rho_1) = \delta^{\text{\tiny \textbf{D}}}[z - (\epsilon_1 - e_1)]$ satisfies condition 2. At last, for a bipartite system $\Omega$ evolving under $U \otimes U'$ acting on a Hilbert space $\mathcal{H}\otimes\mathcal{H}'$, the local operator $H_1 \otimes \mathbbm{1}_{\mathcal{H}'}$ evolves as $\Delta(H_1 \otimes \mathbbm{1}_{\mathcal{H}'}, U \otimes U') = \sum_j \delta_j(H_1, U) P_j^\Delta(H_1, U) \otimes \mathbbm{1}_{\mathcal{H}'}$. This ensures that $M_{\text{\tiny OBS}}(z, H_1 \otimes \mathbbm{1}_2, U \otimes U') = M_{\text{\tiny OBS}}(z, H_1, U) \otimes \mathbbm{1}_{\mathcal{H}'}$, fulfilling condition 4. $\mathbbm{M}_{\text{\tiny OBS}}$ is indeed a CRIN protocol.
    
    The final and most crucial task is to prove that the OBS protocol is the \emph{only} one satisfying the CRIN conditions. Specifically, we aim to demonstrate that for any CRIN protocol $\mathbbm{M}'$, and any $H_1$ and $U$, any element $M'(z, H_1, U)$ of the POVM $\mathbbm{M}'(H_1, U)$ satisfy  
    \begin{equation}
        M'(z, H_1, U) = M_{\text{\tiny OBS}}(z, H_1, U), \label{assumptioninitial}
    \end{equation}
    for all $z$. To this end, consider an arbitrary CRIN protocol $\mathbbm{M}'$. For simplicity, assume arbitrary $H_1$ and $U$ such that $\Delta(H_1, U)$ is diagonal in a discrete non-degenerate basis $\{\ket{\delta_i(H_1, U)}\}$. The general cases are addressed similarly in the Appendix B. We first establish that 
    \begin{equation}
        \ba{l}
        \bra{\delta_{i}(H_1,U)} M'(z,H_1,U)\ket{\delta_{i}(H_1,U)} = \delta^{\text{\tiny \textbf{D}}}[z-\delta_{i}(H_1,U)] = \\
        = \bra{\delta_{i}(H_1,U)}M_{\text{\tiny OBS}}(z, H_1,U)\ket{\delta_{i}(H_1,U)},
        \ea\label{Mdiagonal}
    \end{equation}
    for all elements $\{\ket{\delta_{i}(H_1,U)}\}$ of the eigenbasis. To prove this, consider operators $U'$ and $H_2$ acting on $\mathcal{H}'$ and $\mathcal{H}\otimes \mathcal{H}'$, respectively, along with a vector $\ket{v}\in \mathcal{H}'$, such that:
    \begin{eqnarray}
        &[H_1 \otimes \mathbbm{1}_{\mathcal{H}'} + H_2, U \otimes U'] = 0, \label{commutresult10met} \\
        &H_2 \ket{\delta_i(H_1, U), v} = E_i \ket{\delta_i(H_1, U), v}, \\
        &(U^\dagger \otimes U^{'\dagger}) H_2 (U \otimes U') \ket{\delta_i(H_1, U), v} = E_i' \ket{\delta_i(H_1, U), v}, \label{result12ndrelation0met}
    \end{eqnarray}
    where $E_i' = E_i - \delta_i(H_1, U)$. As we demonstrate in Result 2, the existence of such $H_2$, $U'$, and $\ket{v}$ is guaranteed. Since $\mathbbm{M}'$ is CRIN, Conditions 1 and 3 imply for $U \otimes U'$ in Eq.~\eqref{commutresult10met} and all $z$ that (see Appendix A)
    \begin{equation}
        M'(z, H_1 \otimes \mathbbm{1}_{\mathcal{H}'}, U \otimes U') = M'(-z, H_2, U \otimes U'). \label{req2proof}
    \end{equation}
    Moreover, considering Condition 4, we obtain
    \begin{equation}
         M'(z, H_1 \otimes \mathbbm{1}_{\mathcal{H}'}, U \otimes U') = M'(z, H_1, U) \otimes \mathbbm{1}_{\mathcal{H}'}.
         \label{localtoglobal}
    \end{equation}
    Considering the initial state $\rho_i = \ket{\delta_{i}(H_1,U),v}\bra{\delta_{i}(H_1,U),v}$ and condition 2, we obtain for any arbitrary $z'$:
    \begin{equation}
        \Tr [M'(z', H_2, U \otimes U') \rho_i] = \delta^{\text{\tiny \textbf{D}}}[z' + \delta_{i}(H_1,U)]. \label{Mdiagonal2}
    \end{equation}
    Substituting $z' \to -z$, we have:
    \begin{equation}
        \Tr [M'(-z, H_2, U \otimes U') \rho_i] = \delta^{\text{\tiny \textbf{D}}}[z - \delta_{i}(H_1,U)]. \label{Mdiagonal3}
    \end{equation}
    Combining Eqs.~\eqref{req2proof}, \eqref{localtoglobal}, and \eqref{Mdiagonal3}, we obtain Eq.~\eqref{Mdiagonal}. This deduction is valid for arbitrary $H_1$, $U$, $z$, and $\ket{\delta_{i}(H_1,U)}$, demonstrating Eq.~\eqref{Mdiagonal} for any CRIN protocol $M'(z, H_1, U)$.

    Next, we prove that for any element $M'(z, H_{1}, U)$ of an CRIN protocol $\mathbbm{M}'$ and for all $i \neq j$,
    \begin{equation}
        \ba{l}
            \bra{\delta_{i}(H_{1},U)}M'(z,H_{1},U)\ket{\delta_{j}(H_{1},U)}=0=\\
            =\bra{\delta_{i}(H_{1},U)}M_{\text{\tiny OBS}}(z,H_{1},U)\ket{\delta_{j}(H_{1},U)}.
        \ea\label{Moffdiagonal1}
    \end{equation}
    To show this, first note that $M'(z, H_{1}, U)$ is Hermitian, non-negative, and can be diagonalized in a discrete basis $\{\ket{w_j}\}$ (for the case of a continuous basis, see Appendix B), such that $M'(z, H_{1}, U) = \sum_j w_j \ket{w_j}\bra{w_j}$, with $w_j \geq 0$. Since $\{\ket{\delta_{i}(H_1, U)}\}$ spans the Hilbert space $\mathcal{H}$, each $\ket{w_j}$ can be expressed as $\ket{w_j} = \sum_k \gamma_{jk} \ket{\delta_{k}(H_1, U)}$, where $\gamma_{jk} = \braket{\delta_{k}(H_1, U) | w_j}$. As a result, we can write:
    \begin{equation}
        \ba{l}
            M'(z, H_{1}, U) = \sum_{j,k} w_j |\gamma_{jk}|^2 \ket{\delta_{k}(H_1, U)}\bra{\delta_{k}(H_1, U)} + \\
            + \sum_{j, k \neq k'} w_j \gamma_{jk'}^{*} \gamma_{jk} \ket{\delta_{k}(H_1, U)}\bra{\delta_{k'}(H_1, U)}. \label{DiagMruim}
        \ea
    \end{equation}
    Two and only two cases arise: either there exists $\bar{k}$ such that $z = \delta_{\bar{k}}(H_1, U)$, or $z \neq \delta_{k}(H_1, U)$ for all $k$. 

    \emph{Case 1}: $z = \delta_{\bar{k}}(H_1, U)$. From Eq.~\eqref{Mdiagonal}, $\bra{\delta_{k}(H_1, U)} M'(z, H_{1}, U) \ket{\delta_{k}(H_1, U)} = 0$ for all $k \neq \bar{k}$. Eq.~\eqref{DiagMruim} thus implies that $\sum_j w_j |\gamma_{jk}|^2 = 0$ for all $k \neq \bar{k}$. Since $w_j \geq 0$, we must have $w_j \gamma_{jk} = 0$ for all $k \neq \bar{k}$. Consequently, all off-diagonal terms in Eq.~\eqref{DiagMruim} vanish, proving Eq.~\eqref{Moffdiagonal1} for all $i \neq j$.

    \emph{Case 2}: $z \neq \delta_{k}(H_1, U)$ for all $k$. Eq.~\eqref{Mdiagonal} implies $\bra{\delta_{k}(H_1, U)} M'(z, H_{1}, U) \ket{\delta_{k}(H_1, U)} = 0$ for all $k$. From Eq.~\eqref{DiagMruim}, this leads to $\sum_j w_j |\gamma_{jk}|^2 = 0$ for all $k$, and hence $w_j \gamma_{jk} = 0$. Thus, the off-diagonal terms vanish, and Eq.~\eqref{Moffdiagonal1} holds for all $i \neq j$.

    Combining Eq.~\eqref{Mdiagonal} with Eq.~\eqref{Moffdiagonal1}, we find that for all $i$ and $j$, $\bra{\delta_{i}(H_{1}, U)} M'(z, H_{1}, U) \ket{\delta_{j}(H_{1}, U)} = \bra{\delta_{i}(H_{1}, U)} M_{\text{\tiny OBS}}(z, H_{1}, U) \ket{\delta_{j}(H_{1}, U)}$. Thus, $M'(z, H_{1}, U) = M_{\text{\tiny OBS}}(z, H_{1}, U)$ for all $z$, $H_{1}$, and $U$. Therefore, any CRIN protocol $\mathbbm{M}'$ must coincide with the OBS protocol $\mathbbm{M}_{\text{\tiny OBS}}$, completing the proof.
\end{proof}

It is interesting to notice that this result is not confined to energy variations--it extends to  the variation of any quantum observable, such as linear and angular momentum, or particle number. Additionally, we demonstrate in the Appendix D that the result holds for explicitly time-dependent observables, with the CRIN conditions appropriately adapted for such cases.

Result 1 may offer valuable insights into the foundations of the measurement of quantum processes. Since we expect the CRIN conditions to hold, then \emph{the OBS protocol emerges as the only consistent framework for determining the fluctuations of VPQs}. Consequently, two-time observables must be considered to characterize variations in energy, charge, particle position, or other observables in quantum systems. Any deviation from the OBS protocol necessarily leads to the violation of at least one CRIN condition. For instance, as we illustrate in the next section, the TPM protocol fails to satisfy condition 1, resulting in the statistical violation of energy conservation. 

Interestingly, since the statistical framework for two-time observables mirrors that of conventional ``one-time'' quantum observables, the CRIN conditions naturally extend standard quantum phenomena to two-time observables. For instance, by assuming the CRIN conditions, entanglement, steering, or quantum superposition can be considered within the scope of two-time quantum observables (see Refs. \cite{Silva2021,Silva2023} for a discussion). Furthermore, we can deduce a two-time uncertainty relation for any initial state $\rho$ \cite{Silva2021}: 
\begin{equation} 
\sigma_{\text{\tiny $\Delta (H_j,U)$}}(\sigma_{\text{\tiny $U^\dagger H_j U$}} + \sigma_{\text{\tiny $H_j$}}) \geq |\braket{[U^{\dagger}H_jU, H_j]}|, \label{Heisenberg} 
\end{equation}
where $j\in\{1,2\}$, $\sigma_{\text{\tiny $O$}} = \sqrt{\braket{O^2} - \braket{O}^2}$ is the variance of an observable $O$ and $\braket{O} = \Tr[O\rho]$ its expectation value. Remarkably, this leads to scenarios where the variation  $\Delta(H_j, U)$ can become perfectly defined (i.e., $\sigma_{\text{\tiny $\Delta(H_j,U)$}} \to 0$), even when the individual energies $H_j$ at $0$ and $U^\dagger H_j U$ at $t$ cannot be completely determinate (i.e., $\sigma_{\text{\tiny $U^\dagger H_j U$}} + \sigma_{\text{\tiny $H_j$}} \to \infty$). We illustrate this feature in the trapped ion example in the next section.

As shown in Refs. \cite{Silva2021,Silva2023,Silva2024} and in the example of the next section, the OBS protocol is derived by measuring an observable commuting with $\Delta(H_1,U)$, allowing us to explicitly determine $\wp_{\text{\tiny OBS}}(z, H_1, U, \rho)$. However, in many experimental scenarios, direct measurements of certain quantities are impractical, and instead, indirect measurement schemes employing auxiliary probes are used \cite{Mohammady2019,Batalhao2014}. A natural question arises: can the same $\wp_{\text{\tiny OBS}}(z, H_1, U, \rho)$ be obtained through the use of a probe to measure the energy? The second key result addresses this challenge:

\begin{result} For any $U$, $H_1$ acting on $\mathcal{H}$, and eigenstate $\ket{\delta_i(H_1, U)}$ of $\Delta(H_1, U)$, there exists a unitary $U'$ acting on an auxiliary Hilbert space $\mathcal{H}'$, an additional Hamiltonian $H_2$ acting on $\mathcal{H} \otimes \mathcal{H}'$, and a vector $\ket{v} \in \mathcal{H}'$ such that: 
\begin{eqnarray} 
&[H_1 \otimes \mathbbm{1}_{\mathcal{H}'} + H_2, U \otimes U'] = 0, \label{commutresult10} \\
&H_2 \ket{\delta_i(H_1, U), v} = E_i \ket{\delta_i(H_1, U), v},\label{result122ndrelation0}\\
&(U^\dagger \otimes U^{'\dagger}) H_2 (U \otimes U') \ket{\delta_i(H_1, U), v} = E_i' \ket{\delta_i(H_1, U), v}, \label{result12ndrelation0} 
\end{eqnarray} 
where $\ket{\delta_i(H_1, U), v} = \ket{\delta_i(H_1, U)} \otimes \ket{v}$, and $E_i$ and $E_i' = E_i - \delta_i(H_1, U)$ are real numbers. \end{result}

Since the full derivation is somewhat lengthy, we present here only the main steps of the proof and defer the technical details to Appendix C.

\begin{proof}
 Since $U$ is unitary and diagonalizable by a countable eigenbasis $\{\ket{u_i}\}$, we assume that for every basis element $\ket{u_i}$, the relations $U\ket{u_i} = u_i\ket{u_i}$ and $u_i = \mathrm{e}^{i\theta_i}$ hold, where $\theta_i \in \mathbbm{R}$. The set of all $\theta_i$ is referred to as $\Theta = \{\theta_i\}$. Based on Lemma 1 in the Appendix C, we show that for any such $U$ and associated $\Theta$, we can define a $U'$ acting on a Hilbert space $\mathcal{H}'$ with countable basis $\{\ket{u_i'}\}$ satisfying the following condition: \textbf{(U)} \emph{Each basis element $\ket{u_i'}$ satisfies $U'\ket{u_i'} = u_i'\ket{u_i'}$, where $u_i' = \mathrm{e}^{i\theta_i'}$ and $\theta_i'$ is real. The set of all $\theta_i'$ associated with $\{\ket{u_i'}\}$ is denoted by $\Theta' = \{\theta_i'\}$ and is countably infinite set. Also, for any indexes $m$, $k$, and $j$, if $\theta_j' \in \Theta'$ and $\theta_m, \theta_k \in \Theta$, there exists infinitely many $\theta_l' \in \Theta'$ such that $\theta_m + \theta_j' = \theta_k + \theta_l'$ mod $2\pi$.}

Considering $U'$ that satisfy \textbf{(U)}, we are interested in constructing $H_2$ that satisfy Eq. \eqref{commutresult10}, whose components with respect to the basis $\ket{u_i,u_k'}$ can be written as
    \begin{equation}
        \ba{l}
        \bra{u_m,u_j}[H_1\otimes\mathbbm{1}_{\mathcal{H}'}+H_2,U\otimes U']\ket{u_k,u_l}=\\
        =(\mathrm{e}^{i(\theta_k+\theta_l')}-\mathrm{e}^{i(\theta_m+\theta_j')})((H_1)_{\text{\scriptsize $mk$}}\delta^{\text{\scriptsize $jl$}}+(H_2)_{\text{\scriptsize $mk$}}^{\text{\scriptsize $jl$}})=0.\label{resultadoparcial1}
        \ea
    \end{equation}
We considered the notation $\bra{u_m}H_1\ket{u_k}=(H_1)_{\text{\scriptsize $mk$}}$, $\bra{u_m,u_j'}H_2\ket{u_k,u_l'}=(H_2)_{\text{\scriptsize $mk$}}^{\text{\scriptsize $jl$}}$ and $\delta^{\text{\scriptsize $jl$}}= \bra{u_j'}\mathbbm{1}_{\mathcal{H}'}\ket{u_l'}$. The subindexes and superindexes refer, respectively, to the components of the basis of $U$ and $U'$.  As a result of property \textbf{(U)}, it follows that for any $m$, $k$ and $j$, there exist at least one $l$ such that $(\mathrm{e}^{i(\theta_k+\theta_l')}-\mathrm{e}^{i(\theta_m+\theta_j')})=0$. For these cases, we can consider $(H_2)_{mk}^{jl}$ as any arbitrary value and still the relation Eq. \eqref{resultadoparcial1} holds. Therefore, we can define free complex variables $h_{mk}^{jl}$ and assume the components of $H_2$ to be 
    \begin{equation}
        \ba{lll}
        (H_2)_{mk}^{jl}=h_{mk}^{jl}&\text{if}&(\mathrm{e}^{i(\theta_k+\theta_l')}-\mathrm{e}^{i(\theta_m+\theta_j')})=0\\
        (H_2)_{mk}^{jl}=-(H_1)_{mk}\delta^{jl}&\text{if}&(\mathrm{e}^{i(\theta_k+\theta_l')}-\mathrm{e}^{i(\theta_m+\theta_j')})\neq 0
        \ea
        \label{componentsHY} 
    \end{equation}
    for either the indexes $k=m$ and $l \geq j$, or the indexes $k>m$ and any $l$. For the other indexes, we define
    \begin{equation}
    (H_2)_{mk}^{jl}=[(H_2)_{km}^{lj}]^*.
        \label{conditionsindex4}
    \end{equation} 
    By taking into account that $H_1$ is hermitian, these two equations guarantee that $H_2$ is hermitian and satisfies Eq.\eqref{commutresult10} (see Lemma 2 of the Appendix C for more details). Therefore, we assume from this point on that $H_2$ is  defined by Eqs. \eqref{componentsHY} and \eqref{conditionsindex4} and all we need to do to prove the rest of the result is to characterize $h_{mk}^{jl}$ such that Eqs. \eqref{result122ndrelation0} and \eqref{result12ndrelation0} are satisfied. 

We begin by considering Eq.~\eqref{result122ndrelation0} and then show that $H_2$, satisfying Eqs.~\eqref{result122ndrelation0} and \eqref{commutresult10}, also satisfies Eq.~\eqref{result12ndrelation0}. For this, let $\ket{v} \in \mathcal{H}'$ be a state whose components $\beta_j = \braket{u_j'|v}$ satisfy $\Re(\beta_j) \neq 0$ and $\Im(\beta_j) \neq 0$, i.e. non-vanishing real and imaginary parts. As shown in Lemma 3 of the Appendix C, it is always possible to define such a $\ket{v}$. We analyze how to define the free variables $h_{mk}^{jl}$ in Eq.~\eqref{componentsHY} so that Eq.~\eqref{result122ndrelation0} is satisfied. For this, consider the components of Eq.~\eqref{result122ndrelation0} in the basis $\ket{u_m, u_j'}$:
\begin{equation}
    E_{i}\alpha_{mi}\beta_j = \sum_{k,l}(H_2)_{mk}^{jl}\alpha_{ki}\beta_l, \label{eigenvalueeq2}
\end{equation}
where $\alpha_{mi} = \braket{u_m|\delta_i(H_1,U)}$. Defining the set $\mathbbm{N}_{mj}$ of all pairs $\{k,l\}$ satisfying $(\mathrm{e}^{i(\theta_k+\theta_l')} - \mathrm{e}^{i(\theta_m+\theta_j')}) = 0$, we can rewrite Eq.~\eqref{eigenvalueeq2}, using Eqs.~\eqref{componentsHY} and \eqref{conditionsindex4}, as:
\begin{equation}
    E_{i}\alpha_{mi}\beta_j = -\sum_{\{k,l\} \notin \mathbbm{N}_{mj}}(H_1)_{mk}\delta^{jl}\alpha_{ki}\beta_l + \sum_{\{k,l\} \in \mathbbm{N}_{mj}}h_{mk}^{jl}\alpha_{ki}\beta_l. \label{eigenvalueeq3}
\end{equation}
By property \textbf{(U)}, there always exists at least one pair $\{k,l\}$ where $\alpha_{ki} \neq 0$ and $\mathbbm{N}_{mj}$ is non-empty. Thus, for each $m$ and $j$, at least one free variable $h_{mk}^{jl}$ is available. In Appendix C we demonstrate that $h_{mk}^{jl}$ can always be tuned to ensure Eq.~\eqref{eigenvalueeq3} holds for any $\alpha_{mi}$, $m$, $j$, $E_i$, and $(H_1)_{mk}$, proving Eq.~\eqref{result122ndrelation0} for any $U$, $H_1$, and $\ket{\delta_{i}(H_1,U)}$.

To complete the proof, we show that the same $H_2$ and $\ket{v} \in \mathcal{H}'$ satisfying Eqs.~\eqref{commutresult10} and \eqref{result122ndrelation0} also satisfy Eq.~\eqref{result12ndrelation0}. Applying $U^{\dagger} \otimes U^{'\dagger}$ to both sides of Eq.~\eqref{commutresult10}, we obtain:
\begin{equation}
    \Delta(H_1, U) \otimes \mathbbm{1}_{\mathcal{H}'} = -\Delta(H_2, U \otimes U'). \label{deltasiguais}
\end{equation}
Since $\Delta(H_1, U) \otimes \mathbbm{1}_{\mathcal{H}'}\ket{\delta_i(H_1,U),v} = \delta_i(H_1,U)\ket{\delta_i(H_1,U),v}$, Eq.~\eqref{deltasiguais} implies: $[(U^{\dagger} \otimes U^{'\dagger}) H_2 (U \otimes U') - H_2]\ket{\delta_i(H_1,U),v} = -\delta_i(H_1,U)\ket{\delta_i(H_1,U),v}$. Using the fact that $H_2$ satisfies Eq.~\eqref{result122ndrelation0} and the relation $(U^{\dagger} \otimes U^{'\dagger}) H_2 (U \otimes U') = H_2 + \Delta(H_2, U \otimes U')$, we deduce: $(U^{\dagger} \otimes U^{'\dagger}) H_2 (U \otimes U') \ket{\delta_i(H_1,U),v} = (E_i - \delta_i(H_1,U))\ket{\delta_i(H_1,U),v}$, completing the proof.
\end{proof}

Result 2 reveals a significant insight into the interplay between the OBS protocol and the CRIN conditions. When employing directly the OBS protocol for an energy operator $H_1$ without considering result 2, the measurement process projects the system into an eigenstate $\ket{\delta_i(H_1, U)}$, with the variation of energy inferred as $\delta_i(H_1, U)$. However, the variation $\delta_i(H_1, U)$ is not directly measured. Instead, the protocol relies on measuring an auxiliary observable $O^\Delta$ that commutes with $\Delta(H_1, U)$~\cite{Silva2021,Silva2024,Silva2023}. Considering result 2, an alternative perspective is offered, enabling a way to indirectly obtain the OBS statistics for $\Delta(H_1, U)$ by means of a probe, while adhering to the CRIN conditions. Specifically, Result 2 guarantees the existence of a probe related with an auxiliary Hamiltonian $H_2$ and dynamics $U'$ such that the total energy $H_1 \otimes \mathbbm{1}_{\mathcal{H}'} + H_2$ is conserved under $U \otimes U'$. Moreover, once the combined state $\ket{\delta_i(H_1, U), v}$ is prepared, it has well-defined values for $H_2$ at both times $0$ and $t$, enabling the variation of $H_2$ to be precisely determined as $-\delta_i(H_1, U)$. By considering CRIN conditions 2 and 3, one can infer that the variation of $H_2$ is $-\delta_i(H_1, U)$ with 100\% certainty (when disregarding practical experimental disturbances). Furthermore, conditions 1 and 4 ensure that this conclusion remains valid, as the conservation of $H_1 \otimes \mathbbm{1}_{\mathcal{H}'} + H_2$ enforces the same probability for $H_1$ to vary by $\delta_i(H_1, U)$. This result applies universally to any observable $O_1$ and $O_2$ by substituting $H_1,H_2 \to O_1,O_2$ in Result 2, therefore not restricting it to energy variations.

We emphasize that Result 2 is an existence result: it shows that, for every eigenstate $\ket{\delta_i(H_1,U)}$ of the variation operator, one can construct an extended energy-conserving dynamics in which this state is assigned definite initial and final values of the auxiliary Hamiltonian $H_2$. In general, the auxiliary Hamiltonian $H_2$, the unitary $U'$, and the vector $\ket{v}$ depend on the chosen eigenstate $\ket{\delta_i(H_1,U)}$. Thus, implementing the construction for all eigenstates may require different auxiliary settings. This does not weaken the statement of the result; rather, it clarifies its operational scope. In particular, the result establishes that each eigenvalue of $\Delta(H_1,U)$ can be embedded into an energy-conserving scenario with definite auxiliary energy changes. Moreover, in relevant cases this result can be used directly, as we show next.

\section{Trapped ion case study}

We consider a trapped ion system similar to that in Ref.~\cite{Lindenfels2019}. The total Hamiltonian is $H = H_{\text{\tiny HO}} + H_{\text{\tiny e}}$. Here, $H_{\text{\tiny HO}} = \frac{P^2}{2m}\otimes \mathbbm{1}_{\text{\tiny s}} + \frac{m\omega^2X^2}{2}\otimes \mathbbm{1}_{\text{\tiny s}} = \hbar\omega (N + 1/2)\otimes \mathbbm{1}_{\text{\tiny s}}$ represents the center-of-mass energy of a Ca$^+$ ion, and $H_{\text{\tiny e}} \approx \hbar (\omega_z/2 + \Delta_S k_{\text{\tiny SW}} X/2) \otimes \sigma_z$ describes the coupling between the ion's center-of-mass position and the spin of its covalent electron for small displacements, mediated by an optical dipole force (see Fig.~\ref{fig:corr}). The system evolves in isolation under $U_t = \exp(-itH/\hbar)$ until time $t$. Since $[H, \sigma_z] = [X, \sigma_z] = [P, \sigma_z] = 0$, then $U_t = U_t'\exp[-it(\hbar \omega_z\sigma_z-m\omega^2 a^2)/2\hbar]$, where we defined $U_t' = \exp(-itH'/\hbar)$, $H' = P^2/(2m) \otimes \mathbbm{1}_{\text{\tiny s}} + m\omega^2/2 \left(X \otimes \mathbbm{1}_{\text{\tiny s}} + a(\mathbbm{1}_{\text{\tiny CM}} \otimes \sigma_z)\right)^2$, and $a = (\hbar \Delta_S k_{\text{\tiny SW}}) / (2m\omega^2)$. Defining $X' = X \otimes \mathbbm{1}_{\text{\tiny s}} + a (\mathbbm{1}_{\text{\tiny CM}} \otimes \sigma_z)$, we find that $[X', P] = i\hbar\mathbbm{1}$ and $[X', X] = 0$. Given the Heisenberg evolution of any operator $O(t)=U_t^{\dagger} O U_t$ \cite{Sakurai1994}, it follows that $X'(t) = U_t^{'\dagger} X' U_t'$ and $P(t) = U_t^{'\dagger} (P \otimes \mathbbm{1}_{\text{\tiny s}}) U_t'$. Differentiating these expressions with respect to $t$ results in  $\partial_t X'(t) = P(t) / m$ and $\partial_t P(t) = -m\omega^2 X'(t)$. These solve to $X'(t) = X'(0)\cos(\omega t) + \frac{P(0)}{m\omega}\sin(\omega t)$ and $P(t) = -m\omega X'(0)\sin(\omega t) + P(0)\cos(\omega t)$, with initial conditions $X'(0) = X \otimes \mathbbm{1}_{\text{\tiny s}} + a (\mathbbm{1}_{\text{\tiny CM}} \otimes \sigma_z)$ and $P(0) = P \otimes \mathbbm{1}_{\text{\tiny s}}$. For the time $\tau = \pi / \omega$, we find: $X'(\tau) = -X \otimes \mathbbm{1}_{\text{\tiny s}} - a (\mathbbm{1}_{\text{\tiny CM}} \otimes \sigma_z)$, $P(\tau) = -P \otimes \mathbbm{1}_{\text{\tiny s}}$.
Considering these expressions, we obtain the Heisenberg operators:
\begin{eqnarray}
    &&H_{\text{\tiny HO}}(\tau)=H_{\text{\tiny HO}}+\Delta (H_{\text{\tiny HO}},U_\tau),\label{HhoandHetau}\\
    &&H_{\text{\tiny e}}(\tau)=H_{\text{\tiny e}}+\Delta (H_{\text{\tiny e}},U_\tau),\label{HhoandHetau2}\\
    &&\Delta(H_{\text{\tiny HO}}, U_{\tau}) = \hbar \Delta_S k_{\text{\tiny SW}} \left(X \otimes \sigma_z + a\right) = -\Delta(H_{\text{\tiny e}}, U_{\tau}). \label{Delta0tauexpressao}
\end{eqnarray}  
Therefore, 
\begin{equation}
    [\Delta (H_{\text{\tiny HO}},U_\tau),H_{\text{\tiny HO}}]=[H_{\text{\tiny HO}}(\tau),H_{\text{\tiny HO}}]=i2\hbar\omega^2 a  P\otimes \sigma_z,
\end{equation} 
and 
\begin{equation}
   \begin{array}{ll}
    [\Delta (H_{\text{\tiny e}},U_\tau),H_{\text{\tiny e}}]&=[H_{\text{\tiny e}}(\tau),H_{\text{\tiny e}}]=[\Delta (H_{\text{\tiny HO}},U_\tau),X\otimes\sigma_z]=\\&=-[\Delta (H_{\text{\tiny e}},U_\tau),X\otimes\sigma_z]=0.
\end{array} 
\end{equation}
Moreover, it follows that
\begin{equation}
    \Delta (H_{\text{\tiny HO}},U_\tau)\ket{x,\pm}=\delta_{x,\pm}\ket{\delta_{x,\pm}}=-\Delta (H_{\text{\tiny e}},U_\tau)\ket{x,\pm}\label{eigendeltamethods}
\end{equation}
where $\delta_{x,\pm}=\hbar \Delta_S k_{\text{\tiny SW}}(\pm x+a)$ and $\ket{\delta_{x,\pm}}=\ket{x,\pm}$. We are now ready to study the consequences of results 1 and 2.

\subsection{Conservation of energy}
\begin{figure*}
    \centering
    \includegraphics[width=0.8\linewidth]{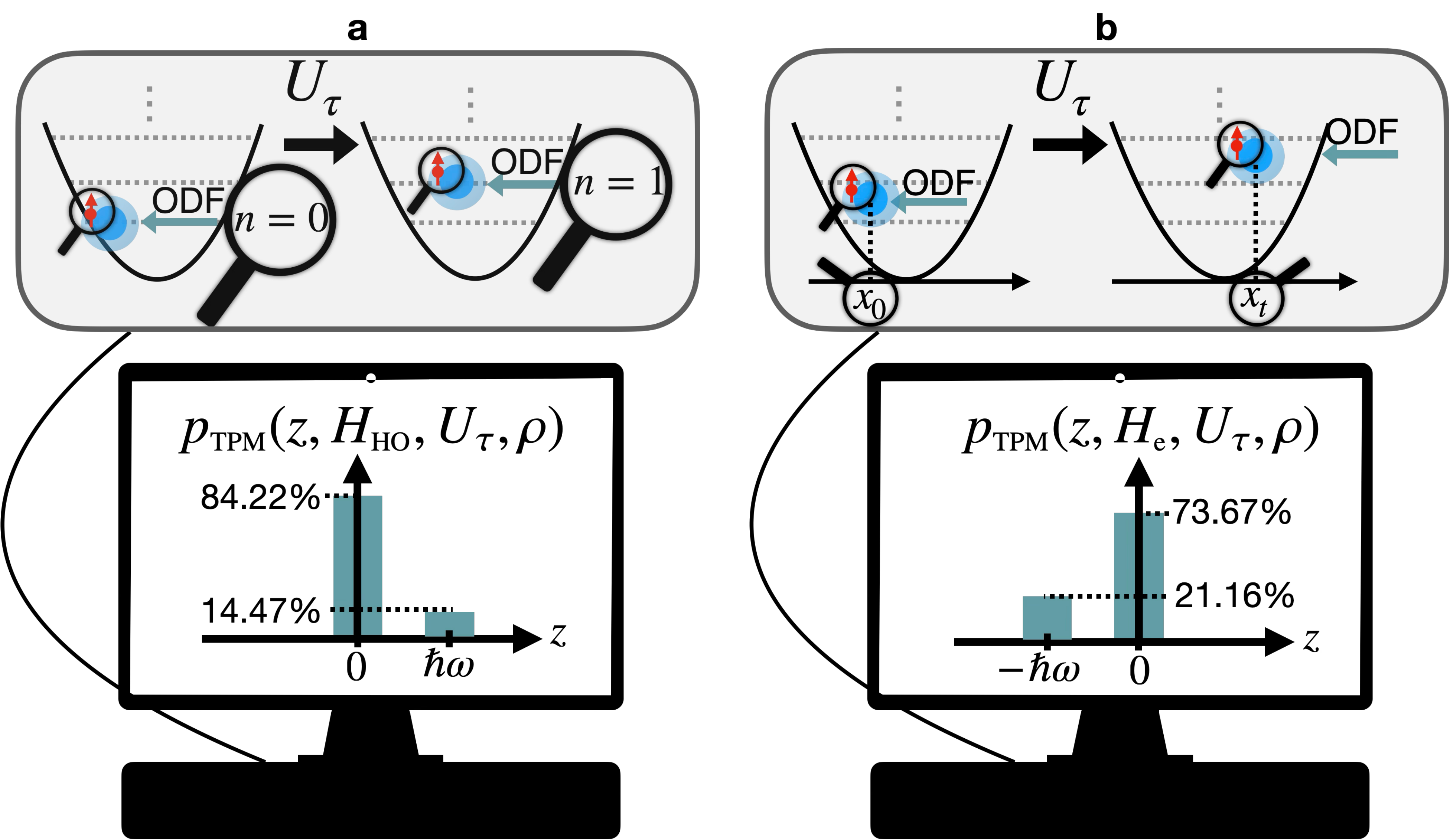}
    \caption{\textbf{Violation of energy conservation in TPM protocol}. The system $\Omega$ consists of a trapped Ca$^+$ ion model inspired in the Ref.~\cite{Lindenfels2019}. The total energy of the system is $H = H_{\text{\tiny HO}} + H_{\text{\tiny e}}$. $H_{\text{\tiny HO}} = \hbar \omega (N + 1/2)\otimes\mathbbm{1}_{\text{\tiny s}}$ is the energy related with the center of mass (CM). The position of the CM is coupled to the spin of a covalent electron (red arrow) through a spin-dependent optical dipole force (ODF), with a coupling energy approximately $H_{\text{\tiny e}} \approx \hbar (\omega_z/2 + \Delta_S k_{\text{\tiny SW}} X/2) \otimes \sigma_z$~\cite{Lindenfels2019}. One round of the TPM measurement of the variations of $H_{\text{\tiny HO}}$ and $H_{\text{\tiny e}}$ are represented in the above part of figures \textbf{a} and \textbf{b}, respectively. After collecting the results of many rounds of experiments, we represent in the lower graphs the two highest values of the probabilities $p_{\text{\tiny TPM}}(z, H_{i}, U_\tau, \rho)$ of the energies $H_i\in\{H_{\text{\tiny HO}},H_{\text{\tiny e}}\}$ of varying $z$ during the evolution $U_\tau$ for the initial state $\rho$. Energy conservation requires the increase of $H_{\text{\tiny HO}}$ to be equal to the decrease of $H_{\text{\tiny e}}$.  However, \textbf{a} and \textbf{b} reveal that TPM predicts $H_{\text{\tiny HO}}$ to be less likely to increase by $\hbar\omega$ than $H_{\text{\tiny e}}$ to decrease by $-\hbar\omega$, violating energy conservation. The calculations are done in the section ``Trapped ion case study'' and Appendix E, considering an initial state $\rho = \ket{0}\bra{0} \otimes \ket{+}\bra{+}$, with $N\ket{0} = 0\ket{0}$ and $\sigma_z\ket{+} = \ket{+}$, and parameters: $\tfrac{\omega}{1.4} = \tfrac{\Delta_S}{2.73}=\tfrac{\omega_z}{13}=2\pi$ MHz, $k_{\text{\tiny SW}} = 2\pi / (280 \,\text{nm})$, and $m = 6.68 \times 10^{-26}$ kg. The interval considered to compute the variation is $[0,\tau]$, where $\tau=\pi/\omega$.}
    \label{fig:corr}
\end{figure*}

To understand better the role of energy conservation in a protocol, let us compute the TPM statistics for the variation of $H_{\text{\tiny HO}} $ and $H_{\text{\tiny e}}$ for a preparation $\rho = \ket{0,+}\bra{0,+}$, where $N\ket{0} = 0\ket{0}$ and $\sigma_z\ket{+} = \ket{+}$. This is schematically represented in Fig.~\ref{fig:corr}. Let us first consider the TPM applied to the energy $H_{\text{\tiny e}}$ (Fig.~\ref{fig:corr}b). In this case, we can use a result similar to Ref.~\cite{Talkner2007} to show that if $[H_{\text{\tiny e}}(\tau), H_{\text{\tiny e}}] = 0$, $\wp_{\text{\tiny TPM}}(z, H_{\text{\tiny e}}, U_\tau, \rho) = \wp_{\text{\tiny OBS}}(z, H_{\text{\tiny e}}, U_\tau, \rho)$ (see Appendix E). Thus, it suffices $\wp_{\text{\tiny OBS}}(z, H_{\text{\tiny e}}, U_\tau, \rho)$ to describe $\wp_{\text{\tiny TPM}}(z, H_{\text{\tiny e}}, U_\tau, \rho)$. Using Eq.~\eqref{eigendeltamethods} and the OBS definition, we find:
\begin{equation}
    M_{\text{\tiny OBS}}(z, H_{\text{\tiny e}}, U_\tau) = \sum_{s\in\{+,-\}}\int_{-\infty}^{\infty}dx \delta^{\text{\tiny \textbf{D}}}[z - \delta_{x,s}(H_{\text{\tiny e}}, U_\tau)] P_{x,s},\label{metobspovmhe}
\end{equation}
where $P_{x,s} = \ket{x,s}\bra{x,s}$ and the sum $\sum_{s \in \{+,-\}}$ runs over the two possible values of $s = \pm$. Substituting $M_{\text{\tiny OBS}}(z, H_{\text{\tiny e}}, U_\tau)$ into Eq.~\eqref{obsdistrib} and simplifying using Eq.~\eqref{eigendeltamethods}, we find:
\begin{equation}
    \ba{rl}
    \wp_{\text{\tiny OBS}}(z, H_{\text{\tiny e}}, U_\tau,\rho)&=\frac{1}{\hbar \Delta_S k_{\text{\tiny SW}}}|\braket{-a-\frac{z}{\hbar \Delta_S k_{\text{\tiny SW}}}|0}|^2,\label{wdistmet1}
    \ea
\end{equation}
where $\braket{-a - \tfrac{z}{\hbar \Delta_S k_{\text{\tiny SW}}} | 0}$ corresponds to $\braket{x | 0}$, with $x \to -a - \tfrac{z}{\hbar \Delta_S k_{\text{\tiny SW}}}$. For the ground state $\ket{0}$, we have:
\begin{equation}
     |\braket{x|0}|^2=\mathcal{N}(x,0,\sigma^2), \label{gaussian0}
 \end{equation}
 where $\sigma\equiv\sqrt{\hbar/(2m\omega)}$ and, from now on, we consider
 \begin{equation}
     \mathcal{N}(y,y_c,\sigma_y^2)=\frac{1}{\left(2\pi \sigma_y^2\right)^{1/2}} \exp\left[-\frac{(y-y_c)^2}{2\sigma_y^2}\right]\label{Normaldef}
 \end{equation}
as a normal distribution of variable $y$, with mean $y_c$ and variance $\sigma_y^2$. Substituting Eq.~\eqref{gaussian0} into Eq.~\eqref{wdistmet1}, and noting that $\wp_{\text{\tiny TPM}}(z, H_{\text{\tiny e}}, U_\tau, \rho) = \wp_{\text{\tiny OBS}}(z, H_{\text{\tiny e}}, U_\tau, \rho)$, we find: $\wp_{\text{\tiny TPM}}(z, H_{\text{\tiny e}}, U_\tau, \rho) = \mathcal{N}(z, z_a, \sigma_z^2)$, where $z_a = -\hbar \Delta_S k_{\text{\tiny SW}} a$ and $\sigma_z = \hbar \Delta_S k_{\text{\tiny SW}} \sigma$. Finally, we compute the probability $p_{\text{\tiny TPM}}(-n\hbar\omega, H_{\text{\tiny e}}, U_\tau,\rho)\coloneqq \int_{I_n^-}\,dz\,\wp_{\text{\tiny TPM}}(z, H_{\text{\tiny e}}, U_\tau,\rho)$ of $H_{\text{\tiny e}}$ of varying by an amount $z$ in intervals $I_n^- = [-(n+1/2)\hbar\omega, -(n-1/2)\hbar\omega]$. Figure~\ref{fig:corr}b shows $p_{\text{\tiny TPM}}(-n\hbar\omega, H_{\text{\tiny e}}, U_\tau, \rho)$ for $n = 0$ and $n = 1$.

Next, we use eq. \eqref{eq:TPMprobgeneral} to find: $\wp_{\text{\tiny TPM}}(z,H_{\text{\tiny HO}}\otimes \mathbbm{1}_{\text{\tiny s}}, U_\tau,\rho)= \sum_{m,n=0}^{\infty}\sum_{r,s\in\{+,-\}} \delta^{\text{\tiny \textbf{D}}}[z - (e_{m,r} - e_{n,s})] |\bra{m,r} U_\tau \ket{n,s}|^2  \bra{n,s}\rho\ket{n,s}$, where $\sum_{r,s \in \{+,-\}}$ sums over $r, s \in \{+,-\}$, and $\ket{n,s}$ are eigenvectors of $H_{\text{\tiny HO}}$ satisfying $H_{\text{\tiny HO}} \ket{n,\pm} = e_{n,\pm} \ket{n,\pm} = \hbar\omega(n + 1/2) \ket{n,\pm}$. Since $[U_\tau, \sigma_z] = 0$ and $\bra{n,s} \rho \ket{n,s} = 1$ for $n = 0$ and $s = +$, and $0$ otherwise, it follows: $ \wp_{\text{\tiny TPM}}(z, H_{\text{\tiny HO}} , U_\tau, \rho) = \sum_{m=0}^{\infty} \delta^{\text{\tiny \textbf{D}}}[z - (e_{m,+} - e_{0,+})] |\bra{m} U_\tau^+ \ket{0}|^2$, where $U_\tau^{+}=\bra{+}U_\tau\ket{+}=\mathrm{e}^{-i\theta_\tau}\exp \left[-\frac{i\tau}{\hbar}\left(\frac{P^2}{2m}+\frac{m\omega^2}{2}X_+^{2}\right)\right]$, $X_+=\bra{+}X'\ket{+}=X+a$, and $\theta_\tau=\tfrac{\hbar\omega_zt}{2\hbar}-\tfrac{m\omega^2a^2t}{2\hbar}$.  

Analogously, $
p_{\text{\tiny TPM}}(n\hbar\omega, H_{\text{\tiny HO}} , U_\tau, \rho) \coloneqq \int_{I_n^+} dz \, \wp_{\text{\tiny TPM}}(z, H_{\text{\tiny HO}} , U_\tau, \rho) = |\bra{n} U_\tau^+ \ket{0}|^2$ is the probability of $H_{\text{\tiny HO}} $ varying in intervals $I_n^+ = [(n-1/2)\hbar\omega, (n+1/2)\hbar\omega]$. We derive $|\bra{n} U_\tau^+ \ket{0}|^2$ analytically, but leave its details to the Appendix E. In Fig.~\ref{fig:corr}a, we show $p_{\text{\tiny TPM}}(n\hbar\omega, H_{\text{\tiny HO}} , U_\tau, \rho)$ for $n = 0$ and $n = 1$.

It can be explicitly proven that the TPM protocol satisfies the reality, state-independence, and no-signaling conditions. Moreover, when $[H_{1},U^\dagger H_1 U]=[H_{2},U^\dagger H_2 U]=0$, it will also satisfy the conservation condition. However, the latter will be violated in general when energy operators $H_{1(2)}$ and $U^\dagger H_{1(2)}U$ do not commute. The result above shows exactly that feature: the TPM probability assigned to a given energy variation on one side of the process does not, in general, coincide with the probability assigned to the corresponding opposite variation on the other side. The same issue does not arise for the OBS protocol. Indeed, OBS is defined directly from the Heisenberg-picture variation operator, and therefore immediately satisfy
\begin{equation}
    \Delta(H_1,U)=U^{\dagger}H_1U-H_1=-(U^{\dagger}H_2U-H_2)=-\Delta(H_2,U).
\end{equation} As a consequence, whenever the total energy is conserved, the corresponding variation operators satisfy the appropriate operator identity, which in turn immediately guarantees condition 1.

\begin{figure}
    \centering
    \includegraphics[width=1\linewidth]{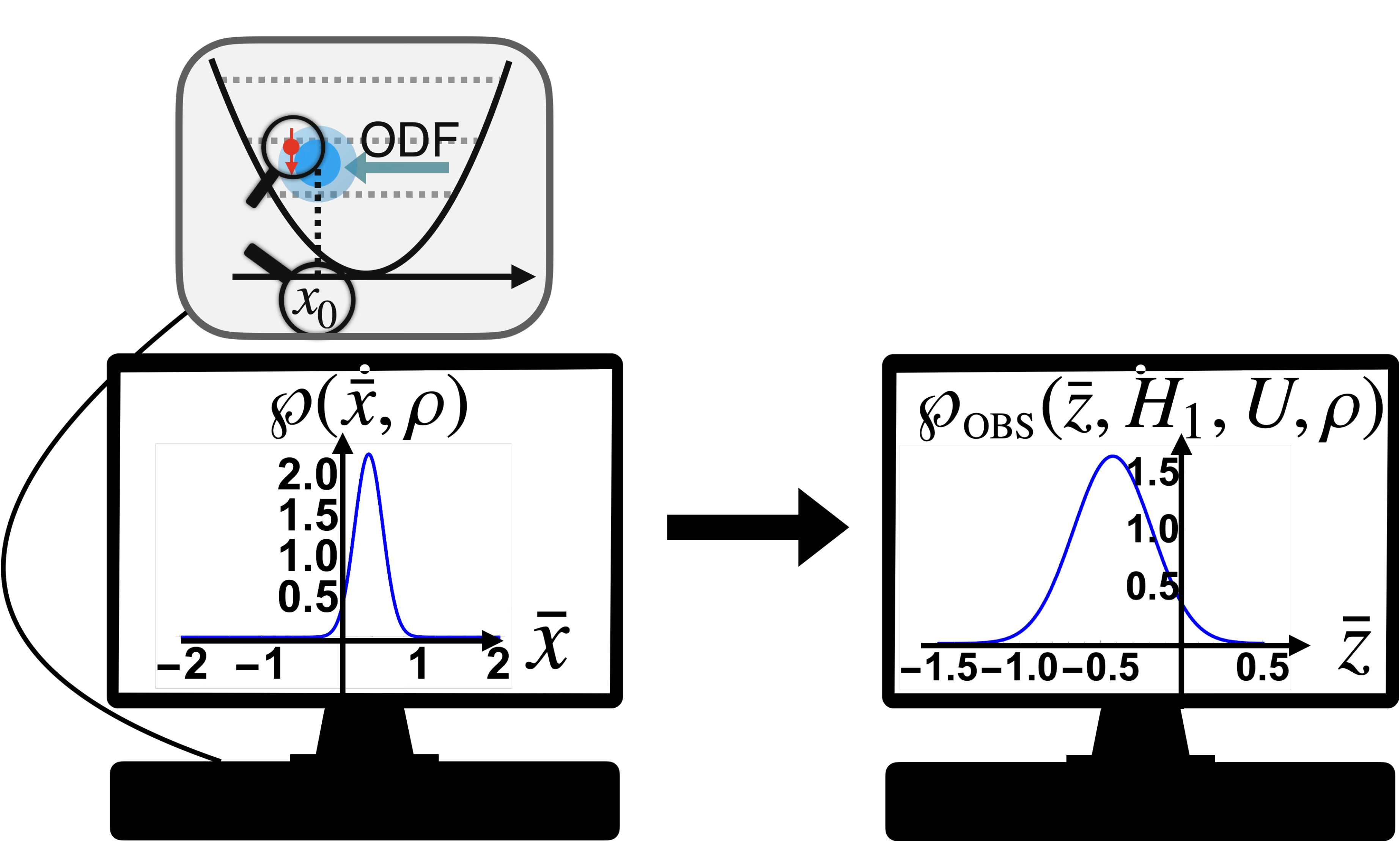}
    \caption{\textbf{Schematic implementation of the OBS protocol for the trapped Ca$^+$ ion system described in Fig. \ref{fig:corr}}. The variation of the center-of-mass energy, $H_1 = H_{\text{\tiny HO}} \otimes \mathbbm{1}_s=\hbar \omega (N + 1/2)\otimes \mathbbm{1}_s$, is inferred entirely from measurements of $X \otimes \sigma_z$, an observable that commutes with $\Delta(H_1, U)$, with $U=U_\tau$. This highlights how OBS protocol can be applied via commuting observables. We consider an initial preparation $\rho=\ket{\alpha,-}\bra{\alpha,-}$, where $\ket{\alpha}$ is the coherent state and $\sigma_z\ket{-}=-\ket{-}$. The left graph below the scheme shows the expected probability distribution $\wp(\bar{x},\rho)$ of measuring $X \otimes \sigma_z$  at time $0$ and finding $\bar{x} = k_{\text{\tiny SW}} x$ for the dimensionless position. Eq. \eqref{obsdistrib} together with $\wp(\bar{x},\rho)$ can be used to obtain the OBS probability distribution $\wp_{\text{\tiny OBS}}(\bar{z}, H_1, U, \rho)$ for the dimensionless variation $\bar{z} = z / \hbar\omega$ (see section ``Trapped ion case study'' and Appendix E). $\wp_{\text{\tiny OBS}}(\bar{z}, H_1, U, \rho)$ is shown in the right graph as a normalized Gaussian centered at $\braket{\Delta(H_1, U)}/(\hbar\omega)$, with variance $\sigma_{\text{\tiny $\Delta(H_1, U)$}}/\hbar\omega = \sqrt{\braket{\Delta^2(H_1, U)}-\braket{\Delta(H_1, U)}^2}/\hbar\omega$.}
    \label{fig:OBSEXAMPLE}
\end{figure}
\subsection{OBS protocol, non-commutative operators, and result 2}

To study in detail the consequences of result 2, let us consider an initial state $\rho = \ket{\alpha, -}\bra{\alpha, -}$, where 
\begin{equation}
    \ket{\alpha}=\mathrm{e}^{-|\alpha|^2/2}\sum_{n=0}^{\infty}\frac{\alpha^n}{\sqrt{n!}}\ket{n},\label{propalpha}
\end{equation}
is the coherent state~\cite{Cohen2020} and $\sigma_z \ket{-} = -\ket{-}$. We show in Fig.~\ref{fig:OBSEXAMPLE}, a graph of the OBS distribution for this state:
\begin{equation}
    \wp_{\text{\tiny OBS}}(z, H_{\text{\tiny HO}} , U_\tau,\rho)=\int_{-\infty}^{\infty}dx \delta^{\text{\tiny \textbf{D}}}[z - \delta_{x,-}] |\braket{x|\alpha}|^2.\label{wobsalphamet}
\end{equation}
Using the coherent state probability $|\braket{x|\alpha}|^2 = \mathcal{N}(x, \bra{\alpha} X \ket{\alpha}, \sigma_X^2)$, with $\bra{\alpha} X \ket{\alpha} = \sqrt{2\hbar/(m\omega)} \Re(\alpha)$ and $\sigma_X = \sqrt{\hbar/(2m\omega)}$ \cite{Cohen2020}, and introducing the dimensionless position $\bar{x} = k_{\text{\tiny SW}} x$, we get the probability $\wp(\bar{x},\rho)= \mathcal{N}(\bar{x}, k_{\text{\tiny SW}} \bra{\alpha} X \ket{\alpha}, k_{\text{\tiny SW}}^2 \sigma_X^2)$ of the center of mass of the ion to be found at dimensionless position $\bar{x}$ when initially prepared in the state $\rho$. The left-hand side of Fig.~\ref{fig:OBSEXAMPLE} shows the theoretical prediction of such a distribution. This could be obtained experimentally by measuring the center of mass of the ion, which can be achieved experimentally~\cite{Lindenfels2019}.

Now, given $|\braket{x|\alpha}|^2$, we can compute $\wp_{\text{\tiny OBS}}(\bar{z}, H_{\text{\tiny HO}} , U_\tau, \rho)$ with respect to the dimensionless variation $\bar{z} = z/(\hbar\omega)$, represented in the right-hand side of Fig.~\ref{fig:OBSEXAMPLE}. Substituting $|\braket{x|\alpha}|^2 = \mathcal{N}(x, \bra{\alpha} X \ket{\alpha}, \sigma_X^2)$ into Eq.~\eqref{wobsalphamet}, we find our theoretical prediction: $\wp_{\text{\tiny OBS}}(\bar{z}, H_{\text{\tiny HO}} , U_\tau,\rho)= \mathcal{N}(x,\bar{z}_\alpha,\sigma_{\bar{z}}^2)$, where $\bar{z}_\alpha = -\frac{2m\omega^2 a \bra{\alpha} X \ket{\alpha}}{\hbar} + \frac{2m\omega |a|^2}{\hbar}$ and $\sigma_{\bar{z}} = \frac{2m\omega a \sigma_X}{\hbar}$. The right-hand side of Fig.~\ref{fig:OBSEXAMPLE} displays $\wp_{\text{\tiny OBS}}(\bar{z}, H_{\text{\tiny HO}} , U_\tau, \rho)$, with parameters as described in the figure. 

We now compute the Heisenberg two-time uncertainty relations for $H_{\text{\tiny HO}} $ using Eqs.~\eqref{HhoandHetau} and \eqref{Delta0tauexpressao}. Using the coherent state properties from \cite{Cohen2020}, we arrive in the results: $\sigma_{\bar{z}}\coloneqq\sigma_{\Delta(H_{\text{HO}},U_\tau)}/(\hbar\omega)=\sqrt{ \hbar\Delta_S^2 k_{SW}^2/\omega^3}$, $\sigma_{H_{\text{HO}}(\pi/\omega)}/(\hbar\omega)=\sqrt{(\Re(\alpha)-\sigma_{\bar{z}})^2+\Im^2(\alpha)}$,  $\sigma_{H_{\text{HO}}}/(\hbar\omega)=|\alpha|$, and $|\braket{[H_{\text{\tiny HO}}(\pi/\omega),H_{\text{\tiny HO}}(0)]}|/(\hbar^2\omega^2)=2\sigma_{\bar{z}}\Im(\alpha)$. Notice that for sufficiently large $\omega$, $\Re(\alpha)$ and finite $\Im(\alpha)$, $\sigma_{\Delta(H_{\text{HO}},U_\tau)}/(\hbar\omega)\to 0$ while $\sigma_{H_{\text{HO}}(\pi/\omega)}/(\hbar\omega),\sigma_{H_{\text{HO}}}/(\hbar\omega) \to \infty$ and $|\braket{[H_{\text{\tiny HO}}(\pi/\omega),H_{\text{\tiny HO}}(0)]}|/(\hbar^2\omega^2)$ finite. This is the situation which we predicted in section \ref{sec:mainresults}. 

Finally, let us analyze how result 2 can be approximately seen for the coherent state. We depict in Fig.~\ref{fig:comparison}a the probabilities of measuring $H_{\text{\tiny HO}}(\pi/\omega)$ and $H_{\text{\tiny HO}}(0)$. At $t=0$, $H_{\text{\tiny HO}}(0)\ket{n,-} = \hbar\omega(n + 1/2)\ket{n,-}$, and the probability $p(n, H_{\text{\tiny HO}}(0))$ is: $p(n,H_{\text{\tiny HO}}(0))=|\braket{n,-|\alpha,-}|^2=\mathrm{e}^{-|\alpha|^2}\frac{|\alpha|^{2n}}{n!}$.
At $t = \tau$, we show in the Appendix E, $p(n, H_{\text{\tiny HO}}(\tau)) = \mathrm{e}^{-|\alpha'|^2} \frac{|\alpha'|^{2n}}{n!}$, where $\alpha' = [\Re(\alpha) - 2a\sqrt{m\omega/(2\hbar)}] + i\Im(\alpha)$. Fig.~\ref{fig:comparison}(a) shows $p(n, H_{\text{\tiny HO}}(0))$ (blue squares) and $p(n, H_{\text{\tiny HO}}(\tau))$ (red circles). We compare these probabilities with the distributions depicted in Fig.~\ref{fig:comparison}b $\wp(\bar{e}, H_{\text{\tiny e}}(t))$ of $H_{\text{\tiny e}}$ at $t=0$ and $t=\tau$. Using Eqs.~\eqref{HhoandHetau} and \eqref{Delta0tauexpressao}, we compute: $\wp(\bar{e}, H_{\text{\tiny e}}(0)) = \mathcal{N}(\bar{e}, (\hbar\omega)^{-1} \braket{H_{\text{\tiny e}}(0)}, (\hbar\omega)^{-2} \sigma_{H_{\text{\tiny e}}}^2)$ and $
\wp(\bar{e}, H_{\text{\tiny e}}(\tau)) = \mathcal{N}(\bar{e}, (\hbar\omega)^{-1} \braket{H_{\text{\tiny e}}(\tau)}, (\hbar\omega)^{-2} \sigma_{H_{\text{\tiny e}}}^2)$, where $\sigma_{H_{\text{\tiny e}}} = \hbar \Delta_S k_{\text{\tiny SW}} \sigma_X / 2$, $\braket{H_{\text{\tiny e}}(0)} = -(\hbar/2)(\omega_z - \Delta_S k_{\text{\tiny SW}} \bra{\alpha}X\ket{\alpha})$, and $\braket{H_{\text{\tiny e}}(\tau)} = (\hbar/2)(-\omega_z + \Delta_S k_{\text{\tiny SW}}(\bra{\alpha}X\ket{\alpha} - 2a))$. Fig.~\ref{fig:comparison}(b) shows $\wp(\bar{e}, H_{\text{\tiny e}}(0))$ (blue solid line) and $\wp(\bar{e}, H_{\text{\tiny e}}(\tau))$ (red dashed line). Notice that the absolute value of the energy variation is the same in both \textbf{a} and \textbf{b}, while the energy fluctuations in the former are significantly larger than in the later. This difference is consistent with Eq. \eqref{Heisenberg}: in \textbf{a}, the commutator $[U^\dagger H_1 U, H_1] \neq 0$ is non-zero, so that when the variation is approximately well defined, its energy at times $0$ and $\tau$ are not; on the other hand, in \textbf{b}, $[U^\dagger H_2 U, H_2] = 0$, so that is possible to have low fluctuation energies at the same times. Moreover, notice that the figures illustrate that energy variation of $H_1$ are more readily discernible when measuring $H_2$ at times $0$ and $\tau$, showing one example of Result 2 being applied.

\begin{figure}
    \centering
    \includegraphics[width=1\linewidth]{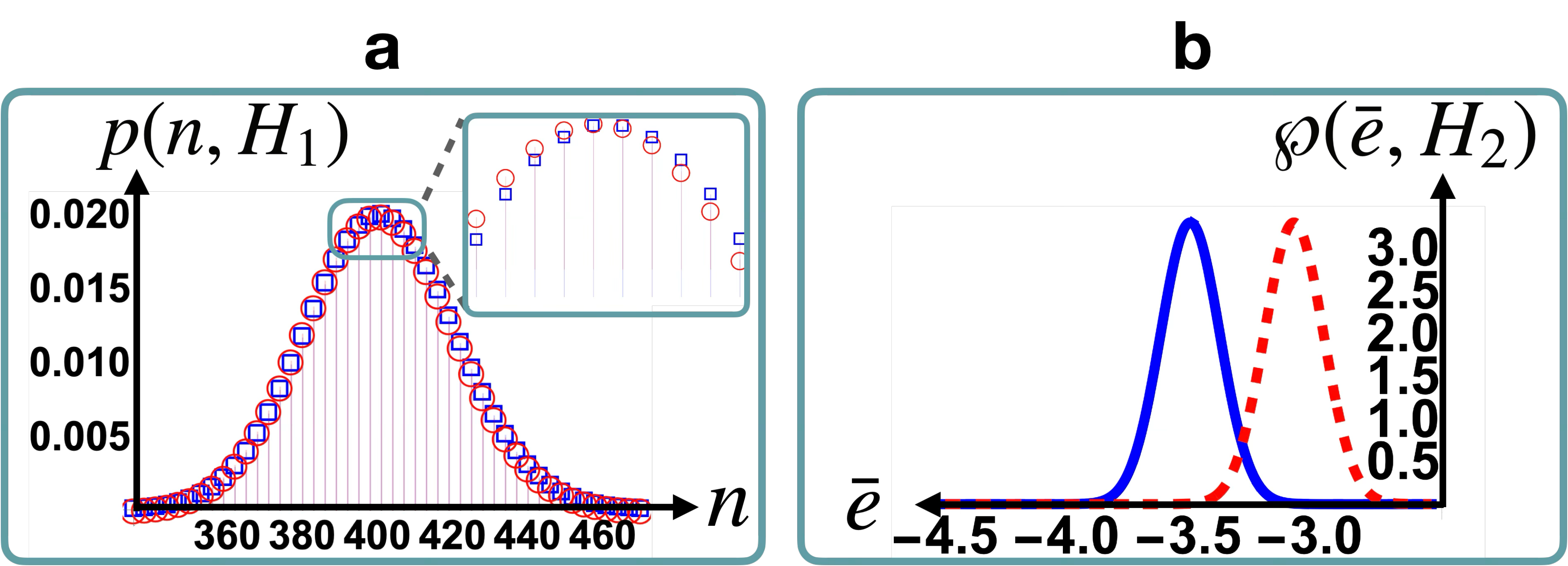}
    \caption{\textbf{The role of commutation in Eq. \eqref{Heisenberg} for the trapped Ca$^+$ ion system}. \textbf{(a)}  For $H_1=H_{\text{\tiny HO}}$, we plot the probabilities $\bra{n, -} \rho \ket{n, -}$ (blue squares) and $\bra{n, -} U \rho U^\dagger \ket{n, -}$ (red circles) of the energy $H_1$ being $E_n = (n + 1/2)\hbar\omega$ at $t = 0$ and $t = \tau$. The inset zooms into $395 \leq n \leq 405$, where small differences in the distributions cause small differences in the averages $\braket{U^\dagger H_1 U} - \braket{H_1} = \braket{\Delta(H_1, U)}$ correspond to the center of the normal distribution in Fig. \ref{fig:OBSEXAMPLE}. The variation of $H_1$ has relatively low fluctuations, with $\sigma_{\text{\tiny $\Delta(H_1, U)$}} \approx 0.243 \hbar\omega$. However, the uncertainties $\sigma_{\text{\tiny $H_1$}}\approx \sigma_{\text{\tiny $U^\dagger H_1 U$}} \approx |\Im(\alpha)|\hbar\omega = 20\hbar\omega$ are large, consistent with the lower bound in Eq. \eqref{Heisenberg}, $|\braket{[U^\dagger H_1 U,H_1]}| \approx 9.707 \hbar^2\omega^2$. \textbf{(b)} The probabilities distributions of finding $H_2= H - H_1 = H_{\text{\tiny e}}$ with the dimensionless value $\bar{e}=e/(\hbar\omega)$ at $t = 0$ (solid blue line) and $t = \tau$ (dashed red line) show low fluctuations. Here, $[U^\dagger H_2 U, H_2] = 0$, allowing for $\sigma_{\text{\tiny $H_2$}} =\sigma_{\text{\tiny $U^\dagger H_2 U$}} \approx 0.121\hbar\omega$ and $\sigma_{\text{\tiny $\Delta(H_2,U)$}} \approx \sigma_{\text{\tiny $\Delta(H_1,U)$}}\approx 0.243\hbar\omega$. Although the absolute value of the energy variation is the same in both panels, the energy fluctuations in \textbf{a} are significantly larger than in \textbf{b}. This difference is consistent with Eq. \eqref{Heisenberg}: in \textbf{a}, the commutator $[U^\dagger H_1 U, H_1] \neq 0$ is non-zero, while in \textbf{b}, $[U^\dagger H_2 U, H_2] = 0$. Notice that the figures illustrate that energy differences in $H_1$ are more readily discernible when measuring $H_2$ compared to $H_1$, supplementing Result 2.}
    \label{fig:comparison}
\end{figure}
\section{Final remarks}

To conclude, we have shown that the OBS protocol is unique in complying with the CRIN principles.
Given the fundamental nature of these principles, we expect them to be fulfilled by any POVM protocol aiming to measure VPQs. In this sense, results 1 and 2 provide strong support for considering the OBS protocol as a standard for measuring variation of quantum observables. 

This result strengthens and extends our previous findings in Ref.~\cite{Silva2024}. There, we showed that, among the \emph{established protocols} considered in the literature, only the OBS protocol simultaneously satisfies energy conservation and state independence while recovering the classical limit. Here, we go a significant step further: \emph{any alternative protocol} whose statistics differ from those of OBS must violate at least one of the CRIN conditions and is therefore physically inconsistent within this framework. Moreover, since OBS protocol satisfy the classical-limit criterion introduced in Ref.~\cite{Silva2024}, a immediate consequence of our results is that OBS is the only protocol compatible with both the CRIN conditions and the classical limit. Moreover, our results apply to the variation of arbitrary observables upon replacing the energy operators with the corresponding observables in the results. As shown in the Appendix D, the framework can also be extended to explicitly time-dependent observables, provided that the conservation law requirements are suitably adapted.

Looking forward, several promising directions emerge. For instance, our findings open intriguing directions for exploring the observable properties of two-time observables. Can Bell's inequalities applied to two-time observables reveal new insights for quantum cryptography protocols or uncover deeper aspects of quantum nonlocality? How does this framework connect with relativity principles? Moreover, could entanglement and other quantum correlations enhance the capabilities of quantum machines when applied to processes rather than instantaneous (one-time) observables? Since the OBS protocol is the only measurement scheme consistent with the CRIN conditions, these questions are not merely speculative, as they emerge naturally within a physically framework consistent with the CRIN conditions.

Another avenue is to apply the OBS protocol in quantum computing platforms \cite{Arute2019,Neill2017,Kreppel2023}, to gain insights into charge fluctuations and energy dissipation in quantum processors. Is it possible to approximate to eigenstates like $\ket{\delta_i(H_1, U)}$ with virtually no fluctuation in currents or energy variation in quantum devices? The OBS protocol framework also opens new possibilities for advancing measurement techniques in complex quantum systems. For instance, Result 2 demonstrated that attaching an auxiliary system governed by $H_2$ and $U'$ can enable indirect measurement of variations of observables within a specific system. This approach could be extended to quantum many-body systems \cite{lukin2019}, where adding few extra degrees of freedom and leveraging conservation laws could provide insights into highly nontrivial observables. Moreover, our findings invite deeper questions about conservation laws and measurement precision. For example, the WAY theorem \cite{Loveridge2011,Gisin2018} implies a fundamental tradeoff between energy conservation and the precise measurement of non-commuting observables. Could an auxiliary system be introduced, as in result 2, to restore energy conservation (by means of condition 1) while enabling precise measurements? 

It is important to remark that the framework presented in this work is grounded in a universal perspective, with an open quantum system (OQS) part of a larger, closed universe treated as a special case. Adapting the OBS protocol to account for scenarios involving noise and environmental interactions, perhaps using the adjoint Lindblad generator \cite{Breuer2002}, is a necessary step. Within this line of research, we are currently investigating the relation between the statistics of variation of observables and the lack of detailed balance at equilibrium (nonreciprocity) and persistent quantum currents \cite{Alicki2023}.

\section{Acknowledgments}
We thank Renato Angelo, Shay Blum, {\v C}aslav Brukner, Ohad Cremerman, Karen Hovhannisyan, Michael Iv, Gabriel Landi, Uri Peskin, Saar Rahav, and Ferdinand Schmidt-Kaler for fruitful conversations. T.A.B.P.S. was supported by the Helen Diller Quantum Center - Technion under Grant No. 86632417 and by the ISRAEL SCIENCE FOUNDATION (grant No. 2247/22).  D.G.K. is supported by the ISRAEL SCIENCE FOUNDATION (grant No. 2247/22) and by the Council for Higher Education Support Program for Hiring Outstanding Faculty Members in Quantum Science and Technology in Research Universities.

\section{Contributions}

Both authors contributed to the theoretical conceptualization, deductions, and writing of the manuscript.

\section{Competing interests}

The authors declare no competing interests.


\clearpage
\onecolumngrid

\section{APPENDIX A: Preliminaries}

Before presenting the main results and derivations of the paper in detail, we introduce key formalities and preliminary results that will be essential throughout the Appendices. We primarily consider the same systems discussed in the main text. Specifically, we analyze a general quantum system, $\Omega$, composed of arbitrary subsystems, under minimal assumptions: $\Omega$ is initially prepared in an arbitrary quantum state $\rho$ in a Hilbert space $\mathcal{H}$ and evolves from time $0$ to $t$ under an arbitrary unitary operator $U$, which has a countable eigenbasis. Our focus is on the measurement of the variation of \emph{a part} of the total energy of $\Omega$, represented by a time-independent Hermitian operator $H_1$ acting on $\mathcal{H}$.

\subsection{General two-time observables}
In the ``Main results'' section of the main text, we presented Result 1 under the assumption that the variation observable $\Delta(H_1,U)$ is diagonalizable in a discrete, countable and non-degenerate basis. Here, in the Appendices, we extend our analysis to cases where $\Delta(H_1,U)$ may exhibit degeneracy and/or possess a continuous spectrum. In this subsection, we outline key details related to these more general scenarios.
 
\subsubsection{Degenerate observables}

We first consider the case where $\Delta(H_1,U)$ is degenerate but still diagonalizable in a discrete, countable basis. For this, we assume the existence of a basis $\ket{\delta_i^{j}(H_1,U)}$ such that 
\begin{equation} \Delta(H_1,U)\ket{\delta_i^{j}(H_1,U)}=\delta_i(H_1,U)\ket{\delta_i^{j}(H_1,U)} \end{equation}
where $j\in{1,2,3,\dots,g(i)}$, meaning that each eigenvalue $\delta_i(H_1,U)$ has a degeneracy $g(i)$. This implies that the basis vectors satisfy the orthonormality condition
\begin{equation} \braket{\delta_{i'}^{j'}(H_1,U)|\delta_i^{j}(H_1,U)}=\delta^{\text{\tiny \textbf{K}}}_{ii'}\delta^{\text{\tiny \textbf{K}}}_{jj'},
\end{equation}
where $\delta^{\text{\tiny \textbf{K}}}_{ii'}$ and $\delta^{\text{\tiny \textbf{K}}}_{jj'}$ are Kronecker's delta, and that the operator
\begin{equation}
    P_i(H_1,U)=\sum_{j}\ket{\delta_i^{j}(H_1,U)}\bra{\delta_i^{j}(H_1,U)}\label{projectorbasis1}
\end{equation}
is a projector onto the subspace $\mathcal{E}_i$ associated with the eigenvalue $\delta_i(H_1,U)$~\cite{Cohen2020}. We can then express $\Delta(H_1,U)$ in terms of these projectors as
\begin{equation}
    \Delta (H_1,U)=\sum_{i}\delta_i(H_1,U) P_{i}(H_1,U).\label{projectorbasis12}
\end{equation}
Using this decomposition, we define the OBS protocol POVM $\mathbbm{M}_{\text{\tiny OBS}}(H_1,U)$ for a given energy operator $H_1$ and evolution $U$ in terms of the operators
\begin{equation}
    M_{\text{\tiny OBS}}(z,H_1,U)=\sum_{i}\delta^{\text{\tiny \textbf{D}}}[z-\delta_i(H_1,U)] P_{i}(H_1,U).
        \label{OBSDEGENERENCES}
\end{equation}
where $\delta^{\text{\tiny \textbf{D}}}$ is the Dirac's delta.
Notice that the case where $\Delta (H_1,U)$ has no degeneracies is simply a special case where $g(i)=1$ and $P_{i}(H_1,U)=\ket{\delta_i(H_1,U)}\bra{\delta_i(H_1,U)}$ for all $i$.
\subsubsection{Continuous observables}
The next step is to describe cases where $\Delta(H_1,U)$ has a continuous spectrum. Although the unitary operator $U$ remains diagonalizable in a discrete, countable basis, it is possible that not only $\Delta(H_1,U)$ but also the energy operators $H_1$ or $U^\dagger H_1 U$ exhibit continuous spectra.

When $\Delta(H_1,U)$ has a continuous spectrum, we introduce the projectors $\ket{w(H_1,U),y}\bra{w(H_1,U),y}dw\,dy$, which satisfy 
\begin{equation} 
\Delta (H_1,U)\ket{w(H_1,U),y}\bra{w(H_1,U),y}dw\,dy=w(H_1,U)\ket{w(H_1,U),y}\bra{w(H_1,U),y}dw\,dy, 
\end{equation}
allowing us to express $\Delta(H_1,U)$ in the spectral decomposition
\begin{equation}
    \Delta (H_1,U)=\int_{-\infty}^{\infty}dw w(H_1,U)\, P_{w}(H_1,U)\label{continuousoperatorspectral}
\end{equation}
where the projectors onto the eigenspaces associated with $w(H_1,U)$ are given by 
\begin{equation}
    P_{w}(H_1,U)=\int_{-\infty}^{\infty}dy\ket{w(H_1,U),y}\bra{w(H_1,U),y}.\label{continuousprojector}
\end{equation}
Here, we assume a general case where each eigenvalue $w(H_1,U)$ has an uncountable degeneracy. However, degeneracy may also be discrete, in which case
\begin{equation} P_{w}(H_1,U)=\sum_{y=1}^{g(w)}\ket{w(H_1,U),y}\bra{w(H_1,U),y}, \end{equation}
where $g(w)$ denotes the degeneracy of each eigenvalue $w(H_1,U)$. If there is no degeneracy, the projectors reduce to
\begin{equation} P_{w}(H_1,U)=\ket{w(H_1,U)}\bra{w(H_1,U)}. \end{equation}
With this continuous formulation of $\Delta(H_1,U)$, we define the OBS POVM for $H_1$ and $U$ as
\begin{equation}
    M_{\text{\tiny OBS}}(z,H_1,U)=\int_{-\infty}^{\infty}dw\,\delta^{\text{\tiny \textbf{D}}}[z-w(H_1,U)]P_{w}(H_1,U)=P_{z}(H_1,U)=\int_{-\infty}^{\infty}dy \ket{z,y}\bra{z,y},\label{obscontinuouspovm}
\end{equation}
where $\ket{z,y}$ is an eigenvector of $\Delta(H_1,U)$ with eigenvalue $z$, satisfying $\Delta(H_1,U)\ket{z,y}=z\ket{z,y}$.

To extend our analysis to cases where $H_1$ or $U^\dagger H_1 U$ have continuous spectra, we need to adapt the CRIN conditions accordingly. While Conditions 1, 3, and 4 remain well-defined for both discrete and continuous spectra, Condition 2 (the Reality Condition) is not generally well-posed in the continuous case.

To illustrate this issue, consider an energy operator $H_1$ that is diagonalized in a continuous and unbounded basis, such as $H_{\text{\tiny e}}$ in the main text. In such a scenario, there may exist an eigenvector $\ket{e}$ satisfying

\begin{equation}
    H_1\ket{e}=e\ket{e} \quad \text{and}\quad  U^\dagger H_1 U\ket{e}=\epsilon_e\ket{e}\label{conditionenergyttimes} 
\end{equation}
but $\ket{e}$ is not \emph{normalizable}. In this case, Condition 2 becomes ill-defined, as no \emph{normalized} state $\rho$ can simultaneously assign a well-defined energy value at both times. This raises the question: how should Condition 2 be adapted to accommodate such continuous scenarios?

To address this question, we first analyze the role of the Reality Condition in shaping the operator $M(z,H_1,U)$ in the discrete case. In this setting, if a normalized state $ \ket{e} $ is an eigenvector of both $H_1$ and $U^\dagger H_1 U$ with respective eigenvalues $e$ and $\epsilon_e$, then it must also be an eigenstate $\ket{\delta_j^k(H_1,U)}$ of $\Delta(H_1,U)$ with eigenvalue $\delta_j(H_1,U) = \epsilon_e - e$. Consequently, the Reality Condition imposes the constraint:
\begin{equation}
    \bra{\delta_j^k(H_1,U)}M(z,H_1,U)\ket{\delta_j^k(H_1,U)}=\delta^{\text{\tiny \textbf{D}}}[z-\delta_j^k(H_1,U)].
\end{equation}
This implies that any $M(z,H_1,U)$ satisfying the Reality Condition must take the form:
\begin{equation}
    M(z,H_1,U)=\delta^{\text{\tiny \textbf{D}}}[z-\delta_j^k(H_1,U)]\ket{\delta_j^k(H_1,U)}\bra{\delta_j^k(H_1,U)}+M_n(z)\label{trialdiscrete}
\end{equation}
where $M_n(z)$ is a Hermitian operator that satisfies $\bra{\delta_j^k(H_1,U)}M_n(z)\ket{\delta_j^k(H_1,U)}=0$. 

Now, let us examine the continuous case. As in the discrete scenario, if Eq. \eqref{conditionenergyttimes} holds, then $\ket{e}$ must also be an eigenstate of $\Delta(H_1,U)$. Given that $\Delta(H_1,U)$ has a continuous and potentially degenerate spectrum, as described in Eq. \eqref{continuousoperatorspectral}, the state $\ket{e}$ must take the form $\ket{w_e, y_e}$, where $w_e = \epsilon_e - e$ represents the eigenvalue of $\Delta(H_1,U)$, and $y_e$ is a real parameter accounting for the degeneracy.

We cannot simply assume for the continuous case that $\bra{w_e,y_e}M(z,H_1,U)\ket{w_e,y_e}=\delta^{\text{\tiny \textbf{D}}}(z-w_e)$, since $\ket{w_e,y_e}$ is not normalized. However, we can adapt the reasoning from Eq. \eqref{trialdiscrete}. Intuitively, given that $\ket{w_e,y_e}$ is an eigenstate of both $H_1$ and $U^\dagger H_1 U$, we can assume:
\begin{equation}
    M(z,H_1,U)=\delta^{\text{\tiny \textbf{D}}}(z-w_e)\ket{w_e,y_e}\bra{w_e,y_e}dwdy+M_{n}'(z)=\ket{z,y_e}\bra{z,y_e}dwdy+M_{n}'(z) \,\,\text{for $z=w_e$}\label{Mtrial}
\end{equation}
where $dw'dy'\bra{w_e,y_e}M_{n}'(z)\ket{w_e,y_e}dwdy=0$. This means that $M(z,H_1,U)$ has its diagonal element $\ket{w_e,y_e}\bra{w_e,y_e}dwdy$ defined as $\delta^{\text{\tiny \textbf{D}}}(z-w_e)$, while the remaining terms, represented by $M_{n}'(z)$, are arbitrary but must have a vanishing diagonal component.

Since the differentials $dw$ and $dy$ are not rigorously defined, expressing the reality condition in the form of Eq. \eqref{Mtrial} lacks precision. To formalize this further, we assume that there exist two open intervals $I_{w_e}$ and $I_{y_e}$ with bounded lengths $d_{w_e}=\int_{I_{w_e}}dw$ and $d_{y_e}=\int_{I_{y_e}}dy$, such that:
\begin{equation}
    M(z,H_1,U)=\int_{I_{y_e}}dy \ket{z,y}\bra{z,y}+M_{nd}(z), \quad \text{for $z=w_e$}\label{Mrealityofficial}
\end{equation}
where $M_{nd}(z)$ is Hermitian and satisfies:
\begin{equation}
    \bra{w(H_1,U),y}M_{nd}(z)\ket{w'(H_1,U),y'}=0,\quad \text{for all $w,w',y,y'$ such that $w,w'\in I_{w_e}$ and $y,y'\in I_{y_e}$.}\label{Mrealityofficial2}
\end{equation}
By defining $M(z,H_1,U)$ as in Eqs. \eqref{Mrealityofficial} and \eqref{Mrealityofficial2}, we ensure that the reality condition is well-defined for the continuous case. To see this, consider an arbitrary normalized state:
\begin{equation} \ket{\psi}=\int_{I_1}dw\int_{I_2}dy\ket{w(H_1,U),y}\psi(w(H_1,U),y), \end{equation}
where $\psi(w(H_1,U),y)=\braket{w(H_1,U),y|\psi}$ and $I_1$ and $I_2$ are real intervals. If $I_1$ and/or $I_2$ do not contain $w_e$ and $y_e$, respectively, then the reality condition implemented via Eq. \eqref{Mrealityofficial} imposes no restriction on $\bra{\psi}M(z,H_1,U)\ket{\psi}$. However, if $I_1\subset I_w$ and $I_2\subset I_y$, meaning they contain $w_e$ and $y_e$, respectively, then:
\begin{equation} \bra{\psi}M(z,H_1,U)\ket{\psi}=\int_{I_2}dy|\psi(z,y)|^2. \end{equation}
which leads to:
\begin{equation} \bra{\psi}\int_{I_1}dzM(z,H_1,U)\ket{\psi}=1. \end{equation}
This implies that whenever the state $\ket{\psi}$ is within the resolution range $I_{w_e}$ and $I_{y_e}$, i.e., sufficiently close to an eigenstate of $H_1$ and $U^{\dagger}H_1U$, its variation of energy is necessarily confined to the interval $I_{w_e}$ around $w_e$.

We thus assume that Eqs. \eqref{Mrealityofficial} and \eqref{Mrealityofficial2} will be taken to define the continuous analogue of the reality condition. Formally we then impose that:
\begin{itemize}
    \item \textbf{Reality condition for the continuous case}:  \emph{For any system $\Omega$, operators $H_1$, evolution $U$, if there is an eigenvector $\ket{w_e,y_e}$ of $\Delta(H_1,U)$, such that $H_1\ket{w_e,y_e}=e\ket{w_e,y_e}$,  $U^{\dagger}H_1 U\ket{w_e,y_e}=\epsilon_e\ket{w_e,y_e}$, and $\Delta(H_1,U)\ket{w_e,y_e}=w_e\ket{w_e,y_e}=(\epsilon_e-e)\ket{w_e,y_e}$ then there are intervals $I_{w_e}$ and $I_{y_e}$ with bounded lengths such that the operators $\{M(z, H_1,U)\}$ satisfy Eqs. \eqref{Mrealityofficial} and \eqref{Mrealityofficial2}}.
\end{itemize}

\subsection{Some formalities regarding POVMs}

Throughout the main text, we expressed the POVM elements as ${M(z,H_1,U)}$, ensuring that they satisfy $\int_{-\infty}^{\infty}M(z,H_1,U)dz=\mathbbm{1}$ and $M(z,H_1,U)\geq 0$. For an initial state $\rho$, the probability density was defined as $\wp(z,H_1,U,\rho)=\Tr[M(z,H_1,U)\rho]$. However, for continuous variables, POVM elements are not strictly defined as $M(z,H_1,U)$. This notation is adopted throughout the main text and these Appendices for simplicity, but here we clarify its formal basis.

When considering energy variation as a continuous variable $z$, POVMs are formally defined as operators $\mathcal{M}(I,H_1,U)$ acting on open intervals $I\subset\mathbbm{R}$ such that:
\begin{equation}
    \Tr [\mathcal{M}(I,H_1,U)\rho]=P(I,H_1,U,\rho) \label{POVMFORMAL}
\end{equation}
where $P(I,H_1,U,\rho)$ gives the probability that the energy variation falls within $I$. POVMs must satisfy $\mathcal{M}((-\infty,\infty),H_1,U)=\mathbbm{1}$ (as energy variation can take any real value) and $\mathcal{M}(I,H_1,U)\geq 0$ for all $I\subset \mathbbm{R}$.

The probability density function $\wp(z,H_1,U,\rho)$ is thus defined as $dP/dz$ (see page 54 of Ref. \cite{Athreya2006}), allowing us to express:
\begin{equation} P(I,H_1,U,\rho)=\int_I dz\, \wp(z,H_1,U,\rho). \end{equation}

A key concept in measure theory \cite{Athreya2006,Fristedt2013} is that of ``almost everywhere in z'' ($z$-a.e.), meaning a statement holds for all $z\in \mathbbm{R}$ except on subsets $I_0$ of measure zero, i.e., open intervals $I_0$ such that
\begin{equation} \int_{I_0} dz=0. \end{equation}
Two probability density functions $\wp(z,H_1,U,\rho)$ and $\wp'(z,H_1,U,\rho)$ are equivalent if $\wp=\wp'$ $z$-a.e. (see pages 35 and 116 of \cite{Fristedt2013}), meaning they yield the same probabilities for any interval $I$:
\begin{equation} \int_I dz\, \wp(z,H_1,U,\rho) =\int_I dz\,\wp'(z,H_1,U,\rho). \end{equation}
For example, if $\wp(z_0,H_1,U,\rho) \neq \wp'(z_0,H_1,U,\rho)$ for a single point $z_0$ but $\wp(z,H_1,U,\rho) = \wp'(z,H_1,U,\rho)$ everywhere else, they still define the same probability distribution since a single point has zero measure.

Similarly, it is common to define a ``density operator'' $M(z,H_1,U)$ for POVM elements:
\begin{equation} \mathcal{M}(I,H_1,U)=\int_I dz\,M(z,H_1,U). \end{equation}
This allows $\Tr[M(z,H_1,U)\rho]$ to represent the probability density. In both these Appendices and the main text, we referred to $M(z,H_1,U)$ as POVM elements in $\mathbbm{M}(H_1,U)$ for simplicity. Formally, they are not exactly elements of a POVM set, and they are defined by the POVM map $\mathcal{M}(I,H_1,U)$ in Eq. \eqref{POVMFORMAL} (see \cite{Heinosaari2012,Busch1996} for rigorous POVM definitions). The crucial point for our proofs is the equivalence relation: if $M'(z,H_1,U)=M(z,H_1,U)$ $z$-a.e., they define the same POVM $\mathcal{M}(I,H_1,U)$ for all interval of measure nonzero $I\subset \mathbbm{R}$, since:
\begin{equation} \mathcal{M}(I,H_1,U)=\int_I dz\,M(z,H_1,U)=\int_I dz\, M'(z,H_1,U). \end{equation}
Thus, to prove that two POVMs $\mathcal{M}(I,H_1,U)$ and $\mathcal{M}'(I,H_1,U)$ are identical for all $I\subset \mathbbm{R}$, it suffices to show:
\begin{equation} M'(z,H_1,U)=M(z,H_1,U) \quad \text{$z$-a.e.} \end{equation}
This formalism will be essential in extending the proof of Result 1 in Appendix B.

\subsection{Condition 1 and 3 imply $[H_1+H_2,U]=0\implies M(z,H_1,U)=M(-z,H_2,U)$}

In this subsection we show that for every POVM $\mathbbm{M}$ satisfying condition 1 and 3, it follows that $M(z,H_1,U)=M(-z,H_2,U)$. To prove this result, we first notice that, by Condition 1, whenever $[H_1 + H_2, U] = 0$, the following relation holds for any initial state $\rho$:
\begin{equation}
    \wp(z,H_1,U,\rho)=\Tr[M(z,H_1,U)\rho]=\Tr[M(-z,H_2,U)\rho]=\wp(-z,H_2,U,\rho)\label{condition1and3assumptions}
\end{equation}
This leads to two possible scenarios:
\begin{itemize}
    \item (i) $M(z,H_1,U)$ and $M(-z,H_2,U)$ act on a Hilbert space $\mathcal{H}$ with countable basis $\{\ket{i}\}$, or
    \item (ii) $M(z,H_1,U)$ and $M(-z,H_2,U)$ act on a Hilbert space $\mathcal{H}$ with continuous basis.
\end{itemize}
We consider both cases:
\begin{itemize}
    \item \textbf{Case (i):} In this scenario, $M(z,H_1,U)$ and $M(-z,H_2,U)$ can only differ if there exist basis elements $\ket{j}$ and $\ket{k}$ in $\{\ket{i}\}$ such that
\begin{equation}
        \bra{j}M(z,H_1,U)\ket{k}\neq \bra{j}M(-z,H_2,U)\ket{k}.\label{conditionjj}
    \end{equation}
    However, this contradicts Eq.~\eqref{condition1and3assumptions} in all possible cases. To see why, consider the following two scenarios:
    \begin{itemize}
        \item \textbf{Case $\mathbf{k= j}$:} If Eq.~\eqref{conditionjj} holds for $k = j$, then by Condition 3, $M(z,H_1,U)$ and $M(-z,H_2,U)$ must be independent of $\rho$, meaning we can freely choose $\rho$ while still satisfying Eq.~\eqref{condition1and3assumptions}. Setting $\rho = \ket{j}\bra{j}$, we obtain:
    \begin{equation}
        \Tr[M(z,H_1,U) \rho] = \bra{j}M(z,H_1,U)\ket{j} \neq \bra{j}M(-z,H_2,U)\ket{j} = \Tr[M(-z,H_2,U) \rho].
    \end{equation}
    This directly contradicts Eq.~\eqref{condition1and3assumptions}.

    \item \textbf{Case $\mathbf{k\neq j}$:} Suppose Eq.~\eqref{conditionjj} holds for some $k \neq j$, so that we define the difference:
    \begin{equation}
        \Phi_{jk} = \bra{j}M(z,H_1,U)\ket{k} - \bra{j}M(-z,H_2,U)\ket{k},
    \end{equation}
    where $|\Phi_{jk}| \neq 0$. By Condition 3, we can consider any state $\rho = \ket{\psi} \bra{\psi}$ without altering $M(z,H_1,U)$ or $M(-z, H_2,U)$. Choosing 
    \begin{equation}
        \ket{\psi} = \frac{a\ket{j} + b\ket{k}}{|a|^2 + |b|^2},
    \end{equation}
    where $a = 1$ and $b = \Phi_{jk}^*$ (i.e., the conjugate of $\Phi_{jk}$), we compute:
    \begin{equation}
        \Tr[M(z,H_1,U)\rho] - \Tr[M(-z,H_2,U)\rho] = \frac{ab^*\Phi_{jk}^*+a^* b \Phi_{jk}}{|a|^2+|b|^2} = \frac{|\Phi_{jk}|^2+  |\Phi_{jk}|^2}{1+|\Phi_{jk}|^2} > 0,
    \end{equation}
    which contradicts Eq.~\eqref{condition1and3assumptions}. 
\end{itemize}
Given that we consider arbitrary $j$ and $k$, it follows that for all $j$ and $k$ it cannot be the case in which Eq.~\eqref{conditionjj} holds. As a result, we conclude:
    \begin{equation}
        M(z,H_1,U) = M(-z,H_2,U).
    \end{equation}
    \item \textbf{Case (ii):} Since the space $\mathcal{H}$ cannot be described by a countable basis, we consider arbitrary projectors $\ket{y}\bra{y}dy$, with $dy$ real, so that any operator $A$ acting on $\mathcal{H}$ can be written as:
    \begin{equation}
        A = \iint_{-\infty}^{\infty} dy\,dy' \ket{y}\bra{y} A \ket{y'}\bra{y'}.
        \label{decompositioncont}
\end{equation}
Assuming that $M(z,H_1,U) \neq M(-z,H_2,U)$, and given that $M(z,H_1,U)$ acts on and has its image in a space dense in $L_2(-\infty,\infty)$ (i.e., the space of normalizable measurable functions defined by the 2-norm that is bounded. See Section 3.2 of \cite{Athreya2006}), then there must exist states $\ket{\psi}$ and $\ket{\phi}$ such that 
\begin{equation}
    \bra{\phi}(M(z,H_1,U) - M(-z,H_2,U))\ket{\psi} \neq 0.\label{initialassumption}
\end{equation}
Thus, there must be intervals $I_1$ and $I_2$ with nonzero measure, i.e., $\int_{I_1} dy \neq 0$ and $\int_{I_2} dy \neq 0$, such that 

\begin{equation}
    \ket{\psi_1} = \frac{\int_{I_1} dy \ket{y} \psi(y)}{\sqrt{\int_{I_1} dy |\psi(y)|^2}}, \quad \ket{\phi_1} =\frac{ \int_{I_2} dy \ket{y} \phi(y)}{\sqrt{\int_{I_2} dy |\phi(y)|^2}},
\end{equation}
and
\begin{equation}
    \bar{\Phi} \coloneqq \bra{\phi_1}DM(z)\ket{\psi_1}=\int_{I_2} dy \int_{I_1} dy' \phi_1^{*}(y) \bra{y}DM(z)\ket{y'} \psi_1(y') \neq 0, \label{auxmzmmz}
\end{equation}
where we define, for simplicity, $DM(z) = M(z,H_1,U) - M(-z,H_2,U)$, as well as $\braket{y|\psi_1} = \psi_1(y)$ and $\braket{y|\phi_1} = \phi_1(y)$, with $|\bar{\Phi}| > 0$. If, instead, Eq.~\eqref{auxmzmmz} does not hold for all nonzero measure intervals $I_1$ and $I_2$, then we obtain 
\begin{equation}
    \bra{\phi} DM(z) \ket{\psi} = \bra{\phi}(M(z,H_1,U) - M(-z,H_2,U))\ket{\psi} = 0,
\end{equation}
which contradicts our initial assumption in Eq. \eqref{initialassumption}. There are only two possible cases: either (ii)a $I_1 \cap I_2 \neq \emptyset$ or (ii)b $I_1 \cap I_2 = \emptyset$. We analyze both scenarios, demonstrating that assuming $\bra{\phi_1}(M(z,H_1,U) - M(-z,H_2,U))\ket{\psi_1} \neq 0$ leads to a contradiction:

\begin{itemize}
    \item \textbf{Case (ii)a:} Suppose first that $I_1 = I_2$. In this case, we consider the state $\rho_\psi = \ket{\psi}\bra{\psi}$, where
    \begin{equation}
        \ket{\psi} = \frac{a\ket{\psi_1} + b\ket{\phi_1}}{|a|^2 + |b|^2}
    \end{equation}
    with $a$ and $b$ being complex numbers. By Condition 3, it follows that
    \begin{equation}
        \bra{\psi_1} DM(z) \ket{\psi_1} = \bra{\phi_1} DM(z) \ket{\phi_1} = 0,
    \end{equation}
    which, considering that both $M(z,H_1,U)$ and $M(-z,H_2,U)$ are Hermitian, results in
    \begin{equation}
        \Tr[DM(z)\rho_\psi] = \frac{a^* b \bar{\Phi}^* + ab^* \bar{\Phi}}{|a|^2 + |b|^2}.
    \end{equation}
    Choosing $a = 1$ and $b = \bar{\Phi}$, we obtain
    \begin{equation}
        \Tr[DM(z)\rho_\psi] = \frac{2|\bar{\Phi}|^2}{1 + |\bar{\Phi}|^2} > 0,
    \end{equation}
    which contradicts Condition 3 and Eq.~\eqref{condition1and3assumptions}. Therefore, Eq.~\eqref{auxmzmmz} cannot hold for $I_1 = I_2$.

    Next, consider the case where $I_1 \cap I_2 = I_3'$. We define the intervals $I_1' = I_1 \setminus I_3'$ and $I_2' = I_2 \setminus I_3'$, where $B \setminus A$ denotes the complement of $A$ with respect to $B$. Thus, there is no intersection between $\{I_1', I_2', I_3'\}$. It must then follow that
    \begin{equation}
        \Phi_{jk} = \int_{I_j'} dy \int_{I_k'} dy' \phi_1^{*}(y) \bra{y}DM(z)\ket{y'} \psi_1(y') \neq 0, \label{auxjkM}
    \end{equation}
     must hold for some $j \neq k$, with $j \in \{2,3\}$ and $k \in \{1,3\}$. If this were not the case, then, considering the argument that $\int_{I_j'} dy \int_{I_j'} dy' \phi_1^{*}(y) \bra{y}DM(z)\ket{y'} \psi_1(y') = 0$, Eq.~\eqref{auxmzmmz} could not hold without the existence of $j \neq k$ such that Eq.~\eqref{auxjkM} is valid. Therefore, for the indices $j \neq k$ for which Eq.~\eqref{auxjkM} holds, we define the state $\rho_{jk} = \ket{\psi_{jk}}\bra{\psi_{jk}}$ such that
    \begin{equation}
        \ket{\psi_{jk}} = \frac{a\ket{\psi_k} + b\ket{\psi_j}}{|a|^2 + |b|^2},
    \end{equation}

    with $a = 1$, $b = \Phi_{jk}^*$, and
    \begin{equation}
        \ket{\psi_j} = \frac{\int_{I_j'} \ket{y} \phi_1(y) dy}{\sqrt{\int_{I_j'} |\phi_1(y)|^2 dy}}, \quad \ket{\psi_k} = \frac{\int_{I_k'} dy \psi_{1}(y)}{\sqrt{\int_{I_k'} |\psi_{1}(y)|^2 dy}}.
    \end{equation}
    By Condition 3, $\bra{\psi_k}DM(z)\ket{\psi_k} = \bra{\psi_j}DM(z)\ket{\psi_j}$, leading to
    \begin{equation}
        \Tr[DM(z)\rho_{jk}] = \frac{|\bar{\Phi}_{jk}|^2 + |\bar{\Phi}_{jk}|^2}{\sqrt{N_j' N_k'}(1 + |\bar{\Phi}_{jk}|^2)} > 0, \label{auxjkM2}
    \end{equation}
    contradicting Eq.~\eqref{condition1and3assumptions}.

    \item \textbf{Case (ii)b:} The proof follows analogously by substituting $I_j' \to I_1$, $I_k' \to I_2$, $\ket{\psi_j} \to \ket{\phi_1}$, and $\ket{\psi_k} \to \ket{\psi_1}$ in Eqs.~\eqref{auxjkM}-\eqref{auxjkM2}.
\end{itemize}
By considering the only possibilities (ii)a and (ii)b for the continuous case (ii), we showed that there can be no intervals $I_1$ and $I_2$ such that Eq. \eqref{auxmzmmz} holds, and, therefore, it cannot be true that $M(z,H_1,U)\neq M(-z,H_2,U)$.
\end{itemize}

Since for the only possible cases (i) and (ii), $M(z,H_1,U)=M(-z,H_2,U)$, then we can only conclude that conditions 1 and 3 together with $[H_1+H_2,U]=0$ imply $M(z,H_1,U)=M(-z,H_2,U)$.
\section{ APPENDIX B: Completion of Result 1 -- degeneracies and continuous basis}

For the proof of Result 1 in the ``Main results'' section of the main text, we demonstrated that the OBS protocol satisfies the CRIN conditions for any energy operator $H_1$ and unitary evolution $U$. This proof holds independently of whether $\Delta(H_1,U)$ has degeneracies or a continuous spectrum. However, to establish that the OBS protocol is the \emph{only} CRIN protocol, we considered the case where $\Delta(H_1,U)$ has a discrete, non-degenerate basis. Here, we extend this proof by employing the formalism introduced in the Preliminaries section, completing the proof of Result 1 for cases where $\Delta(H_1,U)$ exhibits degeneracies and/or has a continuous spectrum.

\subsection{Discrete degeneracies}

To prove that the OBS protocol is the \emph{only} CRIN  protocol when $\Delta(H_1,U)$ have discrete basis with degeneracies, we consider $\Delta (H_1,U)$ as in Eq. \eqref{projectorbasis12}, diagonalized by a basis $\{\ket{\delta_i^{j}(H_1, U)}\}$. Similar to the non-degenerate case, we aim here to demonstrate that for any CRIN protocol $\mathbbm{M}'$, and any $H_1$ and $U$, the POVM $\mathbbm{M}'(H_1, U)$ should coincide with the OBS protocol $\mathbbm{M}_{\text{\tiny OBS}}(H_1, U)$. For that, as we discussed in the Appendix A, is enough for us to prove that  
    \begin{equation}
        M'(z, H_1, U) = M_{\text{\tiny OBS}}(z, H_1, U), \quad \text{z-a.e..}\label{assumptioninitialsm}
    \end{equation}
    Proving this relation establishes that there is no CRIN protocol $\mathbbm{M}'$ such that $\mathbbm{M}'(H_1, U) \neq \mathbbm{M}_{\text{\tiny OBS}}(H_1, U)$, completing the proof.

    To this end, consider an arbitrary CRIN protocol $\mathbbm{M}'$. We first establish that, because $\mathbbm{M}'$ is CRIN, then it follows:  
    \begin{equation}
        \bra{\delta_{i}^{k}(H_1,U)} M'(z,H_1,U)\ket{\delta_{i}^{k}(H_1,U)} = \delta^{\text{\tiny \textbf{D}}}[z-\delta_{i}(H_1,U)]  
        = \bra{\delta_{i}^{k}(H_1,U)}M_{\text{\tiny OBS}}(z, H_1,U)\ket{\delta_{i}^{k}(H_1,U)},
        \label{Mdiagonalsm}
    \end{equation}
    for all elements $\{\ket{\delta_{i}^{k}(H_1,U)}\}$ of the eigenbasis. To prove this, we consider similar arguments as in the non-degenerate case. Consider operators $U'$ and $H_2$ acting on $\mathcal{H}'$ and $\mathcal{H}\otimes \mathcal{H}'$, respectively, along with a vector $\ket{v}\in \mathcal{H}'$, such that:
    \begin{eqnarray}
        &[H_1 \otimes \mathbbm{1}_{\mathcal{H}'} + H_2, U \otimes U'] = 0, \label{commutresult10sm} \\
        &H_2 \ket{\delta_i^{k}(H_1, U), v} = E_i^{k} \ket{\delta_i^{k}(H_1, U), v}, \\
        &(U^\dagger \otimes U^{'\dagger}) H_2 (U \otimes U') \ket{\delta_i(H_1, U), v} = E_i^{'k} \ket{\delta_i^{k}(H_1, U), v}, \label{result12ndrelation0sm}
    \end{eqnarray}
    where $E_i^{'k} = E_i^{k} - \delta_i(H_1, U)$. By Result 2, the existence of such $H_2$, $U'$, and $\ket{v}$ is guaranteed. Since $\mathbbm{M}'$ satisfies all CRIN conditions for all arbitrary energy and evolution operators, Conditions 1, 3 and 4 (Conservation Laws, Independent of the initial state, No-Signaling) imply the following for the evolution $U \otimes U'$ in Eq.~\eqref{commutresult10sm} (see section ``Condition 1 and 3 imply...'' of Appendix A):
    \begin{eqnarray}
        M'(z, H_1 \otimes \mathbbm{1}_{\mathcal{H}'}, U \otimes U') = M'(-z, H_2, U \otimes U'), \label{req2proofsm} \\
        M'(z, H_1 \otimes \mathbbm{1}_{\mathcal{H}'}, U \otimes U') = M'(z, H_1, U) \otimes \mathbbm{1}_{\mathcal{H}'}. \label{localtoglobalsm}
    \end{eqnarray}
    Additionally, using the Reality Condition (Condition 2), we can consider the initial state $\rho_i = \ket{\delta_{i}(H_1,U),v}\bra{\delta_{i}(H_1,U),v}$, which is an eigenstate of $(U^\dagger \otimes U^{'\dagger}) H_2 (U \otimes U')$ and $H_2$ with eigenvalues $E_i^{'k}$ and $E_i^{k}$, respectively. For any arbitrary $z'$, we then deduce:
    \begin{equation}
        \Tr [M'(z', H_2, U \otimes U') \rho_i] = \delta^{\text{\tiny \textbf{D}}}[z' + \delta_{i}(H_1,U)]. \label{Mdiagonal2sm}
    \end{equation}
    Substituting $z' \to -z$, we have:
    \begin{equation}
        \Tr [M'(-z, H_2, U \otimes U') \rho_i] = \delta^{\text{\tiny \textbf{D}}}[z - \delta_{i}(H_1,U)]. \label{Mdiagonal3sm}
    \end{equation}
    Combining Eqs.~\eqref{req2proofsm}, \eqref{localtoglobalsm}, \eqref{Mdiagonal3sm}, and the fact that $\mathbbm{M}'$ is state independent (according to the condition of Independence on state), we obtain Eq.~\eqref{Mdiagonalsm}. This deduction is valid for arbitrary $z$, and $\ket{\delta_{i}(H_1,U)}$, demonstrating Eq.~\eqref{Mdiagonalsm} for any CRIN protocol $M'(z, H_1, U)$.

    Our next goal is to prove that for any element $M'(z, H_{1}, U)$ of an CRIN protocol $\mathbbm{M}'$, it follows that
    \begin{equation}
            \bra{\delta_{i}^k(H_{1},U)}M'(z,H_{1},U)\ket{\delta_{j}^{k'}(H_{1},U)}=0=\bra{\delta_{i}^{k}(H_{1},U)}M_{\text{\tiny OBS}}(z,H_{1},U)\ket{\delta_{j}^{k'}(H_{1},U)}.\quad \text{z-a.e.}\label{Moffdiagonal1sm}
    \end{equation}
    for either $i \neq j$ and any $k$ and $k'$ or for $i=j$ and $k\neq k'$. To show this, first note that $M'(z, H_{1}, U)$ is Hermitian, non-negative, and can be diagonalized as $M'(z, H_{1}, U) = \int_{-\infty}^{\infty} w(z) \ket{w(z)}\bra{w(z)}dw$, with $w(z) \geq 0$. Notice that here we considered the continuous basis $\ket{w(z)}\bra{w(z)}dw$ seeking generality but we could consider the discrete in a similar way as in the main text. As a result, we get:
    \begin{equation}
        \ba{l}
            M'(z, H_{1}, U) = \int dw\sum_{ik} w(z) |\gamma_{ik}(z,w)|^2 \ket{\delta_{i}^{k}(H_1, U)}\bra{\delta_{i}^{k}(H_1, U)} + \\
            + \int dw\sum_{i,k\neq k'} w(z) \gamma_{ik'}^{*}(z,w) \gamma_{ik}(z,w)\ket{\delta_{i}^k(H_1, U)}\bra{\delta_i^{k'}(H_1, U)}+\int dw\sum_{i\neq i',k, k'} w(z) \gamma_{i'k'}^{*}(z,w) \gamma_{ik}(z,w)\ket{\delta_{i}^k(H_1, U)}\bra{\delta_{i'}^{k'}(H_1, U)}. \label{DiagMruimsm}
        \ea
    \end{equation}
    where the integrals are over $(-\infty,\infty)$ and $\gamma_{ik}(z,w) = \braket{\delta_{i}^k(H_1, U) | w(z)}$.
    Two and only two cases arise: either (1) $z \neq \delta_{i}(H_1, U)$ for all $i$ or (2) there exists $j$ such that $z = \delta_{j}(H_1, U)$. Let us consider both cases:
    
    \emph{Case (1)}: In this case, Eq.~\eqref{Mdiagonalsm} implies $\bra{\delta_i^{k}(H_1, U)} M'(z, H_{1}, U) \ket{\delta_i^{k}(H_1, U)} = 0$ for all $i$ and $k$. From Eq.~\eqref{DiagMruimsm} and the fact that $w(z)\geq 0$, leads to $\int dw \,w(z) |\gamma_{ik}(z,w)|^2 = 0$ for all $i$ and $k$, and hence $\int_I dw w(z) \gamma_{ik}(z,w)=\int_I dw w(z) \gamma_{ik}^{*}(z,w) = 0$ for any interval $I$ of measure non-zero. Thus, all the off-diagonal terms of $M'(z, H_{1}, U)$ with respect to the basis $\{\ket{\delta_i^{k}(H_1, U)}\}$ vanish, and Eq.~\eqref{Moffdiagonal1sm} holds for all $i \neq j$ and $k$ and $k'$ and all $i=j$ with $k\neq k'$.

    \emph{Case 2}: In this case, from Eq.~\eqref{Mdiagonalsm}, $\bra{\delta_{i}^k(H_1, U)} M'(z, H_{1}, U) \ket{\delta_{i}^k(H_1, U)} = 0$ for all $j \neq i$. From Eq.~\eqref{DiagMruimsm}, this implies $\int dw \,w(z) |\gamma_{ik}(z,w)|^2 = 0$ for all $i \neq j$ and all $k$. Since $w(z) \geq 0$ and real, we must have  $\int_I dw w(z) \gamma_{ik}(z,w)=\int_I dw w(z) \gamma_{ik}^{*}(z,w) = 0$ for any non-empty interval $I$, for all $i \neq j$ and all $k$. Consequently, Eq.~\eqref{Moffdiagonal1sm} holds for all $i \neq j$ and any $k$ and $k'$. All left is to prove that for $i=j$ and $k\neq k'$, Eq.~\eqref{Moffdiagonal1sm} also holds. In this direction, notice that considering the expansion in Eq. \eqref{DiagMruimsm} and considering Eq.~\eqref{Mdiagonalsm} and the fact that Eq.~\eqref{Moffdiagonal1sm} holds for all $i \neq j$ and $k$ and $k'$, it follows that, for the given $j$ such that $z=\delta_{j}(H_1,U)$ and any given $k\neq k'$,
    \begin{equation}
           \bra{\delta_{j}^{k}(H_1, U)}M'(z, H_{1}, U)\ket{\delta_{j}^{k'}(H_1, U)} =  \Gamma_{jkk'}(z). \label{DiagMruim2sm}
    \end{equation} 
    where 
    \begin{equation}
        \Gamma_{jkk'}(z)=\int dw w(z) \gamma_{jk'}^{*}(z,w) \gamma_{jk}(z,w).
    \end{equation}
    Let us assume that there are two indexes $l\neq l'$ and a non-empty interval $I\subset \mathbbm{R}$ such that $\int_I\Gamma_{ill'}(z)dz=\mathcal{J}\neq 0$, i.e. such that Eq. \eqref{Moffdiagonal1sm} does not hold for $i=j$ and $l\neq l'$. Since $\Gamma_{ill'}(z)=0$ for $i\neq j$, then it is necessary that $\delta_{j}(H_1,U)\in I$. In this sense, consider an initial state $\rho_\psi=\ket{\psi}\bra{\psi}$ where
    \begin{equation}
        \ket{\psi}=\frac{\alpha\ket{\delta_j^{l}(H_1,U)}+\beta\ket{\delta_j^{l'}(H_1,U)}}{\sqrt{|\alpha|^2+|\beta|^2}},
    \end{equation} 
     $\alpha=1$, and $\beta=\mathcal{J}^{*}$. As a result, we get from Eqs. \eqref{Mdiagonalsm} and \eqref{DiagMruim2sm}, and the fact that $\delta_{j}(H_1,U)\in I$, that 
    \begin{equation}
        \int_{I}dz\Tr[ M'(z,H_{1},U)\rho_\psi ]=\frac{|\alpha|^2+|\beta|^2+2\Re[\int_{I}dz\Gamma_{\bar{i}ll'}(z)\alpha^* \beta]}{|\alpha|^2+|\beta|^2}=1+\frac{2|\mathcal{J}|^2}{1+|\mathcal{J}|^2}>1,
    \end{equation}
    This, however, contradicts the fact that $M'(z,H_{1},U)$ are elements of a POVM, since for any interval $I$ and any state $\rho$, the POVM elements $M'(z,H_{1},U)$ should satisfy $\int_{I}dz\Tr[ M(z,H_{X},U)\rho ]\leq 1$. Therefore, it cannot be the case that there is non-empty interval $I$ and indexes $l\neq l'$ such that $\int_I\Gamma_{jll'}(z)dz\neq 0$. In other words, it means that $\Gamma_{jll'}(z)=0$ $z$-a.e. and Eq. \eqref{Moffdiagonal1sm} is thus proved for all $i=j$ and $k\neq k'$.

    Combining Eq.~\eqref{Mdiagonalsm} and Eq.~\eqref{Moffdiagonal1sm} we find that for all $i$ and $j$, 
    \begin{equation}
        \bra{\delta_{i}(H_{1}, U)} M'(z, H_{1}, U) \ket{\delta_{j}(H_{1}, U)} = \bra{\delta_{i}(H_{1}, U)} M_{\text{\tiny OBS}}(z, H_{1}, U) \ket{\delta_{j}(H_{1}, U)} \quad \text{$z$-a.e.}.
    \end{equation} 
    Thus, 
    \begin{equation}
        M'(z, H_{1}, U) = M_{\text{\tiny OBS}}(z, H_{1}, U) \quad \text{$z$-a.e.}
    \end{equation}
    for all $H_{1}$, and $U$. Therefore, we conclude that any CRIN protocol $\mathbbm{M}'$ must be equivalent to the OBS protocol $\mathbbm{M}_{\text{\tiny OBS}}$, completing the proof for the degenerate, discrete case.

\subsection{Continuous, degenerate case}

To prove that the OBS protocol is the \emph{only} CRIN  protocol when $\Delta(H_1,U)$ have continuous, degenerate basis, we consider $\Delta (H_1,U)$ as in Eq. \eqref{continuousoperatorspectral}, diagonalized in terms of the projectors $\{\ket{w(H_1,U),y}\bra{w(H_1,U),y}\}$. Similar to the non-degenerate case, we aim here to demonstrate that for any CRIN protocol $\mathbbm{M}'$, and any $H_1$ and $U$, the POVM $\mathbbm{M}'(H_1, U)$ should coincide with the OBS protocol $\mathbbm{M}_{\text{\tiny OBS}}(H_1, U)$. For that, as we discussed in Appendix A, it is enough for us to prove that  
    \begin{equation}
        M'(z, H_1, U) = M_{\text{\tiny OBS}}(z, H_1, U), \quad \text{z-a.e..}\label{assumptioninitialsmcont}
    \end{equation}
    Proving this relation establishes that there is no CRIN protocol $\mathbbm{M}'$ such that $\mathbbm{M}'(H_1, U) \neq \mathbbm{M}_{\text{\tiny OBS}}(H_1, U)$, completing the proof. To this end, consider an arbitrary CRIN protocol $\mathbbm{M}'$ and its associated POVM $\mathbbm{M}'(H_1,U)$ and elements $M(z,H_1,U)$. Regarding $H_1$ and $U$,  consider operators $U'$ and $H_2$ acting on $\mathcal{H}'$ and $\mathcal{H}\otimes \mathcal{H}'$, respectively, along with a vector $\ket{v}\in \mathcal{H}'$, such that:
    \begin{eqnarray}
        &[H_1 \otimes \mathbbm{1}_{\mathcal{H}'} + H_2, U \otimes U'] = 0, \label{commutresult10smcont} \\
        &H_2 \ket{w(H_1,U),y}\otimes \ket{v} = E_w^{y} \ket{w(H_1,U),y}\otimes \ket{v}, \\
        &(U^\dagger \otimes U^{'\dagger}) H_2 (U \otimes U') \ket{w(H_1,U),y}\otimes \ket{v} = E_w^{'y} \ket{w(H_1,U),y}\otimes \ket{v}, \label{result12ndrelation0smcont}
    \end{eqnarray}
    where $E_w^{'y} = E_w^{y} - w(H_1, U)$. By Result 2, the existence of such $H_2$, $U'$, and $\ket{v}$ is guaranteed. Since $\mathbbm{M}'$ satisfies all CRIN conditions for all arbitrary energy and evolution operators, Conditions 1, 3 and 4 (Conservation Laws, Independent of the initial state, No-Signaling) imply the following for the evolution $U \otimes U'$ in Eq.~\eqref{commutresult10smcont} (see Appendix A):
    \begin{eqnarray}
        M'(z, H_1 \otimes \mathbbm{1}_{\mathcal{H}'}, U \otimes U') = M'(-z, H_2, U \otimes U'), \label{req2proofsmcont} \\
        M'(z, H_1 \otimes \mathbbm{1}_{\mathcal{H}'}, U \otimes U') = M'(z, H_1, U) \otimes \mathbbm{1}_{\mathcal{H}'}. \label{localtoglobalsmcont}
    \end{eqnarray}
    Additionally, since $\ket{w(H_1,U),y}\otimes \ket{v}$ is a non-normalized eigenvector of $H_2$ and $(U^\dagger\otimes U^{'\dagger})H_2(U\otimes U^{'})$, with difference in the eigenvalues as $-w(H_1,U)$, then we can consider the Reality Condition (Condition 2), defined for the continuous case in Appendix A, to deduce that there exist open intervals $I_{w}$ and $I_{y}$ with bounded lengths $d_{w}=\int_{I_{w}}dw$ and $d_{y}=\int_{I_{y}}dy$, such that
\begin{equation}
    M'(z', H_2, U \otimes U')=\int_{I_{y}}dy \ket{z,y,v}\bra{z,y,v}+M_{nd}(z', H_2, U \otimes U'), \quad \text{for $z'=-w(H_1,U)$}
\end{equation}
where $M_{nd}(z', H_2, U \otimes U')$ is hermitian and for $z'=-w(H_1,U)$ satisfies
\begin{equation}
    \bra{w'(H_1,U),y',v}M_{nd}(z', H_2, U \otimes U')\ket{w''(H_1,U),y'',v}=0,
\end{equation}
for all $w,w',y,y'$ such that $w',w''\in I_{w}$ and $y',y''\in I_{y}$.
Since $z'$ here is arbitrary, we can redefine it as $z' \to -z$, to obtain:
    \begin{equation}
        M'(-z, H_2, U \otimes U')=\int_{I_{y}}dy \ket{z,y,v}\bra{z,y,v}+M_{nd}(-z, H_2, U \otimes U') \quad \text{for $z=w(H_1,U)$} \label{Mdiagonal3smcont}
    \end{equation}
    Because $\ket{v}$ is normalized, then we can combine Eqs.~\eqref{req2proofsmcont}, \eqref{localtoglobalsmcont}, \eqref{Mdiagonal3smcont}, to deduce that
    \begin{equation}
        \ba{rl}
        \Tr_{\mathcal{H}'}[M'(-z, H_2, U \otimes U')(\mathbbm{1}_{\mathcal{H}}\otimes\ket{v}\bra{v})]&=\Tr_{\mathcal{H}'}[M'(z, H_1\otimes\mathbbm{1}, U \otimes U')(\mathbbm{1}_{\mathcal{H}}\otimes\ket{v}\bra{v})]\\
        &=\Tr_{\mathcal{H}'}[M'(z, H_1, U )\otimes \mathbbm{1}_{\mathcal{H}'}(\mathbbm{1}_{\mathcal{H}}\otimes\ket{v}\bra{v})]\\
        &=M'(z, H_1, U ).\label{MMcontinuous}
        \ea
    \end{equation}
    Considering Eqs. \eqref{Mdiagonal3smcont} and \eqref{MMcontinuous}, we then have that
    \begin{equation}
        M'(z, H_1, U )=\Tr_{\mathcal{H}'}[M'(-z,H_2, U \otimes U')(\mathbbm{1}_{\mathcal{H}}\otimes\ket{v}\bra{v})]=M_d'(z)+M_{nd}'(z)\label{mdnd0}
    \end{equation}
    where
    \begin{equation}
        M_{d}'(z)=\int_{I_{y}}dy \ket{z,y}\bra{z,y}\quad \text{for $z=w(H_1,U)$}\label{mdl1}
    \end{equation}
    and $M_{nd}'(z)$ is hermitian and satisfies
    \begin{equation}
        \bra{w'(H_1,U),y'}M_{nd}'(z)\ket{w''(H_1,U),y''}=0,\label{mndl1}
    \end{equation}
for all $w, w',y,y'$ such that $w',w''\in I_{w}$ and $y',y''\in I_{y}$.
Notice that the results deduced until now hold for arbitrary $w$ and $y$, so that Eqs. \eqref{mdnd0}-\eqref{mndl1} should be valid for any $w$ and $y$. Therefore, for $M'(z, H_1, U )$ to satisfy Eqs. \eqref{mdnd0}-\eqref{mndl1} for all $w$ and $y$, then it must have the form
    \begin{equation}
         M'(z, H_1, U )=M_{\text{\tiny OBS}}(z,H_1,U)+\bar{M}_{nd}(z)\label{obstomcont}
    \end{equation}
    where $M_{\text{\tiny OBS}}(z,H_1,U)$ was defined in Eq. \eqref{obscontinuouspovm} and $\bar{M}_{nd}(z)$ is hermitian and satisfy
    \begin{equation}
        \bra{w'(H_1,U),y'}\bar{M}_{nd}(z)\ket{w''(H_1,U),y''}=0,
    \end{equation}
for all $w,w',y,y'$ such that $w',w''\in I_{z}$, where $I_z$ is some interval with no-zero measure $d_z=\int_{I_z}dw>0$. To end the deduction, all we need to prove is that
\begin{equation}
    \bar{M}_{nd}(z)=0, \quad \text{$z$-a.e..}
\end{equation}
To do so, let us consider the opposite, so that there is some interval with nonzero measure $I$ such that
\begin{equation}
    \int_I\bar{M}_{nd}(z)\neq0
\end{equation}
In this case, there must be two normalized vectors $\ket{\psi}$ and $\ket{\phi}$ such that
\begin{equation}
    \bra{\phi}\int_I\bar{M}_{nd}(z)\ket{\psi}\neq 0.\label{mndphipsi}
\end{equation}
Without loss of generality, let us assume that $\psi(w,y)=\braket{w(H_1,U),y|\psi}$ and $\psi(w,y)=\braket{w(H_1,U),y|\psi}$ are nonzero only when $w\in I_2$ and $w\in I_3$, respectively, with $I_2$ and $I_3$ open intervals. Therefore,
\begin{equation}
    \ket{\psi}=\int_{I_2}dw\int_{-\infty}^{\infty}dy \ket{w(H_1,U),y}\psi(w,y)
\end{equation}
and
\begin{equation}
    \ket{\phi}=\int_{I_3}dw\int_{-\infty}^{\infty}dy \ket{w(H_1,U),y}\phi(w,y).
\end{equation}
Considering that $I_2,I_3\subset I$, then since 
\begin{equation}
    \bra{w'(H_1,U),y'}\bar{M}_{nd}(z)\ket{w''(H_1,U),y''}=0
\end{equation}
for any $w',w''\in I$, it cannot be the case that Eq.\eqref{mndphipsi} holds. On the other hand, suppose that $I_2\cap I=I_3\cap I=\emptyset$, then since 
\begin{equation}
    \bra{\psi}\int_{I_2}dzM'(z, H_1, U )\ket{\psi}=\bra{\phi}\int_{I_3}dzM'(z, H_1, U )\ket{\phi}=1,
\end{equation}
it follows that 
\begin{equation}
    \bra{\psi}\int_{I}dzM'(z, H_1, U )\ket{\psi}=\bra{\phi}\int_{I}dzM'(z, H_1, U )\ket{\phi}=0.
\end{equation}
To deduce this, notice that $M'(z, H_1, U )$ are elements of POVM and for any state $\ket{\psi'}$ and interval $I'$,  $\bra{\psi'}\int_{I'}dzM'(z, H_1, U )\ket{\psi'}\geq 0$, $\bra{\psi}\int_{-\infty}^{\infty}dzM'(z, H_1, U )\ket{\psi}=1$. Assuming that Eq.\eqref{mndphipsi} holds, we can define
\begin{equation}
    \Phi_{23}=\bra{\phi}\int_Idz\bar{M}_{nd}(z)\ket{\psi}\neq 0
\end{equation}
so that, by assuming the state
\begin{equation}
    \ket{\psi_{23}}=\frac{a\ket{\psi}+b\ket{\phi}}{|a|^2+|b|^2}
\end{equation}
we get
\begin{equation}
    \bra{\psi_{23}}\int_{I}dzM'(z, H_1, U )\ket{\psi_{23}}=\frac{ab^{*}\Phi_{23}+ba^{*}\Phi_{23}^{*}}{|a|^2+|b|^2}
\end{equation}
Considering $b=\Phi_{23}$ and $a=-1$, we get that
\begin{equation}
    \bra{\psi_{23}}\int_{I}dzM'(z, H_1, U )\ket{\psi_{23}}=-\frac{2|\Phi_{23}|^2}{1+|\Phi_{23}|^2}<0,
\end{equation}
which cannot be the case for a POVM. Therefore, for $I_2\cap I=I_3\cap I=\emptyset$, Eq. \eqref{mndphipsi} cannot hold. 

Now, let us consider the case in which $I_2\subset I$ and $I_3\cap I=\emptyset$. Then, since
\begin{equation}
    \bra{\phi}\int_{I_3}dzM'(z, H_1, U )\ket{\phi}=1,
\end{equation}
and, therefore,
\begin{equation}
    \bra{\phi}\int_{I}dzM'(z, H_1, U )\ket{\phi}=0,
\end{equation}
it follows that we can consider the state in the form $\ket{\psi_{23}}$ once again to obtain
\begin{equation}
    \bra{\psi_{23}}\int_{I}dzM'(z, H_1, U )\ket{\psi_{23}}=\frac{|a|^2+ab^{*}\Phi_{23}+ba^{*}\Phi_{23}^{*}}{|a|^2+|b|^2}
\end{equation}
Considering $a=1$ and $b=\Phi_{23}$, we get
\begin{equation}
     \bra{\psi_{23}}\int_{I}dzM'(z, H_1, U )\ket{\psi_{23}}=\frac{1+2|\Phi_{23}|^2}{1+|\Phi_{23}|^2}>1,
\end{equation}
which cannot hold for a POVM. Therefore, for $I_2\subset I$ and $I_3\cap I=\emptyset$, Eq. \eqref{mndphipsi} cannot hold. A similar argument can be given to deduce that for the case in which $I_3\subset I$ and $I_2\cap I=0$, Eq. \eqref{mndphipsi} cannot hold as well. 

With the previous arguments, we can thus conclude that for any vectors $\ket{\psi}$ and $\ket{\phi}$ with wave functions $\braket{w(H_1,U)|\psi}$ and $\braket{w(H_1,U)|\psi}$ with nonzero values for $w\in I_2$ and $w\in I_3$ such that $I_2\cap I=\emptyset$, $I_2\subset I$, $I_3\cap I=\emptyset$ or $I_3\subset I$, Eq. \eqref{mndphipsi} cannot hold. 

Let us consider the general case in which $I_2$ and $I_3$ are not necessarily subsets of $I$ and $I_2\cap I\neq \emptyset$ and $I_3\cap I\neq \emptyset$. In this case, we can separate $I_{2}'=I_2\cap I$ and $I_{3}'=I_3\cap I$ and their complements $I_4'=I_2\setminus I_{2}^{'}$ and $I_5'=I_3\setminus I_{3}^{'}$. Now, assuming Eq. \eqref{mndphipsi} to hold implies that there is a $k\in\{2,4\}$ and $j\in\{3,5\}$, such that
\begin{equation}
    \ket{\psi_{k}}=\frac{1}{N_k^{'}}\int_{I_k'}dw\int_{-\infty}^{\infty}dy \ket{w(H_1,U)}\psi(w,y),\quad N_k^{'}= \int_{I_k'}dw\int_{-\infty}^{\infty}dy|\psi(w,y)|^2
\end{equation}
and
\begin{equation}
    \ket{\phi_{j}}=\frac{1}{N_j^{'}}\int_{I_j'}dw\int_{-\infty}^{\infty}dy \ket{w(H_1,U)}\phi(w,y),\quad N_j^{'}= \int_{I_j'}dw\int_{-\infty}^{\infty}dy|\phi(w,y)|^2
\end{equation}
and
\begin{equation}
    \bra{\phi_{j}}\int_I\bar{M}_{nd}(z)\ket{\psi_{k}}\neq 0\label{mndphipsi2}
\end{equation}
holds. If Eq. \eqref{mndphipsi2} does not hold for all $k\in\{2,4\}$ and $j\in\{3,5\}$, then Eq.\eqref{mndphipsi} cannot hold. However, notice that here either $I_j',I_k'\cap I=\emptyset$ or $I_j',I_k'\subset I$. For these cases, we already concluded that Eq. \eqref{mndphipsi2} cannot hold in general. As a result, for general $I_2$ and $I_3$ and $I$, it cannot be the case that Eq. \eqref{mndphipsi} holds. Consequently, there is no nonzero measure interval $I$ such that 
\begin{equation}
    \int_Idz\bar{M}_{nd}(z)\neq 0.
\end{equation}
As a result, $\bar{M}_{nd}(z)=0$ $z$-a.e.. Considering this result, together with Eq.\eqref{obstomcont}, we conclude that
\begin{equation}
    M'(z, H_1, U )=M_{\text{\tiny OBS}}(z,H_1,U) \quad \text{$z$-a.e..}\label{mlmobs22continuous}
\end{equation}
Given that we suppose arbitrary $H_1$, $U$, and arbitrary CRIN protocol $\mathbbm{M}'$, then Eq.\eqref{mlmobs22continuous} is valid for any $H_1$, $U$ and $\mathbbm{M}'$. Therefore, there cannot be a CRIN protocol $\mathbbm{M}'$ different then a OBS protocol.

\subsection{Final remarks regarding result 1}
By considering $\Delta(H_1,U)$ to have continuous or discrete, degenerate or non-degenerate spectrum, we proved that there is no CRIN protocol different from the OBS protocol. As a result, for any $\Delta(H_1,U)$ this result is valid and therefore, Result 1 is fully proved.

\section{Appendix C: Complete proof of result 2}
We derive in the present section result 2. As in the scope of result 1, we consider a unitary evolution $U$ with countable basis and an energy operator of interest $H_1$ acting on the same Hilbert space $\mathcal{H}$ as $U$. We denote $\{\ket{\delta_j(H_1,U)}\}$ as the basis that diagonalizes $\Delta(H_1,U)$, which can be discrete, continuous, degenerate or not. Considering this notation, we restate result 2 here for completeness:

\begin{result2} 
For any unitary evolution $U$, energy operator $H_1$ acting on a Hilbert space $\mathcal{H}$, and eigenstate $\ket{\delta_i(H_1, U)}$ of $\Delta(H_1, U)$, there exists a unitary $U'$ acting on an auxiliary Hilbert space $\mathcal{H}'$, an additional Hamiltonian $H_2$ acting on $\mathcal{H} \otimes \mathcal{H}'$, and a vector $\ket{v} \in \mathcal{H}'$ such that the following equations are satisfied: 
\begin{eqnarray} 
&[H_1 \otimes \mathbbm{1}_{\mathcal{H}'} + H_2, U \otimes U'] = 0, \label{commutresult10} \\
&H_2 \ket{\delta_i(H_1, U), v} = E_i \ket{\delta_i(H_1, U), v},\label{result122ndrelation0}\\
&(U^\dagger \otimes U^{'\dagger}) H_2 (U \otimes U') \ket{\delta_i(H_1, U), v} = E_i' \ket{\delta_i(H_1, U), v}, \label{result12ndrelation0} 
\end{eqnarray} 
where $\ket{\delta_i(H_1, U), v} = \ket{\delta_i(H_1, U)} \otimes \ket{v}$, and $E_i$ and $E_i' = E_i - \delta_i(H_1, U)$ are real numbers. 
\end{result2}

To deduce this result, we separated the proof in many steps, considering a sequence of Lemmas. We present these Lemmas and then consider its main proof. 

\subsection{Lemmas necessary for the deduction of Result 2}

Throughout our main text, we assumed $U$ to have a countable basis $\{\ket{u_i}\}$, then it follows that 
\begin{equation}
    U\ket{u_i} = u_i\ket{u_i},\quad U^{\dagger}\ket{u_i} = u_i^{-1}\ket{u_i},\quad u_i = \mathrm{e}^{i\theta_i}\label{relU}
\end{equation} 
hold, where $\theta_i \in \mathbbm{R}$. We define the countable set $\Theta\subset\mathbbm{R}$ containing all the arguments $\theta_m$ defining the eigenvalues $e^{i\theta_m}$ of $U$. Taking into account these definitions, we consider the following lemma:

\begin{lemas}
    For any unitary $U$ with countable basis, with the eigenvalues' arguments on $\Theta$, there exist at least one unitary operator $U'$ acting on a Hilbert space $\mathcal{H}'$ such that:
    \begin{enumerate}
        \item $U'$ is diagonalized by a countable basis $\{\ket{u_i'}\}$
        \begin{equation}
            U'\ket{u_{j}'} = u_{j}'\ket{u_{j}'},\quad U^{'\dagger}\ket{u_{j}'} = u_{j}^{'-1}\ket{u_{j}'},\quad u_{j}' = \mathrm{e}^{i\theta_{j}'}.\label{relUlinha}
        \end{equation} 
        \item The countable set $\Theta'$ containing all $\theta_j'$ satisfying Eq. \eqref{relUlinha} is such that for any $\theta_j'\in \Theta'$, any $\theta_k\in \Theta$, and any $\theta_l\in \Theta$, there exist countably infinitely many $\theta_m'\in \Theta'$ such that
        \begin{equation}
            \theta_m'+\theta_l=\theta_j'+\theta_k \mod 2\pi,
        \end{equation}
        where $\mathrm{mod}\, 2\pi$ denotes the sum modulo $2\pi$.
    \end{enumerate}\label{lemaulinhaexiste}
\end{lemas}
\begin{proof}
    To prove our result, we only need to show that for an arbitrary $U$, there is at least one $U'$ that satisfy the statement of the Lemma. In this sense, for an arbitrary $U$ satisfying Eq. \eqref{relU}, we propose the following unitary:
    \begin{equation}
        U'=\sum_{n=0}^{\infty}\mathrm{e}^{i\theta_{n}'}\ket{n}\bra{n}\label{defulinhafirst}
    \end{equation}
    acting on a Hilbert space $\mathcal{H}$ with countably infinite eigenbasis $\{\ket{n}\}$ (e.g. $\ket{n}$ can be the energy basis of an oscillator). Here, $\theta_n'$ are real variables, each one related with a projector $\ket{n}\bra{n}$. As a direct result of the definition \eqref{defulinhafirst}, the set $\{\theta_n'\}=\Theta'$ of all arguments $\theta_n'$ is countable as well. 
    
    The definition of each $\theta_n'$ is the most crucial so that condition 2 is satisfied. In this direction, we consider a procedure to define $\theta_j'$ and then we prove that by using this procedure, all $\theta_{j}'$ can be defined. For that consider an arbitrary natural $l\geq 1$ and define
    \begin{equation}
        k_l=\left(\sum_{m=1}^{l-1}(2m+1)^m\right)+1.\label{kl}
    \end{equation}
    The $\theta_{p}'$ for $k_l\leq p\leq k_{l+1}-1$ is defined using the following algorithm, which we call the permutation algorithm. 
    
    \vspace{0.1in}
    
    \textbf{Permutation algorithm}
    
    \vspace{0.1in}
    
    Set $k=k_{l}$ and then do the following:
    \begin{itemize}
        \item For $i_1$ from $-l$ to $l$ do:
        \begin{itemize}
            \item For $i_2$ from $-l$ to $l$ do:

            $\ddots$
            \begin{itemize}
                \item For $i_l$ from $-l$ to $l$ do:
                \begin{itemize}
                    \item Define: $\theta_{k}'=\sum_{j=1}^{l}i_{j}\theta_j$, given that $\theta_j\in\Theta$ defined in Eq. \eqref{relU};
                    \item Set $k\to k+1$;
                \end{itemize}
                \item End $i_l$ loop.
            \end{itemize}

            \reflectbox{$\ddots$}
            \item End $i_2$ loop
        \end{itemize}
        \item End $i_1$ loop
    \end{itemize} 
    \vspace{0.1in}
    
    \textbf{End of the Algorithm}
    
    \vspace{0.1in}
    In other words, the set of $\{\theta_p'\}$ defined in this procedure contains all the $k_{l+1}-k_{l}=(2l+1)^l$ combinations of the type
    \begin{equation}
        \sum_{j=1}^{l}m_{j}\theta_j
    \end{equation}
    with integers $m_{j}$ satisfying $|m_{j}|<l$. Now, we show that using the permutation algorithm, we are able to define all $\theta_j'$ in Eq. \eqref{defulinhafirst}. First, notice that by considering $l=1$, then we have, from Eq. \eqref{kl} that $k_1=1$ and, using the permutation algorithm:
    \begin{equation}
        \theta_1'=-\theta_1\,\,\,\,
        \theta_2'=0\,\,\,\,
        \theta_3'=\theta_1\,\,\,\,
    \end{equation}
    Similarly, for $l=2$, it follows that $k_2=4$ and:
    \begin{equation}
        \ba{l}
        \theta_4'=-2\theta_2-2\theta_1,\,\,\,\,
        \theta_5'=-2\theta_2-\theta_1,\,\,\,\,
        \theta_6'=-2\theta_2,\,\,\,\,
        \theta_{7}'=-2\theta_2+\theta_1,\,\,\,\,
        \theta_{8}'=-2\theta_2+2\theta_1,\,\,\,\,
        \theta_{9}'=-\theta_2-2\theta_1,\,\,\,\,\\
        \theta_{10}'=-\theta_2-\theta_1,\,\,\,\,
        \theta_{11}'=-\theta_2,\,\,\,\,
        \theta_{12}'=-\theta_2+\theta_1,\,\,\,\,
        \theta_{13}'=-\theta_2+2\theta_1,\,\,\,\,
        \theta_{14}'=-2\theta_1,\,\,\,\,
        \theta_{15}'=-\theta_1,\,\,\,\,
        \theta_{16}'=0,\,\,\,\,\\
        \theta_{17}'=\theta_1,\,\,\,\,
        \theta_{18}'=2\theta_1,\,\,\,\,
        \theta_{19}'=\theta_2-2\theta_1,\,\,\,\,
        \theta_{20}'=\theta_2-\theta_1,\,\,\,\,
        \theta_{21}'=\theta_2,\,\,\,\,
        \theta_{22}'=\theta_2+\theta_1,\,\,\,\,
        \theta_{23}'=\theta_2+2\theta_1,\,\,\,\,
        \theta_{24}'=2\theta_{2}-2\theta_1,\,\,\,\,\\
        \theta_{25}'=2\theta_{2}-\theta_1,\,\,\,\,
        \theta_{26}'=2\theta_{2},\,\,\,\,
        \theta_{27}'=2\theta_{2}+\theta_1,\,\,\,\,
        \theta_{28}'=2\theta_{2}+2\theta_1,\,\,\,\,
        \ea
    \end{equation}
    For $l=3$, it follows that $k_3=29$ and so on. Now, there are two options for the basis of $U$: either it is countably infinite or finite. Let us first consider that $U$ is diagonalized by an countably infinite basis. Consider that $\theta_{p}'$ is defined for $1\leq p\leq k_{n}-1$, where $n$ is an arbitrary natural. Then, since the permutation algorithm used for the $l=n$, defines $\theta_{p}'$ for all $p$ such that $k_n\leq p\leq k_{n+1}-1$, it follows that $\theta_{p}'$ becomes defined for $1\leq p\leq k_{n+1}-1$. Therefore, by induction, we can repeat the procedure for arbitrary large $l$, and $\theta_{p}'$ is defined for all $p\geq 1$. The set $\Theta'$ containing all $\theta_{p}'$ defined by this induction procedure is thus countably infinite. 
    
    In the case that $U$ is diagonalized by a countably finite basis of dimension $D<\infty$, it follows that $\theta_j$ is defined for $1\leq j\geq D$. We can define an extension $\Theta_e=\Theta\bigcup_{k=1}^\infty \{\theta_{D+k}\}$ of $\Theta$ by adding to it the extra variables $\theta_{D+1}=\theta_1$, $\theta_{D+2}=\theta_2$, $\cdots$, $\theta_{2D}=\theta_D$, $\theta_{2D+1}=\theta_{1}$, $\theta_{2D+2}=\theta_{2}$, and inductively repeat this procedure. In this case, $\theta_j'$ can be defined using the same procedure as in the case in which $U$ is diagonalized by a countably infinite basis, but considering the elements $\Theta_e$ instead of $\Theta$. 

    Now, we prove that all sums of the form
    \begin{equation}
        \sum_{j=1}^{N}m_{j}\theta_j\label{mjthetaj}
    \end{equation}
    with finite natural $N\geq 1$ and finite integers $m_j$, are inside $\Theta'$. To prove this, notice that, because $m_j$ are finite, then there is an integer $s$ such that $|m_j|\leq s$. There are two possibilities: either $s\geq N$ or $s<N$. Let us consider that $s\geq N$. In this case, if we consider the permutation algorithm for $l=s$, then the set all possible combinations
    \begin{equation}
        \sum_{j=1}^{s}i_{j}\theta_j
    \end{equation}
    for $|i_j|\leq s$ will be contained in the set of all $\theta_k'$ defined in the algorithm for $l=s$. Therefore, there is at least one $\theta_k'$ defined in the algorithm for $l=s$ equal to the sum in Eq. \eqref{mjthetaj}: the sum in Eq. \eqref{mjthetaj} is contained in $\Theta'$. Now, consider that $s<N$. If we consider the permutation algorithm for $l=N$, then the set all possible combinations
    \begin{equation}
        \sum_{j=1}^{N}i_{j}\theta_j
    \end{equation}
    for $|i_j|\leq N$ will be contained in the set of all $\theta_k'$ defined in the algorithm for $l=N$. Therefore, there is at least one $\theta_k'$ defined in the algorithm for $l=N$ equal to the sum in Eq. \eqref{mjthetaj}. Since $N$ is arbitrary as well as $m_j$ then every combination as in Eq. \eqref{mjthetaj} will be contained in $\Theta'$. 

    It is also straightforward to show (just by looking at the permutation algorithm) that all $\theta_j'\in\Theta'$ are of the form of Eq. \eqref{mjthetaj}. As a result, all terms inside $\Theta'$ are of the form of Eq. \eqref{mjthetaj} and all the terms of the form of Eq. \eqref{mjthetaj} are inside $\Theta'$.
    
    To prove condition 2, select arbitrary $\theta_j'\in \Theta'$, $\theta_k\in \Theta$, and $\theta_{q}\in \Theta$. Because $\theta_j'$ is inside $ \Theta'$, then it can be written as
    \begin{equation}
        \theta_j'=\sum_{n=1}^{N}m_{n}\theta_n
    \end{equation}
    where $N$, and $|m_{n}|$ are finite. As a result, 
    \begin{equation}
        \theta_j'+\theta_k-\theta_q=\sum_{n=1}^{N}m_{n}\theta_n+\theta_k-\theta_q
    \end{equation}
    Now, there are infinitely many $\theta_{m}$ such that 
    \begin{equation}
       \theta_{m}=\theta_j'+\theta_k-\theta_q.
    \end{equation}
To demonstrate this, let us define $l_{\max} = \max\{ N, k, q\}$. For any finite natural $l > l_{\max}$, it is straightforward to check that, by applying the permutation algorithm for $l$, the element $\theta_j' + \theta_k - \theta_q$ is defined. In other words, for every $l > l_{\max}$, there exists at least one $m$ such that $k_l \leq m \leq k_{l+1}$ and  
\begin{equation}
    \theta_m = \theta_j' + \theta_k - \theta_q.
\end{equation}
Since the number of $l$'s such that $l> l_{\max}$ is countably infinite, then the set of all such  $\theta_m$ is also countably infinite. Given that $j$, $k$, and $n$ are arbitrary, this establishes condition 2 of the lemma.
\end{proof}

We call the conditions 1 and 2 of the Lemma as condition $\mathbf{U}$ from now on. Moreover, from this point on, we consider $\Theta$ and $\Theta'$ and $U$ and $U'$ as the sets and operators defined in Eq. \eqref{relU}, \eqref{relUlinha}, and Lemma \ref{lemaulinhaexiste}. Now, we prove the final two Lemmas needed for the proof of result 1.
\begin{lemas}
    Consider a given $H_1$ and $U$ acting on a Hilbert space $\mathcal{H}$, such that the countable basis $\{\ket{u_i}\}$ of $U$ satisfy Eq. \eqref{relU}. Also, suppose that there is an hermitian operator $H_2$ and an unitary $U'$ acting on a Hilbert space $\mathcal{H}'$, and that the eigenvalues $u_j'=\mathrm{e}^{i\theta_j'}$ of $U'$ satisfy Eq. \eqref{relUlinha}. Moreover, suppose that 
    \begin{equation}
        0=\bra{u_m,u_j'}[H_1\otimes\mathbbm{1}_{\mathcal{H}'}+H_2,U\otimes U']\ket{u_k,u_l'}=(u_ku_l'-u_mu_j')(\bra{u_m}H_1\ket{u_k}\delta_{jl}+\bra{u_m,u_j'}H_2\ket{u_k,u_l'})\label{lemmabasico}
    \end{equation}
    for every $m,j,k,l$ such that
    \begin{equation}
            k=m\,\, \&\,\, l>j\quad \text{or}\quad k>m \,\, \& \,\, \text{any}\,\, l.\label{conditionsindex}
    \end{equation} Then it follows that
    \begin{equation}
    [H_1\otimes\mathbbm{1}_{\mathcal{H}'}+H_2,U\otimes U']=0.
    \end{equation}
    \label{lemacomponents}
\end{lemas}
\begin{proof}
    To prove this result, all we need is to show that for $m,j,k,l$ \emph{not} satisfying Eq. \eqref{conditionsindex}, it can still be true that $0=\bra{u_m,u_j'}[H_1\otimes\mathbbm{1}_{\mathcal{H}'}+H_2,U\otimes U']\ket{u_k,u_l'}$. For that, we first notice that if Eq. \eqref{lemmabasico} is satisfied for $m,j,k,l$ satisfying Eq. \eqref{conditionsindex}, then it follows that 
    \begin{equation}
        (u_ku_l'-u_mu_j')(\bra{u_m}H_1\ket{u_k}\delta_{jl}+\bra{u_m,u_j'}H_2\ket{u_k,u_l'})=0
    \end{equation}
    which means that either (i) $u_ku_l'-u_mu_j'=\mathrm{e}^{i\theta_k+\theta_l'}-\mathrm{e}^{i\theta_m+\theta_j'}=0$, or that (ii) $\bra{u_m}H_1\ket{u_k}\delta_{jl}+\bra{u_m,u_j'}H_2\ket{u_k,u_l'}=0$, or both (i) and (ii) are simultaneously satisfied. If (i) is satisfied, then we can deduce that
    \begin{equation}
        \bra{u_k,u_l'}[H_1\otimes\mathbbm{1}_{\mathcal{H}'}+H_2,U\otimes U']\ket{u_m,u_j'}=(u_mu_j'-u_ku_l')(\bra{u_k}H_1\ket{u_m}\delta_{jl}+\bra{u_k,u_l'}H_2\ket{u_m,u_j'})=0.\label{lemmabasico2}
    \end{equation}
    The same holds when (i) and (ii) are valid simultaneously. If, on the other hand, (i) is not valid, and (ii) holds, then $\bra{u_m}H_1\ket{u_k}\delta_{jl}=-\bra{u_m,u_j'}H_2\ket{u_k,u_l'}$ and, given that $H_1$ and $H_2$ are hermitian, it follows that $\bra{u_k}H_1\ket{u_m}^*\delta_{jl}=-\bra{u_k,u_l'}H_2\ket{u_m,u_j'}^*$, which also implies in $\bra{u_k}H_1\ket{u_m}\delta_{jl}=-\bra{u_k,u_l'}H_2\ket{u_m,u_j'}$. As a result, it implies once again Eq. \eqref{lemmabasico2}. Therefore, for any $m,j,k,l$ such that $\bra{u_m,u_j'}[H_1\otimes\mathbbm{1}_{\mathcal{H}'}+H_2,U\otimes U']\ket{u_k,u_l'}=0$ it follows that $\bra{u_k,u_l'}[H_1\otimes\mathbbm{1}_{\mathcal{H}'}+H_2,U\otimes U']\ket{u_m,u_j'}=0$. As a result, it follows that $\bra{u_m,u_j'}[H_1\otimes\mathbbm{1}_{\mathcal{H}'}+H_2,U\otimes U']\ket{u_k,u_l'}=0$ for any $m,j,k,l$ satisfying 
    \begin{equation}
        \ba{lll}
            k=m & \& & j>l\\
            k=m & \& & j<l\\
            k<m & \& & \text{any}\,\, l\\
            k>m & \& & \text{any}\,\, l.
        \ea\label{conditionsindex2}
    \end{equation}
    The only case not considered to state that $\bra{u_m,u_j'}[H_1\otimes\mathbbm{1}_{\mathcal{H}'}+H_2,U\otimes U']\ket{u_k,u_l'}=0$ is the case where $k=m$ and $j=l$. In this case, it is easy to show that $(u_m u_j'-u_k u_l')=(u_m u_j'-u_m u_j')=0$. We can thus concluded that $\bra{u_k,u_l'}[H_1\otimes\mathbbm{1}_{\mathcal{H}'}+H_2,U\otimes U']\ket{u_m,u_j'}=0$ for every $m,j,k,l$, which results in $[H_1\otimes\mathbbm{1}_{\mathcal{H}'}+H_2,U\otimes U']=0$. 
\end{proof}
\begin{lemas}
    Consider an arbitrary discrete and countable basis $\{\ket{n}\}$ of an arbitrary discrete Hilbert space $\mathcal{H}$. We can always define a normalized vector $\ket{v}\in \mathcal{H}$ such that $\braket{v|v}=1$ and the real and imaginary parts are non-null, $\Re(\braket{n|v})\neq 0$ $\Im(\braket{n|v})\neq 0$, for any of its components $\braket{n|v}$.
    \label{lemmav}
\end{lemas}
\begin{proof}
    For a finite dimensional case, the result of the Lemma can be straightforwardly checked. Let us thus consider the infinite scenario. Since the basis is countable, then we can always assume a natural number $m$ to identify the $m$-th element $\ket{m}$ of the basis. Then, considering that $m\geq 0$, we can define $\braket{m|v}$ as 
    \begin{equation}
        \braket{m|v}=\mathrm{e}^{i\eta_{m}}\mathrm{e}^{-|\alpha|^2/2}\frac{(\alpha)^{m}}{\sqrt{(m)!}},
    \end{equation}
    where $\alpha\neq 0$ is a complex number and $\eta_{n}$ is real. Notice that this state has the exactly same form as the components of the coherent state in the quantum oscillator for a complex $\alpha$, up to local phases $\eta_{m}$. Moreover, $\ket{v}$ is normalized, since
    \begin{equation}
        \braket{v|v}=\bra{v}\left(\sum_{m=0}^{\infty}\ket{m}\bra{m}\right)\ket{v}=\mathrm{e}^{-|\alpha|^2}\sum_{m=0}^{\infty}\frac{|\alpha|^{2m}}{m!}=1.
    \end{equation}
    Now, to show that $\Re(\braket{m|v}),\Im(\braket{m|v})\neq 0$, first, notice that $|\braket{m|v}|\neq 0$ for all $m$, given that $\alpha\neq 0$. Now, we can define $\eta_{m}$ in the following form 
    \begin{equation}
        \ba{ll}
        \eta_{m}=0&\text{if}\,\, \Re\left(\alpha^{m}\right)\neq 0\,\, \&\,\, \Im\left(\alpha^{m}\right)\neq 0\\
        \eta_{m}=\pi/4 &\text{else}.
        \ea\label{etamdef}
    \end{equation}
    Since $|\braket{m|v}|\neq 0$ for any natural $m$, then Eq. \eqref{etamdef} ensure that $\Re(\braket{m|v}),\Im(\braket{m|v})\neq 0$.
\end{proof}
\subsection{The main proof of Result 2}
Given all the Lemmas in the previous subsection, we are now in position to prove the result 2.


\begin{proof}
     First, we consider a unitary $U'$ that satisfies the $\mathbf{U}$ condition of Lemma \ref{lemaulinhaexiste}. Therefore, it acts on a space $\mathcal{H}'$ that satisfy Eq. \eqref{relUlinha}, whose set $\Theta'=\{\theta_{j}'\}$ of the arguments $\theta_{j}'$ of its eigenvalues $\mathrm{e}^{i\theta_{j}'}$. Because of Lemma \ref{lemaulinhaexiste}, we know that such $U'$ can always be defined. We thus consider the components $[H_1\otimes\mathbbm{1}_{\mathcal{H}'}+H_2,U\otimes U']$, in the basis $\ket{u_i,u_k'}$ that diagonalizes $U\otimes U'$ as follows: 
    \begin{equation}
        \bra{u_m,u_j}[H_1\otimes\mathbbm{1}_{\mathcal{H}'}+H_2,U\otimes U']\ket{u_k,u_l}=(\mathrm{e}^{i(\theta_k+\theta_l')}-\mathrm{e}^{i(\theta_m+\theta_j')})(\bra{u_m}H_1\ket{u_k}\delta_{jl}+\bra{u_m,u_j'}H_2\ket{u_k,u_l'})
    \end{equation}
    To prove our result, we want to define $H_2$ so as to make this expression to have a null value. In this sense, we recall that, because $U'$ satisfy the $\mathbf{U}$ condition, then for any $\theta_j'\in \Theta'$ and for any $m$ and $k$, there exist infinitely many $l$ such that $\theta_l'\in \Theta'$ and
    \begin{equation}
        \theta_m+ \theta_j'=\theta_k+ \theta_l' \mod 2\pi.\label{condicao2noresult2}
    \end{equation} 
    As a result, it follows that for any $m$, $k$ and $j$, there exist infinitely many $l$ such that $(\mathrm{e}^{i(\theta_k+\theta_l')}-\mathrm{e}^{i(\theta_m+\theta_j')})=0$. Similarly, for any $m$, $k$ and $j$, there exist infinitely many $l$ such that $(\mathrm{e}^{i(\theta_k+\theta_j')}-\mathrm{e}^{i(\theta_m+\theta_l')})=0$. For the cases in which $(\mathrm{e}^{i(\theta_k+\theta_l')}-\mathrm{e}^{i(\theta_m+\theta_j')})=0$, we can consider any arbitrary value for $\bra{u_m,u_j'}H_2\ket{u_k,u_l'}$ and still the following relation holds:
    \begin{equation}
        \bra{u_m,u_j'}[H_1\otimes\mathbbm{1}_{\mathcal{H}'}+H_2,U\otimes U']\ket{u_k,u_l'}=\left(\mathrm{e}^{i(\theta_k+\theta_l')}-\mathrm{e}^{i(\theta_m+\theta_j')}\right)(\bra{u_m}H_1\ket{u_k}\delta_{jl}+\bra{u_m,u_j'}H_2\ket{u_k,u_l'})=0.\label{resultadoparcial1}
    \end{equation}
    As a consequence, we can define the free complex variables $h_{mjkl}$ such that $h_{klmj}=h_{mjkl}^*$, so that the components of $H_2$ are defined as
    \begin{equation}
        \ba{lll}
        \bra{u_m,u_j'}H_2\ket{u_k,u_l'}=h_{mjkl}&\text{if}&(\mathrm{e}^{i(\theta_k+\theta_l')}-\mathrm{e}^{i(\theta_m+\theta_j')})=0\\
        \bra{u_m,u_j'}H_2\ket{u_k,u_l'}=-\bra{u_m}H_1\ket{u_k}\delta_{jl}&\text{if}&(\mathrm{e}^{i(\theta_k+\theta_l')}-\mathrm{e}^{i(\theta_m+\theta_j')})\neq 0
        \ea
        \label{componentsHY} 
    \end{equation}
    for the indexes
    \begin{equation}
        k=m\,\, \&\,\, l\geq j\quad \text{or}\quad k>m \,\, \& \,\, \text{any}\,\, l,
    \end{equation}
    and, by taking into account that $H_1$ is hermitian so that $\bra{u_m}H_1\ket{u_k}=\bra{u_k}H_1\ket{u_m}^*$,
    \begin{equation}
        \bra{u_k,u_l'}H_2\ket{u_m,u_j'}=\bra{u_m,u_j'}H_2\ket{u_k,u_l'}^*\label{conditionsindex4}
    \end{equation}
    for the indexes
    \begin{equation}
        k=m\,\, \&\,\, l< j\quad \text{or}\quad k<m \,\, \& \,\, \text{any}\,\, l.
    \end{equation}
    Considering Eq. \eqref{componentsHY}, we completely characterize the components $\bra{u_m,u_j'}H_2\ket{u_k,u_l'}$ of $H_2$ such as to ensure that $H_2$ is hermitian and satisfy 
    \begin{equation}
         0= \bra{u_m,u_j'}[H_1\otimes\mathbbm{1}_{\mathcal{H}'}+H_2,U\otimes U']\ket{u_k,u_l'}=0 
    \end{equation}
    for the indexes $k=m$ $\&$ $l\geq j$ or $k>m$ $\&$ any $l$. Considering Lemma \ref{lemacomponents}, this means that by defining $H_2$ as in Eqs. \eqref{componentsHY} and \eqref{conditionsindex4}, it follows that
    \begin{equation}
        [H_1\otimes\mathbbm{1}_{\mathcal{H}'}+H_2,U\otimes U']=0,\label{commutresult10smproof}
    \end{equation}
    so as to satisfy the first equation of result 2, \emph{viz}. Eq. \eqref{commutresult10}. We assume from this point on that $H_2$ is  defined by Eqs. \eqref{componentsHY} and \eqref{conditionsindex4}.
    
    Consider a state $\ket{v}\in\mathcal{H}'$ whose only requirement is that its components $\braket{u_j'|v}$ in the basis $\ket{u_j'}$ satisfy $\Re(\braket{u_j'|v})\neq 0$ and $\Im(\braket{u_j'|v})\neq 0$, i.e. the real and imaginary part of $\braket{u_j'|v}$ are non-null. From Lemma \ref{lemmav} and the fact that the basis $\{\ket{u_j'}\}$ is countable, we can always define such $\ket{v}$. Considering this $\ket{v}$, we can thus analyze how to define the free variables $h_{klmj}=h_{mjkl}^*$ in \eqref{componentsHY} so that the eigenvalue/eigenvector equation 
    \begin{equation}
           H_2\ket{\delta_i(H_1,U)}\otimes\ket{v}=E_i\ket{\delta_i(H_1,U)}\otimes\ket{v}\equiv E_i\ket{\delta_i(H_1,U),v}\label{eigenvvHy}
    \end{equation} 
    can be satisfied. For that, let us consider the components of \eqref{eigenvvHy} with respect to the specific element of the basis $\ket{u_m,u_j'}$:
    \begin{equation}
        \bra{u_m,u_j'}E_i\ket{\delta_i(H_1,U),v}=E_{i}\alpha_{mi}\beta_j=\bra{u_m,u_j'}H_2\ket{\delta_i(H_1,U),v}=\sum_{k,l}\bra{u_m,u_j'}H_2\ket{u_k,u_l'}\alpha_{ki}\beta_l\label{eigenvalueeq2}
    \end{equation}
    where we named $\alpha_{mi}=\braket{u_m|\delta_i(H_1,U)}$ and $\beta_j=\braket{u_j'|v}$. Consider the set $\mathbbm{N}_{mj}$ of all duple $\{k,l\}\in \mathbbm{N}_{mj}$ that, for the given component $\ket{u_m,u_j'}$ satisfy $(\mathrm{e}^{i(\theta_k+\theta_l')}-\mathrm{e}^{i(\theta_m+\theta_j')})= 0$. We can thus rewrite Eq. \eqref{eigenvalueeq2}, considering Eqs. \eqref{componentsHY} and \eqref{conditionsindex4}, as 
    \begin{equation}
        E_{i}\alpha_{mi}\beta_j=-\sum_{\{k,l\}\notin \mathbbm{N}_{mj}}\bra{u_m}H_1\ket{u_k}\delta^{\text{\tiny \textbf{K}}}_{jl}\alpha_{ki}\beta_l+\sum_{\{k,l\}\in \mathbbm{N}_{mj}}h_{mjkl}\alpha_{ki}\beta_l.
        \label{eigenvalueeq3}
    \end{equation}
    Therefore, if we expect that $H_2$ satisfy Eq. \eqref{eigenvvHy}, then we need to properly define the free variables $h_{mjkl}$ so that Eq. \eqref{eigenvalueeq3} holds. To do so, we first notice that there are only two possible alternatives: either $\alpha_{mi}=0$ and $\alpha_{mi}\neq 0$. We treat each of them separately. 
    
    In the case in which $\alpha_{mi}=0$, for Eq. \eqref{eigenvalueeq3} for $m$ and any arbitrary $j$, it is necessary that
    \begin{equation}
        \sum_{\{k,l\}\in \mathbbm{N}_{mj}}h_{mjkl}\alpha_{ki}\beta_l=\sum_{\{k,l\}\notin \mathbbm{N}_{mj}}\bra{u_m}H_1\ket{u_k}\delta^{\text{\tiny \textbf{K}}}_{jl}\alpha_{ki}\beta_l
        \label{eigenvalueeq4}
    \end{equation}
    holds. If $\ket{\delta_i(H_1,U)}$ is a null vector so that $\alpha_{ki}=0$ for all $k\neq m$, then Eq. \eqref{eigenvvHy} holds directly for any value of $h_{mjkl}$. If, on the other hand, $\ket{\delta_i(H_1,U)}$ is not a null vector, then there must be at least one $k\neq m$ such that $\alpha_{ki}\neq 0$. Moreover, since $\theta_j'\in \Theta'$, then, from \eqref{condicao2noresult2}, for any $m$, $k$ and $j$ we can find infinitely many $l$ such that $\theta_l'\in \Theta'$ and $\theta_k+ \theta_l'=\theta_m+\theta_j'\mod 2\pi$, i.e. such that $\{k,l\}\in \mathbbm{N}_{mj}$. Therefore, there must exist a $\{k,l\}\in \mathbbm{N}_{mj}$ such that $\alpha_{ki}\neq 0$. Furthermore, since, by the definition of $\ket{v}$, $\beta_l\neq 0$ for all $l$, then for the specific $m$ and $j$ in Eq. \eqref{eigenvalueeq4} such that $\alpha_{mi}=0$, there exist $\{k,l\}\in \mathbbm{N}_{mj}$ and $k\neq m$, such that  $\alpha_{ki}\beta_l\neq 0$. As a result, the terms $h_{mjkl}$, being complex variables, can \emph{always} be tuned to satisfy the Eq. \eqref{eigenvalueeq4} for all $m$ and $j$ such that $\alpha_{mi}=0$. To show this, all we need is to show that there is at least one way of defining $h_{mjkl}$ that make Eq. \eqref{eigenvalueeq4} to hold. In this sense, consider the set $\mathbbm{N}_{mj}^{\neq}$ of all the elements $\{k,l\}\in \mathbbm{N}_{mj}$ such that $\alpha_{ki}\neq 0$. Moreover, consider the indexes $\bar{k}$ and $\bar{l}$ as the smaller indexes in which $\bar{k}+\bar{l}=\min_{\{k,l\}\in \mathbbm{N}_{mj}^{\neq}}k+l$. If there exist more then two duples $\{\bar{k}',\bar{l}'\}$ and $\{\bar{k}'',\bar{l}''\}$ such that $\bar{k}'+\bar{l}'=\bar{k}''+\bar{l}''=\min_{\{k,l\}\in \mathbbm{N}_{mj}^{\neq}}k+l$, then we denote $\{\bar{k},\bar{l}\}$ as the duple with smallest from $\bar{k}'$ and $\bar{k}''$. Let us define $h_{mjkl}$ for all $m$ and $j$ such that $\alpha_{mi}=0$ as follows:
    \begin{equation}
        \ba{rll}
        h_{mj\bar{k}\bar{l}}&=h_{\bar{k}\bar{l}mj}^*=\frac{1}{\alpha_{\bar{k}i}\beta_{\bar{l}}} \sum_{\{k,l\}\notin \mathbbm{N}_{mj}}\bra{u_m}H_1\ket{u_k}\delta^{\text{\tiny \textbf{K}}}_{jl}\alpha_{ki}\beta_l &\text{for $\{\bar{k},\bar{l}\}$}\\
        h_{klmj}^*&=0 &\text{for all $\{k,l\}\in \mathbbm{N}_{mj}$, such that $\{k,l\}\neq \{\bar{k},\bar{l}\}$.}
        \ea\label{hmi0}
    \end{equation}  
    Defining $h_{mjkl}$ in this way, Eq. \eqref{eigenvalueeq4} will hold for any $m$ and $j$ such that $\alpha_{mi}=0$. Importantly, notice that these definitions will not alter the other components of Eq. \eqref{eigenvvHy}. To check that, we consider any other components $\ket{u_q,u_r'}$ of Eq. \eqref{eigenvvHy}, such that $\{q,r\}\neq \{m,j\}$:
    \begin{equation}
        \bra{u_q,u_r'}H_2\ket{\delta_j(H_1,U)}\otimes\ket{v}=E_{i}\alpha_{qi}\beta_r=-\sum_{sn\notin \mathbbm{N}_{qr}}\bra{u_q}H_1\ket{u_s}\delta^{\text{\tiny \textbf{K}}}_{rn}\alpha_{si}\beta_n+\sum_{sn\in\mathbbm{N}_{qr}}h_{qrsn}\alpha_{si}\beta_n
        \label{eigenvalueeq5}
    \end{equation}
    The only way that Eq. \eqref{eigenvalueeq5} can be dependent on the terms $h_{mjkl}$ for the specific $m$, $j$ of Eq. \eqref{eigenvalueeq4} or their complex conjugate $h_{mjkl}^*=h_{klmj}$ is either if, for the last sum in Eq. \eqref{eigenvalueeq5}, we set (i) $q=m$, $r=j$, $s=k$, and $n=l$ or (ii) $q=k$, $r=l$, $s=m$, and $n=j$. Since $\{q,r\}\neq \{m,j\}$ then the case (i) is already not valid. On the other hand, in the case in which we consider (ii), then $h_{qrsn}\alpha_{si}\to h_{klmj}\alpha_{mi}$ that will result in $0$ for any $h_{klmj}$, since $\alpha_{mi}=0$. So that, Eq. \eqref{eigenvalueeq5} will be independent of $h_{mjkl}$ and $h_{mjkl}^*=h_{klmj}$. Since this holds for any $\{q,r\}\neq\{m,j\}$, then we proved that for any component $\ket{u_m,u_j'}$ of Eq. \eqref{eigenvvHy} in which $\alpha_{mi}= 0$, we can define the independent variables $h_{mjkl}$ as in Eq. \eqref{hmi0} so as to make 
    \begin{equation}
        \bra{u_{m},u_{j}'}H_2\ket{\delta_i(H_1,U),v}=E_i\alpha_{mi}\beta_{j}=E_i\braket{u_{m},u_{j}'|\delta_i(H_1,U),v}=0\quad \text{for $\alpha_{mi}=0$}
        \label{eigenvalueeq}
\end{equation}
    to hold, in a way that does not interfere in the other components $\ket{u_q,u_r'}\neq \ket{u_m,u_j'}$ of Eq. \eqref{eigenvvHy}. We thus assume from now on $h_{mjkl}$ as in Eq. \eqref{hmi0} for all $m$ such that $\alpha_{mi}=0$.

    Now, let us consider the case in which $\alpha_{mi}\neq 0$. In this case, we rewrite Eq. \eqref{eigenvalueeq3} as
    \begin{equation}
        \ba{rl}
        \displaystyle E_{i}\alpha_{mi}\beta_j&\displaystyle=-\sum_{k\neq m,\{k,l\}\notin \mathbbm{N}_{mj}}\bra{u_m}H_1\ket{u_k}\delta^{\text{\tiny \textbf{K}}}_{jl}\alpha_{ki}\beta_l+\sum_{k\neq m,\{k,l\}\in \mathbbm{N}_{mj}}h_{mjkl}\alpha_{ki}\beta_l-\sum_{l\notin \mathbbm{N}_{j}'}\bra{u_m}H_1\ket{u_m}\delta^{\text{\tiny \textbf{K}}}_{jl}\alpha_{mi}\beta_l+\sum_{l\in \mathbbm{N}_{j}'}h_{mjml}\alpha_{mi}\beta_l\\
        &\displaystyle=-\sum_{k\neq m,\{k,l\}\notin \mathbbm{N}_{mj}}\bra{u_m}H_1\ket{u_k}\delta^{\text{\tiny \textbf{K}}}_{jl}\alpha_{ki}\beta_l+\sum_{k\neq m,\{k,l\}\in \mathbbm{N}_{mj}}h_{mjkl}\alpha_{ki}\beta_l+\sum_{l\in \mathbbm{N}_{j}'}h_{mjml}\alpha_{mi}\beta_l.
        \ea
        \label{eigenvalueeq6}
    \end{equation}
    where we denoted $\mathbbm{N}_{j}'$ as the set of indexes $\{l\}$ such that $\theta_{l}'=\theta_j'$ modulo $2\pi$. Therefore, $l\in \mathbbm{N}_{j}'$ and $l\notin \mathbbm{N}_{j}'$ means that $\mathrm{e}^{i\theta_l'}=\mathrm{e}^{i\theta_j'}$ and $\mathrm{e}^{i\theta_l'}\neq\mathrm{e}^{i\theta_j'}$, respectively. Our goal, just as in the case in which $\alpha_{mi}=0$, is to show that for each $m$ and $j$ there is at least one independent complex variable that can be tuned. Here, however, we have a difficulty since although $h_{mjmj}$ is an independent variable, it must be \emph{real} since $H_2$ is hermitian.  Moreover, the other terms $h_{mjml}$ or $h_{mjkl}$ for $\{k,l\}\neq \{m,j\}$ cannot be considered as independent, since they can interfere in other components. 

    Our way to solve these problems is to consider an induction process. First, we consider 
    \begin{equation}
         h_{mjml}=E_i\delta^{\text{\tiny \textbf{K}}}_{jl} \label{definicaohfundamentaliguaistotal}
    \end{equation}
for any $m$ such that $\alpha_{mi}\neq 0$ and any $j$ and $l$.
Moreover, consider that we order all indexes (because they are countable, we can always do this) and assume $m(1)$ as the first index in which $\alpha_{m(1)i}\neq 0$. Consider the first index $j=1$, such that Eq. \eqref{eigenvalueeq6} can be rewritten, considering Eq. \eqref{definicaohfundamentaliguaistotal}, as:

    \begin{equation}
        \ba{rl}
        \displaystyle E_{i}\alpha_{m(1)i}\beta_1&\displaystyle=-\sum_{k\neq m(1),\{k\}\notin \mathbbm{N}_{m(1)}''}\bra{u_{m(1)}}H_1\ket{u_k}\alpha_{ki}\beta_{1}+\sum_{k\neq m(1),\{k,l\}\in \mathbbm{N}_{m(1)1}}h_{m(1)1kl}\alpha_{ki}\beta_l+E_{i}\alpha_{m(1)i}\beta_1.
        \ea
        \label{eigenvalueeq7}
    \end{equation}
where $\mathbbm{N}_{m(1)}''$ is the set of indexes $\{n\}$ such that $\theta_{n}=\theta_{m(1)}\mod 2\pi$. Now, let us consider the first $k'$ such that $\theta_{k'}\notin \mathbbm{N}_{m(1)}''$ and $\alpha_{k'i}\neq 0$. By the definition of $U'$ and of $\Theta'$ as satisfying condition $\mathbf{U}$, for all $m(i)$, and $k'$, there exist infinitely many $l$ such that $\theta_{l}'+\theta_{k'}=\theta_{1}'+\theta_{m(1)}$, i.e. there are infinitely many $l$ such that $\{k',l\}\in \mathbbm{N}_{m(1)1}$. Let us define the smallest $l$ such that $\{k',l\}\in \mathbbm{N}_{m(1)1}$ as $l(k')$. Moreover, from now on, let us denote, for any $j$, $\beta_{j}=z_{j}\mathrm{e}^{i\gamma_j}$, where by the definition of $\beta_j$, $z_{j}>0$ is the real modulus and $\gamma_j\neq 0,\pi/2$ mod $2\pi$, is the argument. We then define the coefficients $h_{m(1)1k'l}$, with $\{k',l\}\in \mathbbm{N}_{m(1)1}$ as follows:
\begin{equation}
    \begin{array}{ccc}
         h_{m(1)1k'l}&=\bra{u_{m(1)}}H_1\ket{u_{k'}}\mathrm{e}^{i(\gamma_{1}-\gamma_l)}\frac{z_{1}}{z_{l}} & \text{if $l=l(k')$, $\{k',l\}\in \mathbbm{N}_{m(1)1}$, and $k'\notin \mathbbm{N}_{m(1)}'$} \\
         h_{m(1)1k'l}&=0 & \text{if $l\neq l(k')$ and $\{k',l\}\in \mathbbm{N}_{m(1)1}$, or $k'\in \mathbbm{N}_{m(1)}'$}
         \label{definicaohfundamental}
    \end{array}
\end{equation}
Notice that, by definition, it follows that 
\begin{equation}
    \begin{array}{ccc}
         h_{k'lm(1)1}&=\bra{u_{k'}}H_1\ket{u_{m(1)}}\mathrm{e}^{i(\gamma_l-\gamma_{1})}\frac{z_{1}}{z_{l}} & \text{if $l=l(k')$, $\{k',l\}\in \mathbbm{N}_{m(1)1}$, and $k'\notin \mathbbm{N}_{m(1)}'$} \\
         h_{k'lm(1)1}&=0 & \text{if $l\neq l(k')$ and $\{k',l\}\in \mathbbm{N}_{m(1)1}$, or $k'\in \mathbbm{N}_{m(1)}'$}
    \end{array}
\end{equation}
By defining $h_{m(1)1k'l}$ for $\{k',l\}\in \mathbbm{N}_{m(1)1}$ as in Eq. \eqref{definicaohfundamental}, we have that
\begin{equation}
    \begin{array}{ccc}
         h_{m(1)1k'l}\alpha_{k'i}\beta_l&=\bra{u_{m(1)}}H_1\ket{u_k'}\alpha_{k'i}\beta_{1} & \text{if $l=l(k')$, $\{k',l\}\in \mathbbm{N}_{m(1)1}$, and $k'\notin \mathbbm{N}_{m(1)}'$} \\
         h_{m(1)1k'l}\alpha_{k'i}\beta_l&=0 & \text{if $l\neq l(k')$ and $\{k',l\}\in \mathbbm{N}_{m(1)1}$, or $k'\in \mathbbm{N}_{m(1)}'$}\label{definicaohfundamental2}
    \end{array}
\end{equation}
Inserting this expression in Eq. \eqref{eigenvalueeq7}, we have that
\begin{equation}
        \ba{rl}
        \displaystyle E_{i}\alpha_{m(1)i}\beta_1&\displaystyle=-\sum_{k\neq m(1),k\neq k',\{k\}\notin \mathbbm{N}_{m(1)}'}\bra{u_{m(1)}}H_1\ket{u_k}\alpha_{ki}\beta_{1}+\sum_{k\neq m(1),k\neq k',\{k,l\}\in \mathbbm{N}_{m(1)1}}h_{m(1)1kl}\alpha_{ki}\beta_l+E_{i}\alpha_{m(1)i}\beta_1.
        \ea
        \label{eigenvalueeq8}
    \end{equation}
In order words, we eliminated the first term $k'$ of the first and second summation. Since the definition of $h_{m(1)1k'l}$ does not alters the definition of  $h_{m(1)1kl}$ for other $k\neq m(1)$, then we can consider a similar definition for all the other $k\neq m(1)$ in Eq.\eqref{eigenvalueeq8}, as follows:
\begin{equation}
    \begin{array}{ccc}
         h_{m(1)1kl}&=\bra{u_{m(1)}}H_1\ket{u_k}\mathrm{e}^{i(\gamma_{1}-\gamma_l)}\frac{z_{1}}{z_{l}} & \text{if $l=l(k)$, $\{k,l\}\in \mathbbm{N}_{m(1)1}$, $k\neq m(1)$, $\alpha_{ki}\neq 0$, $k\notin \mathbbm{N}_{m(1)}'$} \\
         h_{m(1)1kl}&=0 & \text{if $l\neq l(k)$, $\{k,l\}\in \mathbbm{N}_{m(1)1}$, $k\neq m(1)$, and $\alpha_{ki}\neq 0$, or $k\in \mathbbm{N}_{m(1)}'$}\label{definicaohfundamental3}
    \end{array}
\end{equation}
where $l(k)$ is the smallest index $l$ in which $\{k,l\}\in \mathbbm{N}_{m(1)1}$. Using the definition in Eq. \eqref{definicaohfundamental3} for all $k\neq m$ with $\alpha_{ki}\neq 0$ in Eq. \eqref{eigenvalueeq8}, we then obtain
\begin{equation}
        \bra{u_{m(1)},u_{1}'}H_2\ket{\delta_i(H_1,U),v}=E_{i}\alpha_{m(1)i}\beta_{1}
        \label{eigenvalueeq8}
\end{equation}
as we aimed.

Now, let us consider the second smallest index $m(2)$ such that $\alpha_{m(2)i}\neq 0$. We then have the component $\ket{u_{m(2)},u_{1}}$ of Eq. \eqref{eigenvvHy}, considering definition \eqref{definicaohfundamentaliguaistotal} can be written as
\begin{equation}
        \ba{rl}
        \displaystyle E_{i}\alpha_{m(2)i}\beta_1&\displaystyle=-\sum_{k\neq m(2),\{k\}\notin \mathbbm{N}_{m(2)}'}\bra{u_{m(2)}}H_1\ket{u_k}\alpha_{ki}\beta_1+\sum_{k\neq m(2),\{k,l\}\in \mathbbm{N}_{m(2)1}}h_{m(2)1kl}\alpha_{ki}\beta_l+E_{i}\alpha_{m(2)i}\beta_1
        \ea
        \label{eigenvalueeq9}
\end{equation}
There are two possibilities: either $\theta_{m(2)}=\theta_{m(1)}$, so that $m(2)\in \mathbbm{N}_{m(1)}'$ or $\theta_{m(2)}\neq \theta_{m(1)}$ and  $m(2)\notin  \mathbbm{N}_{m(1)}'$. Let us first consider that $m(2)\in \mathbbm{N}_{m(1)}'$. In this case, from the definition \eqref{definicaohfundamental3}, it follows that $h_{m(1)1m(2)l}=h_{m(2)l m(1)1}^*=h_{m(2)l m(1)1}=0$ for any $l$. As a result, we can define as before, for any $k\neq m(2)$ such that $\alpha_{ki}\neq 0$, the following:
\begin{equation}
    \begin{array}{ccc}
         h_{m(2)1kl}&=\bra{u_{m(2)}}H_1\ket{u_k}\mathrm{e}^{i(\gamma_{1}-\gamma_l)}\frac{z_{1}}{z_{l}} & \text{if $l=l(k)$, $\{k,l\}\in \mathbbm{N}_{m(2)1}$, $k\neq m(2)$, $\alpha_{ki}\neq 0$, and $k\notin \mathbbm{N}_{m(2)}'$} \\
         h_{m(2)1kl}&=0 & \text{if $l\neq l(k)$, $\{k,l\}\in \mathbbm{N}_{m(2)1}$, $k\neq m(2)$, and $\alpha_{ki}\neq 0$, or $k\in \mathbbm{N}_{m(2)}'$}\label{definicaohfundamental4}
    \end{array}
\end{equation}
Notice that this condition does not contradicts the definition \eqref{definicaohfundamental4} for $m(1)$, since in both cases we get $h_{m(2)1m(1)1}=h_{m(1)1m(2)1}=0$ when $m(2)\in \mathbbm{N}_{m(1)}'$ (and therefore $m(1)\in \mathbbm{N}_{m(2)}'$). Considering this definition in Eq . \eqref{eigenvalueeq9} we obtain again
\begin{equation}
    \bra{u_{m(2)},u_{1}'}H_2\ket{\delta_i(H_1,U),v}=E_{i}\alpha_{m(2)i}\beta_{1}.
        \label{eigenvalueeq11}
\end{equation}
Now, let us consider the case in which $\theta_{m(2)}\neq \theta_{m(1)}$ so that $m(2)\notin  \mathbbm{N}_{m(1)}'$. In this case, for any $k\neq m(1)$ and $k\neq m(2)$ we can define $h_{m(2)jkl}$ as in Eq. \eqref{definicaohfundamental4} and insert them in Eq. \eqref{eigenvalueeq9} to obtain
\begin{equation}
       E_{i}\alpha_{m(2)i}\beta_{1}=-\bra{u_{m(2)}}H_1\ket{u_{m(1)}}\alpha_{m(2)i}\beta_{1}+\sum_{\{m(1),l\}\in \mathbbm{N}_{m(2)1}}h_{m(2)1m(1)l}\alpha_{m(1)i}\beta_l+E_{i}\alpha_{m(2)i}\beta_{1}.
        \label{eigenvalueeq12}
\end{equation}
where the summation $\sum_{\{m(1),l\}\in \mathbbm{N}_{m(2)1}}$ is over all $l$ such that $\{m(1),l\}\in \mathbbm{N}_{m(2)1}$ for the fixed index $m(1)$. Now, it cannot be the case that simultaneously $m(2)\notin  \mathbbm{N}_{m(1)}'$ and  $1$ is equal to the smallest $l$ such that $\{m(2),l\}\in \mathbbm{N}_{m(1)1}$, since if $m(2)\notin  \mathbbm{N}_{m(1)}'$ and therefore $\theta_{m(1)}\neq \theta_{m(2)}$, it cannot be the case that $\theta_{1}'+\theta_{m(1)}=\theta_{1}'+\theta_{m(2)}$ and therefore $\{m(1),1\}\notin \mathbbm{N}_{m(2)1}$. Therefore, for any $l$, $h_{m(2)1m(1)l}=h_{m(1)lm(2)1}^{*}$ was not defined in the previous definitions in Eq. \eqref{definicaohfundamental3} for $m(1)$. Thus, we can define $h_{m(2)1m(1)l}$ as in Eq. \eqref{definicaohfundamental4} substituting $k\to m(1)$ to obtain, considering Eq. \eqref{eigenvalueeq12},
\begin{equation}
        \ba{rl}
        \bra{u_{m(2)},u_{1}'}H_2\ket{\delta_i(H_1,U),v}=E_{i}\alpha_{m(2)i}\beta_{1}.
        \ea
        \label{eigenvalueeq13}
\end{equation}
This procedure can be repeated for all $m$ in which $\alpha_{mi}\neq 0$ for the same index $j=1$ using the same reasoning, so that defining for all $m$ 
\begin{equation}
    \begin{array}{ccc}
         h_{m1kl}&=\bra{u_{m}}H_1\ket{u_k}\mathrm{e}^{i(\gamma_{1}-\gamma_l)}\frac{z_{1}}{z_{l}} & \text{if $l=l(k)$, $\{k,l\}\in \mathbbm{N}_{m1}$, $k\neq m$, $\alpha_{ki}\neq 0$, and $k\notin \mathbbm{N}_{m}'$} \\
         h_{m1kl}&=0 & \text{if $l\neq l(k)$, $\{k,l\}\in \mathbbm{N}_{m1}$, $k\neq m$, and $\alpha_{ki}\neq 0$, or $k\in \mathbbm{N}_{m}'$},\label{definicaohfundamental6}
    \end{array}
\end{equation}
it follows that 
\begin{equation}
        \bra{u_{m},u_{1}'}H_2\ket{\delta_i(H_1,U),v}=E_i\alpha_{mi}\beta_{1}=E_i\braket{u_{m},u_{1}|\delta_i(H_1,U),v}.
        \label{eigenvalueeq14}
\end{equation}
Now, let us consider again $m(1)$ but now we consider the second smallest index $j=2$. We need to consider now the fact that it is possible that some $h_{m(1)2k1}=h_{k1m(1)2}^*$ have already being defined in Eq. \eqref{definicaohfundamental6}. Let us define the set $\mathbbm{O}_{12m(1)}$ as the set of indexes $k$ such that $h_{k1m(1)2}$ have already being defined in Eq. \eqref{definicaohfundamental6}. For all $k\notin \mathbbm{O}_{12m(1)}$, we can define
\begin{equation}
    \begin{array}{ccc}
         h_{m(1)2kl}&=\bra{u_{m(1)}}H_1\ket{u_k}\mathrm{e}^{i(\gamma_{2}-\gamma_l)}\frac{z_{2}}{z_{l}} & \text{if $l=l'(k)$, $\{k,l\}\in \mathbbm{N}_{m(1)2}$, $k\neq m(1)$, $\alpha_{ki}\neq 0$, and $k\notin \mathbbm{N}_{m(1)}'$} \\
         h_{m(1)2kl}&=0 & \text{if $l\neq l'(k)$, $\{k,l\}\in \mathbbm{N}_{m(1)2}$, $k\neq m(1)$, and $\alpha_{ki}\neq 0$, or $k\in \mathbbm{N}_{m(1)}'$},\label{definicaohfundamental7}
    \end{array}
\end{equation}
where now $l'(k)$ is the minimum $l$ such that $\{k,l\}\in \mathbbm{N}_{m(1)2}$ and $l>1$. Considering this definition, it follows that the eigenvector-eigenvalue equation for the component $\ket{u_{m(1)},u_{2}'}$, taking into account \eqref{definicaohfundamentaliguaistotal}, becomes
\begin{equation}
        \ba{rl}
        \displaystyle E_{i}\alpha_{m(1)i}\beta_{2}&\displaystyle=-\sum_{\{k\}\notin \mathbbm{N}_{m(1)}'}'\bra{u_{m(1)}}H_1\ket{u_k}\alpha_{ki}\beta_{2}+\sum_{\{k,l\}\in \mathbbm{N}_{m(1)2}}'h_{m(1)2kl}\alpha_{ki}\beta_l+\sum_{l\in \mathbbm{N}_{m(1)2}'}h_{m(1)2m(1)l}\alpha_{m(1)i}\beta_l.\\
        &\displaystyle=-\sum_{\{k\}\notin \mathbbm{N}_{m(1)}'}'\bra{u_{m(1)}}H_1\ket{u_k}\alpha_{ki}\beta_2+\sum_{\{k,1\}\in \mathbbm{N}_{m(1)2}}'h_{m(1)2k1}\alpha_{ki}\beta_{1}+\sum_{l\neq 1\{k,l\}\in \mathbbm{N}_{m(1)2}}'h_{m(1)2kl}\alpha_{ki}\beta_l+\\
        &+E_{i}\alpha_{m(1)i}\beta_{2}.
        \ea
        \label{eigenvalueeq15}
\end{equation}
where we considered the notation
\begin{equation}
    \sum'\equiv \sum_{k\neq m(1),k\in\mathbbm{O}_{12m(1)}}.
\end{equation}
Now, as can be straightforwardly checked, it is generally expected that
\begin{equation}
    -\sum_{\{k\}\notin \mathbbm{N}_{m(2)}'}'\bra{u_{m(1)}}H_1\ket{u_k}\alpha_{ki}\beta_{2}+\sum_{\{k,1\}\in \mathbbm{N}_{m(1)2}}'h_{m(1)2k1}\alpha_{ki}\beta_{1}\neq 0
\end{equation}
i.e. the first and second sum in the second equality of Eq. \eqref{eigenvalueeq15} will not vanish in general. In this case, we have to define $h_{m(1)2kl}$ appropriately so as to satisfy the eigenvalue-eigenvector equation. We then propose the following definition for all $k\neq m(1)$, $k\in \mathbbm{O}_{12m(1)} $: 
\begin{equation}
    \begin{array}{rll}
         h_{m(1)2kl}&=\bra{u_{m(1)}}H_1\ket{u_k}\mathrm{e}^{i(\gamma_{2}-\gamma_l)}\frac{z_{2}}{z_{l}}-h_{m(1)2k1}\mathrm{e}^{i(\gamma_{1}-\gamma_{l})}\frac{z_{1}}{z_{l}} & \text{if $l=l''(k)$, $\{k,l\}\in \mathbbm{N}_{m(1)2}$, $\alpha_{ki}\neq 0$, and $k\notin \mathbbm{N}_{m(1)}'$} \\
         h_{m(1)2kl}&=0 & \text{if $l\neq l''(k)$, $\{k,l\}\in \mathbbm{N}_{m(1)2}$, and $\alpha_{ki}\neq 0$, or $k\in \mathbbm{N}_{m(1)}'$}.\label{definicaohfundamental8}
    \end{array}
\end{equation}
where $l''(k)$ is the smallest $l$ such that $\{k,l\}\in \mathbbm{N}_{m(1)2}$ with $k\in \mathbbm{O}_{12m(1)}$ and $l> 1$. Notice that because there are infinitely many $l$ such that $\{k,l\}\in \mathbbm{N}_{m(1)2}$, then we can always find such $l''(k)$.  Inserting the definition Eq. \eqref{definicaohfundamental8}  in Eq. \eqref{eigenvalueeq15} for all $k\neq m(1)$, $k\in \mathbbm{O}_{12m(1)}$, we obtain 
\begin{equation}
    \bra{u_{m(1)},u_{2}'}H_2\ket{\delta_i(H_1,U),v}=E_i\alpha_{m(1)i}\beta_{2}=E_i\braket{u_{m(1)},u_{2}|\delta_i(H_1,U),v}.
\end{equation}
Let us now consider the index $m(2)$ and $j=2$. There are two possibilities: either $m(2)\in \mathbbm{N}_{m(1)}'$ or $m(2)\notin  \mathbbm{N}_{m(1)}'$. Let us first consider that $m(2)\in \mathbbm{N}_{m(1)}'$. In this case, from the definition \eqref{definicaohfundamental7} or \eqref{definicaohfundamental8}, it follows that $h_{m(1)2m(2)l}=h_{m(2)l m(1)2}^*=h_{m(2)l m(1)2}=0$ for any $l$. Moreover, we need to consider once again the possibility that for some $k\neq m(1)$ and $k\neq m(2)$, $h_{m(2)2k1}=h_{k1m(2)2}^*$ have already being defined in Eq. \eqref{definicaohfundamental6}. Defining analogously as for $m(1)$ the set $\mathbbm{O}_{12m(2)}$ as the set of indexes $k$ such that $h_{k1m(2)2}$ have already being defined in Eq. \eqref{definicaohfundamental6}, then we can define for every $k\neq m(2)$ and $k\neq m(1)$ the following:
\begin{equation}
    \begin{array}{rll}
         h_{m(2)2kl}&=\bra{u_{m(2)}}H_1\ket{u_k}\mathrm{e}^{i(\gamma_{2}-\gamma_l)}\frac{z_{2}}{z_{l}}-h_{m(2)2k1}\mathrm{e}^{i(\gamma_{1}-\gamma_{l})}\frac{z_{1}}{z_{l}} & \text{if $l=l''(k)$, $\{k,l\}\in \mathbbm{N}_{m(2)2}$, $\alpha_{ki}\neq 0$, $k\notin \mathbbm{N}_{m(2)}'$, and $k\in\mathbbm{O}_{12m(2)}$} \\
         h_{m(2)2kl}&=\bra{u_{m(2)}}H_1\ket{u_k}\mathrm{e}^{i(\gamma_{2}-\gamma_l)}\frac{z_{2}}{z_{l}} & \text{if $l=l''(k)$, $\{k,l\}\in \mathbbm{N}_{m(2)2}$, $\alpha_{ki}\neq 0$, $k\notin \mathbbm{N}_{m(2)}'$, and $k\notin\mathbbm{O}_{12m(2)}$} \\
         h_{m(2)2kl}&=0 & \text{if $l\neq l''(k)$, $\{k,l\}\in \mathbbm{N}_{m(2)2}$, and $\alpha_{ki}\neq 0$, or $k\in \mathbbm{N}_{m(2)}'$}.\label{definicaohfundamental9}
    \end{array}
\end{equation}
where $l''(k)$ is the smallest $l$ such that $\{k,l\}\in \mathbbm{N}_{m(2)2}$ with $k\in \mathbbm{O}_{12m(1)}$ and $l>1$. Using the definition \eqref{definicaohfundamental9} it follows that in the component $\ket{u_{m(2)},u_{2)}'}$ of the eigenvalue-eigenvector equation \eqref{eigenvvHy}, we obtain
\begin{equation}
        \bra{u_{m(2)},u_{2}'}H_2\ket{\delta_i(H_1,U),v}=E_i\alpha_{m(2)i}\beta_{2}=E_i\braket{u_{m(2)},u_{2}|\delta_i(H_1,U),v}
        \label{eigenvalueeq16}
\end{equation}
as we wanted. If, on the other hand, $\theta_{m(2)}\neq \theta_{m(1)}$ so that $m(2)\notin  \mathbbm{N}_{m(1)}'$, then it cannot be the case that simultaneously $m(2)\notin  \mathbbm{N}_{m(1)}'$ and  $2$ is equal to the smallest $l$ such that $\{m(2),l\}\in \mathbbm{N}_{m(1)2}$, since if $m(2)\notin  \mathbbm{N}_{m(1)}'$ and therefore $\theta_{m(1)}\neq \theta_{m(2)}$, it cannot be the case that $\theta_{2}'+\theta_{m(1)}=\theta_{2}'+\theta_{m(2)}$ and therefore $\{m(1),2\}\notin \mathbbm{N}_{m(2)2}$. Therefore, for any $l>1$, $h_{m(2)2m(1)l}=h_{m(1)lm(2)2}^{*}$ was not defined in the previous definitions in Eq. \eqref{definicaohfundamental3}. Thus, we can define $h_{m(2)2m(1)l}$ as in Eq. \eqref{definicaohfundamental9} substituting $k\to m(2)$ to obtain Eq. \eqref{eigenvalueeq16} when $\theta_{m(2)}\neq \theta_{m(1)}$ as well. 

Using the same arguments as for $m(2)$, we can consider for any $m$ the definition:
\begin{equation}
    \begin{array}{rll}
         h_{m2kl}&=\bra{u_{m}}H_1\ket{u_k}\mathrm{e}^{i(\gamma_{2}-\gamma_l)}\frac{z_{2}}{z_{l}}-h_{m2k1}\mathrm{e}^{i(\gamma_{1}-\gamma_{l})}\frac{z_{1}}{z_{l}} & \text{if $l=l''(k)$, $\{k,l\}\in \mathbbm{N}_{mj(2)}$, $\alpha_{ki}\neq 0$, $k\notin \mathbbm{N}_{m}'$, and $k\in\mathbbm{O}_{12m}$} \\
         h_{m2kl}&=\bra{u_{m}}H_1\ket{u_k}\mathrm{e}^{i(\gamma_{2}-\gamma_l)}\frac{z_{2}}{z_{l}} & \text{if $l=l''(k)$, $\{k,l\}\in \mathbbm{N}_{m2}$, $\alpha_{ki}\neq 0$, $k\notin \mathbbm{N}_{m}'$, and $k\notin\mathbbm{O}_{12m}$} \\
         h_{m2kl}&=0 & \text{if $l\neq l''(k)$, $\{k,l\}\in \mathbbm{N}_{m2}$, and $\alpha_{ki}\neq 0$, or $k\in \mathbbm{N}_{m}'$}.\label{definicaohfundamental10}
    \end{array}
\end{equation}
where $\mathbbm{O}_{12m}$ is the set of indexes $k$ such that $h_{m2k1}$ have already been defined, $l''(k)$ is the smallest $l>1$ such that $\{k,l\}\in \mathbbm{N}_{m2}$ with $k\in \mathbbm{O}_{12m}$. Using this definition, it follows for any $m$ that 
\begin{equation}
        \bra{u_{m},u_{2}}H_2\ket{\delta_i(H_1,U),v}=E_i\alpha_{mi}\beta_{2}=E_i\braket{u_{m},u_{2}|\delta_i(H_1,U),v}.
        \label{eigenvalueeq17}
\end{equation}

Now, we want to prove a rule of induction considering the procedure that we did until now for all other indexes. Specifically, we want to define the component $h_{m(1)jkl}$ of which $\alpha_{m(1)i}\neq 0$ for any $j$, given that $h_{mj'kl}$ was already defined using the same procedure for all $k$, $l$, $m$ and $j'=j-1,j-2,\cdots 1$ and that for all such $m$ and $j'$
\begin{equation}
        \bra{u_{m},u_{j'}}H_2\ket{\delta_i(H_1,U),v}=E_i\alpha_{mi}\beta_{j'}.
        \label{eigenvalueeqinduction}
\end{equation}
Let us thus consider the indexes $j$ and $m(1)$. We define $\mathbbm{J}(j,m(1),k)$ as the set of all indexes $\{j'\}$ such that $j'\leq j$ and $h_{m(1)jkj'}=h_{kj'm(1)j}$ was already defined in previous steps. Also, we define $\mathbbm{O}(j,m(1))$ of all the $k$ indexes in which at least some $j'<j$, the term $h_{m(1)jkj'}$ was already defined. Considering these sets, we then introduce the following definitions
\begin{equation}
    h_{m(1)jkl}=\bra{u_{m(1)}}H_1\ket{u_k}\mathrm{e}^{i(\gamma_{j}-\gamma_l)}\frac{z_{j}}{z_{l}}-\sum_{j'\in\mathbbm{J}(j,m(1),k) }h_{m(1)jkj'}\mathrm{e}^{i(\gamma_{j'}-\gamma_{l})}\frac{z_{j'}}{z_{l}}\label{definicaohfundamental11}
\end{equation}
if $l=l''(k)$, where $l''(k)$ is the smallest $l$ such that $\{k,l\}\in \mathbbm{N}_{m(1)j}$, $l>j$, $\alpha_{ki}\neq 0$, $k\notin \mathbbm{N}_{m(1)}'$, and $k\in\mathbbm{O}(j,m(1))$. If $k\notin\mathbbm{O}(j,m(1))$, then we can consider the same definition $l=l''(k)$, and consider
\begin{equation}
        h_{m(1)jkl}=\bra{u_{m(1)}}H_1\ket{u_k}\mathrm{e}^{i(\gamma_{j}-\gamma_l)}\frac{z_{j}}{z_{l}}\label{definicaohfundamental12}
\end{equation}
where  $\{k,l\}\in \mathbbm{N}_{mj}$, $\alpha_{ki}\neq 0$, $k\notin \mathbbm{N}_{m(1)}'$. If $l\neq l''(k)$, $\{k,l\}\in \mathbbm{N}_{m(1)j}$, $\alpha_{ki}\neq 0$, and $k\notin \mathbbm{N}_{m(1)}'$ or in case that $k\in \mathbbm{N}_{m(1)}'$
then we define
\begin{equation}
         h_{m(1)jkl}=0.\label{definicaohfundamental13}
\end{equation}
Notice that since for any $k$, $j$, and $m(1)$, there are infinitely many $l$ such that $\{k,l\}\in \mathbbm{N}_{m(1)j} $, then $l''(k)$ can always be defined.

Let us consider the eigenvalue-eigenvector equation for the component $\ket{u_{m(1)},u_{j}'}$, considering \eqref{definicaohfundamentaliguaistotal}:
\begin{equation}
         E_{i}\alpha_{m(1)i}\beta_{j}=-\sum_{k\neq m(1),\{k\}\notin \mathbbm{N}_{m(1)}'}\bra{u_{m(1)}}H_1\ket{u_k}\alpha_{ki}\beta_{j}+\sum_{k\neq m(1),\{k,l\}\in \mathbbm{N}_{m(1)j}}h_{m(1)jkl}\alpha_{ki}\beta_l+E_{i}\alpha_{m(1)i}\beta_{j}.
        \label{eigenvalueeq18}
    \end{equation}
Now, considering the fact that for all $k\in\mathbbm{O}(j,m(1))$ and for such $k$, all indexes $j'\in\mathbbm{J}(j,m(1),k)$, $h_{m(1)jkj'}$ was already defined for some $j'<j$, then Eq. \eqref{eigenvalueeq18} can be rewritten as
\begin{equation}
        \ba{rl}
        \displaystyle E_{i}\alpha_{m(1)i}\beta_{j}&=-\sum_{k\neq m(1),\{k\}\notin \mathbbm{N}_{m(1)}'}\bra{u_{m(1)}}H_1\ket{u_k}\alpha_{ki}\beta_{j}+\sum_{k\in\mathbbm{O}(j,m(1)),j'\in\mathbbm{J}(j,m(1),k)}'h_{m(1)jkj'}\alpha_{ki}\beta_{j'}+\\
        &\displaystyle+\sum_{k\in\mathbbm{O}(j,m(1)),l}'h_{m(1)jkl}\alpha_{ki}\beta_l+\sum_{k\notin\mathbbm{O}(j,m(1)),l}'h_{m(1)jkl}\alpha_{ki}\beta_l+E_{i}\alpha_{m(1)i}\beta_{j}.
        \ea
        \label{eigenvalueeq19}
\end{equation}
where we consider the short notation:
\begin{equation}
    \sum'=\sum_{k\neq m(1),\{k,l\}\in \mathbbm{N}_{m(1)j}}.
\end{equation}
Substituting the definitions in Eqs. \eqref{definicaohfundamental11}-\eqref{definicaohfundamental13} in the third and second sum of Eq. \eqref{eigenvalueeq19}, we can check that
\begin{equation}
        \ba{rl}
        0\displaystyle& =-\sum_{k\neq m(1),\{k\}\notin \mathbbm{N}_{m(1)}'}\bra{u_{m(1)}}H_1\ket{u_k}\alpha_{ki}\beta_{j}+\sum_{k\in\mathbbm{O}(j,m(1)),j'\in\mathbbm{J}(j,m(1),k)}'h_{m(1)jkj'}\alpha_{ki}\beta_{j'}+\\
        &\displaystyle+\sum_{k\in\mathbbm{O}(j,m(1)),l}'h_{m(1)jkl}\alpha_{ki}\beta_l+\sum_{k\notin\mathbbm{O}(j,m(1)),l}'h_{m(1)jkl}\alpha_{ki}\beta_l.
        \ea
        \label{eigenvalueeq20}
\end{equation}
so that
\begin{equation}
        \bra{u_{m(1)},u_{j}}H_2\ket{\delta_i(H_1,U),v}=E_i\alpha_{m(1)i}\beta_{j}.
        \label{eigenvalueeq21}
\end{equation}
Now, notice that for any $r>1$, it follows that either $m(r)\in \mathbbm{N}_{m(1)}'$ or $m(r)\notin  \mathbbm{N}_{m(1)}'$. In the case in which $m(r)\in \mathbbm{N}_{m(1)}'$, it follows that, in accordance with \eqref{definicaohfundamental13}, $h_{m(1)jm(r)l}=h_{m(r)l m(1)j}^*=h_{m(r)l m(1)j}=0$ for any $l$. We can thus consider the same definitions as in Eqs. \eqref{definicaohfundamental11}-\eqref{definicaohfundamental13} substituting $m(1)\to m(r)$ and this will not contradict the fact that $h_{m(1)jm(r)j}=h_{m(r)jm(1)j}^*=h_{m(r)j m(1)j}=0$. As a result, it follows that using  Eqs. \eqref{definicaohfundamental11}-\eqref{definicaohfundamental13}, Eq. \eqref{eigenvalueeq21} will hold similarly substituting $m(1)\to m(r)$. On the other hand, if $\theta_{m(r)}\neq \theta_{m(1)}$ and  $m(r)\notin  \mathbbm{N}_{m(1)}'$, then it cannot be the case that simultaneously $m(r)\notin  \mathbbm{N}_{m(1)}'$ and  $j$ is equal to the smallest $l$ such that $\{m(r),l\}\in \mathbbm{N}_{m(1)j}$. Therefore, for any $l>j-1$, $h_{m(r)jm(1)l}=h_{m(1)lm(r)j}^{*}$ was not yet defined. Thus, we can define $h_{m(r)jkl}$ as in Eqs. \eqref{definicaohfundamental11}-\eqref{definicaohfundamental13} substituting $m(1)\to m(r)$ to obtain Eq. \eqref{eigenvalueeq21} for any $r$. As a result, given that we assume for the induction process that Eq. \eqref{eigenvalueeqinduction} holds for any $m$ and $j'=j-1,j-2,\cdots,1$ and given the fact that for all $r\geq 1$, Eq. \eqref{eigenvalueeq21} holds substituting $m(1)\to m(r)$, then it follows that for all $m$ and all $j'=j,j-1,j-2,\cdots 1$, that
\begin{equation}
        \bra{u_{m},u_{j'}}H_2\ket{\delta_i(H_1,U),v}=E_i\alpha_{mi}\beta_{j'}=E_i\braket{u_{m},u_{j'}|\delta_i(H_1,U),v}.
        \label{eigenvalueeq21}
\end{equation}
Therefore, by assuming that Eq. \eqref{eigenvalueeqinduction} holds for any $m$ and all $j'=j-1,j-2,\cdots 1$, we proved that it also holds for all $m$ and all $j'=j,j-1,j-2,\cdots 1$. Given that $j$ is arbitrary here and that we proved Eqs. \eqref{eigenvalueeq14} and \eqref{eigenvalueeq17} for $j=1$ and $j=2$, respectively, then Eq. \eqref{eigenvalueeqinduction} holds for $j=1$ and $j=2$, and, by induction, it is valid for all $j=3,4,5,\cdots$. As a result, for any $m$ such that $\alpha_{mi}\neq 0$ and any $j$,
\begin{equation}
        \bra{u_{m},u_{j}'}H_2\ket{\delta_i(H_1,U),v}=E_i\alpha_{mi}\beta_{j}.
        \label{eigenvaluemneq}
\end{equation}
Considering this result and \eqref{eigenvalueeq}, we can thus conclude that for the only two possible scenarios, in which either $\alpha_{mi}\neq 0$ or $\alpha_{mi}=0$, we can always tune $h_{mjkl}$ that defines $H_2$ in Eq. \eqref{componentsHY}  to satisfy all the components $\ket{u_m,u_j'}$ of Eq. \eqref{eigenvvHy}. This means that we can always find a hermitian operator $H_2$ that satisfy 
    \begin{equation}
        [H_1\otimes\mathbbm{1}_{\mathcal{H}'}+H_2,U\otimes U']=0\label{commutresult1rep}
    \end{equation}
    and
    \begin{equation}
        H_2\ket{\delta_i(H_1,U),v}=E_i\ket{\delta_i(H_1,U),v}.\label{result12ndrelationrep}
    \end{equation} 
    
We conclude the proof of the result 2 by proving that the same $H_2$ that satisfy Eqs. \eqref{commutresult1rep} and \eqref{result12ndrelationrep}, also satisfy the relation
     \begin{equation}
         (U^{\dagger}\otimes U^{'\dagger})H_2(U\otimes U')\ket{\delta_i(H_1,U),v}=E_i'\ket{\delta_i(H_1,U),v},\label{result12ndrelationrep2}
     \end{equation}
     where $E_i'=E_i-\delta_i(H_1,U)$. To do so, we consider a $H_2$ that satisfy Eqs. \eqref{commutresult1rep} and \eqref{result12ndrelationrep}. Applying $U^{\dagger}\otimes U^{'\dagger}$ on the left of both sides of Eq. \eqref{commutresult1rep}, we obtain
     \begin{equation}
         \Delta(H_1,U)\otimes\mathbbm{1}_{\mathcal{H}'}= U^{\dagger}H_1 U\otimes\mathbbm{1}_{\mathcal{H}'}-H_1\otimes\mathbbm{1}_{\mathcal{H}'}=-[(U^{\dagger}\otimes U^{'\dagger})H_2(U\otimes U')-H_2]=-\Delta(H_2,U\otimes U').
     \end{equation}
    Applying $\Delta(H_2,U\otimes U')$ on $\ket{\delta_i(H_1,U),v}$ considering the above equality, we get
    \begin{equation}
        [(U^{\dagger}\otimes U^{'\dagger})H_2(U\otimes U')-H_2]\ket{\delta_i(H_1,U),v}=\Delta(H_2,U\otimes U')\ket{\delta_i(H_1,U),v}=-\delta_i(H_1,U)\ket{\delta_i(H_1,U),v}\label{eigenrevertido}
    \end{equation}
    which means that $\ket{\delta_i(H_1,U),v}$ is also an eigenvector of $\Delta(H_2,U\otimes U')$ with eigenvalue $-\delta_i(H_1,U)$. Considering Eq. \eqref{eigenrevertido}, the fact that $H_2$ satisfy Eq. \eqref{result12ndrelationrep} and $(U^{\dagger}\otimes U^{'\dagger})H_2(U\otimes U')=H_2+\Delta(H_2,U\otimes U')$, we deduce
    \begin{equation}
        (U^{\dagger}\otimes U^{'\dagger})H_2(U\otimes U')\ket{\delta_i(H_1,U),v}=(H_2+\Delta(H_2,U\otimes U'))\ket{\delta_i(H_1,U),v}=(E_i-\delta_i(H_1,U))\ket{\delta_i(H_1,U),v},
    \end{equation}
    concluding the proof.
\end{proof}

\section{Appendix D: Adapting to the time-dependent case}
It is interesting to notice that results 1 and 2 can be immediately adapted to any other quantum observalbe $A$ substituting $H_1\to A$ in the results and deductions, given that $U$ has a countable basis. Therefore, we can consider the substitution $H_1\to X,P,L,N,\cdots$, i.e. position, linear and angular momentum, number of particles etc. In the present section, we show that we can adapt also our formalism for time-dependent cases, as long as we do adjustments in the CRIN conditions. 

We then consider now that $H_1$ can have a explicit time-dependence in the Schr\"odinger picture, so that $H_1\to H_1(t)$, for any given time $t$. Similarly, we consider that such a system evolves according with a unitary $U$ with countable basis. In this case, the two time observables defining the variation of $H_1(t)$ on the interval is defined, for the discrete case, as
\begin{equation}
    \Delta (H_1(0),H_1(t),U)=U^{\dagger}H_1(t)U-H_1(0)=\sum_{i}\delta_i(H_1(0),H_1(t),U) \ket{\delta_{i}(H(0),H(t),U)},\label{timedependent}
\end{equation}
where $\Delta (H_1(0),H_1(t),U)\ket{\delta_{i}(H(0),H(t),U)}=\delta_{i}(H(0),H(t),U)\ket{\delta_{i}(H(0),H(t),U)}$. An analogous definition can be made for the continuous and/or degenerate case. Notice that $U$ in general will have the dependence of $H_1(t')$ for all $t'\in[0,t]$. This dependence can thus be implicit in $U$ and we take into account this fact in the notation of 
$\Delta (H_1(0),H_1(t),U)$.

Now, we define in analogy to the time-independent case, a time-dependent measurement protocol $\mathbbm{M}_{t}$ \cite{Perarnau2017,Silva2024} to measure the \emph{variation} of any time-dependent arbitrary operator $H_1(t)$ under an arbitrary evolution $U$; for the protocol $\mathbbm{M}_{t}$ and for each $(H_1(0),H_1(t), U)$ triplet, the set $\mathbbm{M}_{t}(H_1(0),H_1(t), U) = \{M(z, H_1(0),H_1(t), U)\}$ defines a POVM whose operators satisfy the usual POVM properties, \emph{viz}. $\int_{-\infty}^{\infty} dz  M(z, H_1(0),H_1(t), U) = \mathbbm{1}$ and $M(z, H_1(0),H_1(t), U) \geq 0$. Given this definition, we redefine the CRIN conditions for the time-dependent case: 
\begin{enumerate}    
    \item \textbf{Conservation laws}: \emph{For any preparation $\rho$, unitary evolution $U$, and time-dependent energy operators defined $H_1(t)$ and $H_2(t)$ representing parts of the energy of any system $\Omega$, if $U^\dagger H_2(t)U+U^\dagger H_1(t)U=H_2(0)+H_1(0)$, then, for any $z$, $\Tr [M(z, H_1(0),H_1(t),U)\rho]=\wp(z,H_1(0),H_1(t),U,\rho)=\wp(-z,H_2(0),H_2(t),U,\rho)=\Tr [M(-z, H_2(0),H_2(t),U)\rho]$}.
    \item \textbf{Reality}:  \emph{For any system $\Omega$, operators $H_1(0)$ and $H_1(t)$, and unitary evolution $U$, if there is an initial state $\rho_1=\ket{e_1}\bra{e_1}$ such that $\ket{e_1}$ is an eigenvector of both $H_{1}(0)$ and $U^{\dagger}H_1(t) U$ with respective eigenvalues $e_1(0)$ and $\epsilon_1(t)$, then the POVM must result in the probabilities $\wp(z,H_1(0),H_1(t),U)=\delta^{\text{\tiny \textbf{D}}}[z-(\epsilon_1(t)-e_1(0))]$}. 
    \item \textbf{ Independence of the initial state}: \emph{For any system $\Omega$, operators $H_1(0)$ and $H_1(t)$, and evolution operators $U$, the elements of the POVM ${M(z, H_1(0),H_1(t),U)}$ must not depend on the initial state $\rho$}. 
    \item \textbf{No-signaling}: \emph{For any system $\Omega$, if $\Omega$ evolves under an arbitrary bipartite unitary evolution $U\otimes U'$ acting on a bipartite Hilbert space $\mathcal{H}=\mathcal{H}_1\otimes \mathcal{H}_2$, then, for every energy operators $H_1(0)\otimes \mathbbm{1}_2$ and  $H_1(t)\otimes \mathbbm{1}_2$ acting locally on $\mathcal{H}_1$, $\mathbbm{M}$ is such that its POVMs satisfy the following relation for every $z$: $M(z,H_1(0)\otimes \mathbbm{1}_{2},H_1(t)\otimes \mathbbm{1}_{2},U\otimes U')=M(z,H_1(0),H_1(t),U)\otimes \mathbbm{1}_{2}$.}
\end{enumerate} 

Given this adaptation (and a similar adaptation of the reality condition for the continuous case), we now comment on how to adapt result 1 and 2 for the time-dependent case. 

First, the following result is an immediate consequence of Result 2:
\begin{corolarios}
    For any unitary evolution $U$, energy operators $H_1(0)$ and $H_1(t)$ acting on a Hilbert space $\mathcal{H}$, and eigenstate $\ket{\delta_i(H_1(0),H_1(t), U)}$ of $\Delta(H_1(0),H_1(t), U)$, there exists a unitary $U'$ acting on an auxiliary Hilbert space $\mathcal{H}'$, an additional operators $H_2(0)$ and $H_2(t)$  acting on $\mathcal{H} \otimes \mathcal{H}'$, and a vector $\ket{v} \in \mathcal{H}'$ such that the following equations are satisfied: 
\begin{eqnarray} 
&(U^\dagger \otimes U^{'\dagger})H_1(t) \otimes \mathbbm{1}_{\mathcal{H}'}(U \otimes U') + (U^\dagger \otimes U^{'\dagger})H_2(U \otimes U')=H_1(0) \otimes \mathbbm{1}_{\mathcal{H}'}+H_2(0) = 0, \label{commutresult10time} \\
&H_2 \ket{\delta_i(H_1(0),H_1(t), U), v} = E_i \ket{\delta_i(H_1(0),H_1(t), U), v},\label{result122ndrelation0time}\\
&(U^\dagger \otimes U^{'\dagger}) H_2 (U \otimes U') \ket{\delta_i(H_1(0),H_1(t), U), v} = E_i' \ket{\delta_i(H_1(0),H_1(t), U), v}, \label{result12ndrelationtime} 
\end{eqnarray} 
where $\ket{\delta_i(H_1(0),H_1(t), U), v} = \ket{\delta_i(H_1(0),H_1(t), U)} \otimes \ket{v}$, and $E_i$ and $E_i' = E_i - \delta_i(H_1(0),H_1(t), U)$ are real numbers.\label{coro}
\end{corolarios}
\begin{proof}
    By directly substituting $H_1\to H_1(0)$, $\ket{\delta_i(H_1,U)}\to \ket{\delta_i(H_1(0),H_1(t),U)}$ in Result 2, we deduce that there is a $H_2(0)$, $U'$, and $\ket{v}$ that satisfy the following equations \footnote{Notice that here we consider $\ket{\delta_i(H_1(0),H_1(t),U)}$ (eigenvector of $U^{\dagger}H_1(t)U-H_1(0)$ ) instead of $\ket{\delta_i(H_1,U)}$ (eigenvector of $U^{\dagger}H_1U-H_1(0)$) as stated in Result 2. As can be straightforwardly checked, the result is not restricted for arbitrary form for $\ket{\delta_i(H_1,U)}$, so that it can be applied for $\ket{\delta_i(H_1(0),H_1(t),U)}$ as well.}:
    \begin{equation}
        [H_1(0)\otimes\mathbbm{1}_{\mathcal{H}'}+H_2(0),U\otimes U']=0,\label{timedependHX1}
    \end{equation}
    and 
    \begin{equation}
    H_2(0)\ket{\delta_i(H_1(0),H_1(t),U),v}=E_{i}\ket{\delta_i(H_1(0),H_1(t),U),v}.\label{timedependHX0}
    \end{equation}
    Applying $U^{\dagger}\otimes U^{'\dagger}$ on both sides of Eq. \eqref{timedependHX1}, we get
    \begin{equation}
        (U^{\dagger}\otimes U^{'\dagger})(H_1(0)\otimes\mathbbm{1}_{\mathcal{H}'}+H_2(0))(U\otimes U')=H_1(0)\otimes\mathbbm{1}_{\mathcal{H}'}+H_2(0). \label{timedependHX2}
    \end{equation}
    Defining
    $H_2(t)=H_1(0)\otimes\mathbbm{1}_{\mathcal{H}'}+H_2(0)-H_1(t)\otimes\mathbbm{1}_{\mathcal{H}'}$ and, considering Eqs. \eqref{timedependHX1} and \eqref{timedependHX2}, we obtain
    \begin{equation}
        \ba{rl}
        (U^{\dagger}\otimes U^{'\dagger})H_2(t)(U\otimes U')&=(U^{\dagger}\otimes U^{'\dagger})(H_1(0)\otimes\mathbbm{1}_{\mathcal{H}'}+H_2(0))(U\otimes U')-(U^{\dagger}\otimes U^{'\dagger})(H_1(t)\otimes\mathbbm{1}_{\mathcal{H}'})(U\otimes U')\\
        &=H_1(0)\otimes\mathbbm{1}_{\mathcal{H}'}+H_2(0)-(U^{\dagger}\otimes U^{'\dagger})(H_1(t)\otimes\mathbbm{1}_{\mathcal{H}'} )(U\otimes U')\\
        &=H_1(0)\otimes\mathbbm{1}_{\mathcal{H}'}+H_2(0)-U^{\dagger}H_1(t)U\otimes\mathbbm{1}_{\mathcal{H}'}\\
        &=H_2(0)-(U^{\dagger}H_1(t)U\otimes\mathbbm{1}_{\mathcal{H}'}-H_1(0)\otimes\mathbbm{1}_{\mathcal{H}'})\\
        &=H_2(0)-(\Delta (H_1(0),H_1(t),U)\otimes\mathbbm{1}_{\mathcal{H}'}).
        \ea
    \end{equation}
    As a result, it follows that
    \begin{equation}
        (U^{\dagger}\otimes U^{'\dagger})(H_2(t)+H_1(t)\otimes\mathbbm{1}_{\mathcal{H}'})(U\otimes U')=H_2(0)+H_1(0)\otimes\mathbbm{1}_{\mathcal{H}'}\label{coro1result}
    \end{equation}
    and
    \begin{equation}
        \ba{rl}
        (U^{\dagger}\otimes U^{'\dagger})H_2(t)(U\otimes U')\ket{\delta_i(H_1(0),H_1(t),U),v}&=(H_2(0)-(\Delta (H_1(0),H_1(t),U)\otimes\mathbbm{1}_{\mathcal{H}'}))\ket{\delta_i(H_1(0),H_1(t),U),v}\\
        &=E_i'\ket{\delta_i(H_1(0),H_1(t),U),v},\label{coro1result2}
        \ea
    \end{equation}
    where $E_i'=E_i-\delta_i(H_1(0),H_1(t),U)$. Considering Eqs. \eqref{timedependHX0}, \eqref{coro1result}, and \eqref{coro1result2}, the result is thus proved.
\end{proof}
Corollary \ref{coro} is therefore the analogous of result 2, and follow almost immediately from it. Considering Corollary \ref{coro}, an analogous of result 1 can be deduced substituting $H_1\to H_1(t)$, $\Delta(H_1,U)\to \Delta(H_1(0),H_1(t),U)$, $\mathbbm{M}\to \mathbbm{M}_{t}$, and considering the adapted time-dependent CRIN conditions. As a consequence, the OBS protocol is also the only protocol that satisfies the time-dependent CRIN conditions for the time-dependent case. 
\section{Appendix E: Results needed for the ion trap example}
\subsection{OBS and TPM for the commuting continuous case}

In Ref. \cite{Talkner2007}, it was shown that whenever a time-dependent Hamiltonian $H(t)$ satisfies
\begin{equation}
    [H(0),U^{\dagger}H(t)U]=0,\label{timedependentcommut}
\end{equation}
then the TPM statistics for a thermal state $\rho_\beta=\mathcal{Z}^{-1}\mathrm{e}^{-\beta H(0)}$ will be the same as the OBS statistics. The authors showed this result by taking into account the TPM characteristic function, defined as
\begin{equation}
    G_{\text{\tiny TPM}}(u)=\int_{-\infty}^{\infty}dz\mathrm{e}^{iuz}\wp_{\text{\tiny TPM}}(z,H(0),H(t),U,\rho_\beta)
\end{equation}
where, for the time-dependent case, 
\begin{equation}
    \wp_{\text{\tiny TPM}}(z,H(0),H(t),U,\rho_\beta)=\Tr[M_{\text{\tiny TPM}}(z,H(0),H(t),U)\rho_\beta],
\end{equation}
and
\begin{equation}
    M_{\text{\tiny TPM}}(z,H(0),H(t),U) = \sum_{jk} \delta^{\text{\tiny \textbf{D}}}[z - (e_j(t) - e_k)] \,\,|\bra{e_j(t)} U \ket{e_k}|^2 \ket{e_k} \bra{e_k}.
\end{equation} 
Here, $\ket{e_j(t)}$ and $\ket{e_k}$ are eigenvectors of $H(t)$ and $H(0)$, respectively. Notice that $M_{\text{\tiny TPM}}(z,H(0),H(t),U)$ is similar to the form defined in the main text, with adjustments due to the time-dependence. When $H(t)$, $H(0)$ and $U$ satisfy  Eq. \eqref{timedependentcommut}, the authors of Ref. \cite{Talkner2007} proved that 
\begin{equation}
    \ba{rl}
    G_{\text{\tiny TPM}}(u)&=\int_{-\infty}^{\infty}dz\mathrm{e}^{iuz}\wp_{\text{\tiny TPM}}(z,H(0),H(t),U,\rho_\beta)=\sum_{jk} \mathrm{e}^{iu(e_j(t) - e_k)}\,\,|\bra{e_j(t)} U \ket{e_k}|^2 \bra{e_k}\rho_\beta \ket{e_k} \\
    &=\Tr[ \mathrm{e}^{iu(U^{\dagger}H(t)U - H(0))}\rho_\beta],\label{Gutpm}
    \ea
\end{equation}
which will be equal to the OBS characteristic function:
\begin{equation}
    G_{\text{\tiny OBS}}(u)=\int_{-\infty}^{\infty}dz\mathrm{e}^{iuz}\wp_{\text{\tiny OBS}}(z,H(0),H(t),U,\rho)=\Tr[ \mathrm{e}^{iu(\Delta (H(0),H(t),U))}\rho_\beta]\label{Guobs}
\end{equation}
where, for the time-dependent case, we considered the two-time OBS statistics \cite{Silva2024}
\begin{equation}
    \wp_{\text{\tiny OBS}}(z,H(0),H(t),U)=\Tr[M_{\text{\tiny OBS}}(z,H(0),H(t),U)\rho_\beta],
\end{equation}
\begin{equation}
    M_{\text{\tiny OBS}}(z,H(0),H(t),U) = \sum_{j} \delta^{\text{\tiny \textbf{D}}}[z - \delta_j(H(0),H(t), U)]\ket{\delta_j(H(0),H(t), U)}\bra{\delta_j(H(0),H(t), U)},
\end{equation}
and
\begin{equation}
    \Delta (H(0),H(t),U)=U^{\dagger}H(t)U - H(0)=\sum_{j}\delta_j(H(0),H(t), U)\ket{\delta_j(H(0),H(t), U)}\bra{\delta_j(H(0),H(t), U)}
\end{equation}
Similarly, the same result can be deduced substituting in Eqs. \eqref{Gutpm} and \eqref{Guobs} the following: $H(0)\to H_1$, $H(t)\to H_1$, $\wp_{\text{\tiny OBS}}(z,H(0),H(t),U)\to \wp_{\text{\tiny OBS}}(z,H_1,U)$ and $\wp_{\text{\tiny TPM}}(z,H(0),H(t),U)\to \wp_{\text{\tiny TPM}}(z,H_1,U)$ (defined in the main text), and $\rho_\beta\to$ any $\rho$. Therefore, it follows, for the time-independent discrete basis case, that if $[H_1,U^\dagger H_1U]=0$, then: 
\begin{equation}
    G_{\text{\tiny TPM}}(u)=\int_{-\infty}^{\infty}dz\mathrm{e}^{iuz}\wp_{\text{\tiny TPM}}(z,H_1,U,\rho)=\int_{-\infty}^{\infty}dz\mathrm{e}^{iuz}\wp_{\text{\tiny OBS}}(z,H_1,U,\rho)=G_{\text{\tiny OBS}}(u),
\end{equation}
and, as a result,
\begin{equation}
    \wp_{\text{\tiny TPM}}(z,H_1,U,\rho)=\wp_{\text{\tiny OBS}}(z,H_1,U,\rho).
\end{equation}
Now, for the deductions in \cite{Talkner2007} and the framework here, we considered discrete basis, so that TPM procedure is clear and needs no adaptation. This is not the case when $H_1$ is diagonalized by a continuous basis $\{\ket{e}\}$, since after a first measurement of $H_1$ of one round of a TPM procedure, the state is not normalizable after the measurement so that the statistics cannot be directly treated \cite{Silva2021}. We can assume that after a first measurement, the state is a normalized state $\ket{\psi_\mu^e}$, such that $\mu$ accounts for the precision of the experiment apparatus used and approximates an eigenstate $\ket{e}$ whenever $\mu\to 0$. For instance, $\braket{e'|\psi_\mu^e}$ can be a Gaussian state with its center located at $e$, such that $\mu$ is the width of the gaussian.
For such states, we expect that
\begin{equation}
    \lim_{\mu\to 0}H_1^n\ket{\psi_\mu^k}=\lim_{\mu\to 0} e_k^n\ket{\psi_\mu^k}\quad\text{and}\quad \lim_{\mu\to 0}|\braket{e|\psi_\mu^{e'}}|^2=\delta^{\text{\tiny \textbf{D}}}(e'-e).\label{tpmcont}
\end{equation}
Whenever $[H_1,U^\dagger H_1U]=0$, since and eigenvector $\ket{e}$ is also eigenvector of $U^\dagger H_1U$, such that $U^\dagger H_1U\ket{e}=\epsilon_e\ket{e}$ for some real $\epsilon_e$, then it also follows that
\begin{equation}
    \lim_{\mu\to 0}U^{\dagger}H_1^n U\ket{\psi_\mu^e}=\lim_{\mu\to 0} \epsilon_e^n\ket{\psi_\mu^e}.\label{tpmcont2}
\end{equation}
As a result, we can define the TPM POVM elements for the continuous case as
\begin{equation}
    M_{\text{\tiny TPM}}(z,H_1,U) = \lim_{\mu\to 0}\iint_{-\infty}^{\infty}dede' \delta^{\text{\tiny \textbf{D}}}[z - (e' - e)] \,\,|\bra{e'} U \ket{\psi_{\mu}^e}|^2 \ket{e} \bra{e},
\end{equation} 
So that, the characteristic function can be computed as
\begin{equation}
    \ba{rl}
    G_{\text{\tiny TPM}}(u)&=\int_{-\infty}^{\infty}dz\mathrm{e}^{iuz}\wp_{\text{\tiny TPM}}(z,H_1,U,\rho)=\lim_{\mu\to 0} \iint_{-\infty}^{\infty}dede'\mathrm{e}^{iu(\epsilon_{e'}-e)}|\bra{\epsilon_{e'}} U \ket{\psi_{\mu}^e}|^2 \bra{e}\rho\ket{e} \\
    &= \lim_{\mu\to 0}\iint_{-\infty}^{\infty}dede' \mathrm{e}^{iu(\epsilon_{e'}-e)} \bra{e'}U  \ket{\psi_{\mu}^e}\bra{e}\rho\ket{e} \bra{\psi_{\mu}^e}U^\dagger\ket{e'},\\
    &=\lim_{\mu\to 0}\iint_{-\infty}^{\infty}dede' \Tr[U^\dagger\ket{e'}\mathrm{e}^{iu\epsilon_{e'}} \bra{e'}U  \ket{\psi_{\mu}^e}\bra{e}\mathrm{e}^{-iue}\rho\ket{e} \bra{\psi_{\mu}^e}]\\
    &=\lim_{\mu\to 0}\int_{-\infty}^{\infty}de \Tr[\mathrm{e}^{iuU^\dagger H_1 U}  \ket{\psi_{\mu}^e}\bra{e}\mathrm{e}^{-iuH_1}\rho\ket{e} \bra{\psi_{\mu}^e}]\\
    &=\lim_{\mu\to 0}\int_{-\infty}^{\infty}de \Tr[\mathrm{e}^{iu\epsilon_e}  \ket{\psi_{\mu}^e}\bra{e}\mathrm{e}^{-iuH_1}\rho\ket{e} \bra{\psi_{\mu}^e}]\\
    &=\lim_{\mu\to 0}\int_{-\infty}^{\infty}de\int_{-\infty}^{\infty}de'\mathrm{e}^{iu\epsilon_e}  \braket{e'|\psi_{\mu}^e}\bra{e}\mathrm{e}^{-iuH_1}\rho\ket{e} \braket{\psi_{\mu}^e|e'}\\
    &=\lim_{\mu\to 0}\int_{-\infty}^{\infty}de\int_{-\infty}^{\infty}de'\mathrm{e}^{iu\epsilon_e}  \bra{e}\mathrm{e}^{-iuH_1}\rho\ket{e} |\braket{\psi_{\mu}^e|e'}|^2\\
    &=\int_{-\infty}^{\infty}de\int_{-\infty}^{\infty}de'\mathrm{e}^{iu\epsilon_e}  \bra{e}\mathrm{e}^{-iuH_1}\rho\ket{e} \delta^{\text{\tiny \textbf{D}}}(e-e')=\int_{-\infty}^{\infty}de  \bra{e}\mathrm{e}^{iu(U^{\dagger}H_1U-H_1)}\rho\ket{e} =\Tr[\mathrm{e}^{iu(U^{\dagger}H_1U-H_1)}\rho]=G_{\text{\tiny OBS}}(u)\label{Gutpm2}
    \ea
\end{equation}
where we considered the properties of Eqs. \eqref{tpmcont} and \eqref{tpmcont2} from the forth line on. As a result, in the limit in which $\mu\to 0$,
\begin{equation}
    \wp_{\text{\tiny TPM}}(z,H_1,U,\rho)=\int_{-\infty}^{\infty}du\mathrm{e}^{-iuz}G_{\text{\tiny TPM}}(u)= \int_{-\infty}^{\infty}du\mathrm{e}^{-iuz}G_{\text{\tiny OBS}}(u)=\wp_{\text{\tiny OBS}}(z,H_1,U,\rho).
\end{equation}

\subsection{Calculating the probabilities $\mathbf{|\bra{n}U_\tau^+\ket{0}|^2}$ for Figure 2}
We consider, as in the main text, that $U_\tau^{+}=\bra{+}U_\tau\ket{+}=\mathrm{e}^{-i\theta_\tau}\exp \left[-\frac{i\tau}{\hbar}\left(\frac{P^2}{2m}+\frac{m\omega^2}{2}X_+^{2}\right)\right]$, $X_+=\bra{+}X'\ket{+}=X+a$, and $\theta_\tau=\tfrac{\hbar\omega_z\tau}{2\hbar}-\tfrac{m\omega^2a^2\tau}{2\hbar}$. Considering $\ket{x_+}$ as the basis that diagonalizes $X_+$ and comparing with the basis $\ket{x}$ that diagonalizes $X$, we have that $\ket{x_+}_{\text{\tiny $X_+$ basis}}=\ket{x_+-a}_{\text{\tiny $X$ basis}}$. Therefore, for any state $\ket{\psi}$, the following equality holds
 \begin{equation}
     \braket{x_+|\psi}_{\text{\tiny $X_+$ basis}}=\braket{x_+-a|\psi}
 \end{equation}
 so that
 \begin{equation}
     \braket{x_+|0}_{\text{\tiny $X_+$ basis}}=\frac{1}{\left(\pi 2\sigma^2\right)^{1/4}} \exp\left[-\frac{(x_+-a)^2}{4\sigma^2}\right]\label{0state+}
 \end{equation}
 and (see page 450 of \cite{Sakurai1994} for the deduction of $\braket{x|n}$ in the $X$ basis)
 \begin{equation}
     \braket{x_+|n}_{\text{\tiny $X_+$ basis}}=(2^{n}n!)^{-1/2}\frac{1}{\left(2\pi \sigma^2\right)^{1/4}} \exp\left[-\frac{(x_+-a)^2}{4\sigma^2}\right]H_n\left(\frac{x_+-a}{\sqrt{2}\sigma}\right)\label{nstate+}
 \end{equation}
where $\sigma=\sqrt{\hbar/(2m\omega)}$ and $H_n(y)$ are Hermite polynomials, that satisfy \cite{Arfken2011}
\begin{equation}
    H_n(y)=\sum_{s=0}^{n/2}(-1)^s (2y)^{n-2s}\frac{n!}{(n-2s)!s!}\label{hermite}
\end{equation} 
From this point on, we will consider all our calculations in the $X_+$ basis. Notice that the term $\frac{P^2}{2m}+\frac{m\omega^2}{2}X_+^{2}$ that appears inside the exponential defining $U_\tau^{+}$ is just a simple harmonic oscillator energy term, with $X$ displaced to $X_+$. Since $[X_+,P]=[X,P]=i\hbar$, then the same algebra rules for the oscillator can be applied here. As a result, we can write
\begin{equation}
    \bra{n}U_\tau^{+}\ket{0}=\mathrm{e}^{-i\theta_\tau}\int_{-\infty}^{\infty}dx_+\braket{n|x_+}\bra{x_+}\exp\left[-\frac{i\tau}{\hbar}\left(\frac{P^2}{2m}+\frac{m\omega^2}{2}X_+^{2}\right)\right]\ket{0}\label{U+n0}
\end{equation}
The term
\begin{equation}
    \bra{x_+}\exp\left[-\frac{i\tau}{\hbar}\left(\frac{P^2}{2m}+\frac{m\omega^2}{2}X_+^{2}\right)\right]\ket{0}
\end{equation}
can be identified as the wave function evolved until time $\tau$ of the state $\ket{0}$, which, in the $X_+$ position basis, is a Gaussian centered around the phase point $(a,0)$ as described in Eq. \eqref{0state+}. This gaussian evolves under a simple harmonic oscillator dynamics until time $\tau$. Using the formalism in \cite{Freire2019,Silva2018}, we obtain
\begin{equation}
    \ba{rl}
    \displaystyle\bra{x_+}\exp\left[-\frac{i\tau}{\hbar}\left(\frac{P^2}{2m}+\frac{m\omega^2}{2}X_+^{2}\right)\right]\ket{0}&\displaystyle=\frac{1}{\left(2\pi \sigma_t^2\right)^{1/4}}\exp\left[-\frac{(x_+-a\cos(\omega t))^2}{\cos(\omega t)(4\sigma^2\cos(\omega t)+\frac{2i\hbar}{m\omega}\sin(\omega t))}-i\frac{m\omega\tan(\omega t)}{2\hbar}x_+^{2}\right]\\
    &\displaystyle=\frac{1}{\left(2\pi \sigma_t^2\right)^{1/4}}\exp\left[-\frac{(x_+-a\cos(\omega t))^2}{4\sigma^2\cos(\omega t)\left(\cos(\omega t)+\frac{i\hbar}{2m\omega\sigma^2}\sin(\omega t)\right)}-i\frac{m\omega\tan(\omega t)}{2\hbar}x_+^{2}\right]\\
    &\displaystyle=\frac{1}{\left(2\pi \sigma_t^2\right)^{1/4}}\exp\left[-\frac{(x_+-a\cos(\omega t))^2\left(\cos(\omega t)-\frac{i\hbar}{2m\omega\sigma^2}\sin(\omega t)\right)}{4\sigma_t^2\cos(\omega t)}-i\frac{m\omega\tan(\omega t)}{2\hbar}x_+^{2}\right]
    \ea\label{evolvedgaussian}
\end{equation}
where
\begin{equation}
    \sigma_t=\sqrt{\sigma^2\cos^2(\omega t)+\frac{\hbar^2}{4m^2\omega^2\sigma^2}\sin^2(\omega t)}
\end{equation}
Now, inserting Eqs. \eqref{nstate+} and \eqref{evolvedgaussian} in Eq. \eqref{U+n0},  considering Eq. \eqref{hermite}, we get
\begin{equation}
    \ba{rl}
    \bra{n}U_\tau^{+}\ket{0}&=(2^{n}n!)^{-1/2}\frac{\mathrm{e}^{-i\theta_\tau}}{\left(\pi^2 4\sigma^2\sigma_t^2\right)^{1/4}} \int_{-\infty}^{\infty}dx_+\exp\left[-\frac{(x_+-a)^2}{4\sigma^2}-\frac{(x_+-a\cos(\omega t))^2\left(\cos(\omega t)-\frac{i\hbar}{2m\omega\sigma^2}\sin(\omega t)\right)}{4\sigma_t^2\cos(\omega t)}-i\frac{m\omega\tan(\omega t)}{2\hbar}x_+^{2}\right]H_n\left(\frac{x_+-a}{\sqrt{2}\sigma}\right)\\
    &=(2^{n}n!)^{-1/2}\frac{\mathrm{e}^{-i\theta_\tau}}{\left(\pi^2 4\sigma^2\sigma_t^2\right)^{1/4}}\sum_{s=0}^{n/2}(-1)^s \frac{n!}{(n-2s)!s!} \times\\
    &\times \int_{-\infty}^{\infty}dx_+\exp\left[-\frac{(x_+-a)^2}{4\sigma^2}-\frac{(x_+-a\cos(\omega t))^2\left(\cos(\omega t)-\frac{i\hbar}{2m\omega\sigma^2}\sin(\omega t)\right)}{4\sigma_t^2\cos(\omega t)}-i\frac{m\omega\tan(\omega t)}{2\hbar}x_+^{2}\right](2\frac{x_+-a}{\sqrt{2}\sigma})^{n-2s}
    \ea
\end{equation}
Now, given that $\tau=\pi/\omega$, then $\sigma_t=\sigma$ and, after some tedious calculations, we obtain 
\begin{equation}
\bra{n}U_\tau^{+}\ket{0}=(2^{n}n!)^{-1/2}\frac{\mathrm{e}^{-\frac{a^2}{2\sigma^2}-i\theta_\tau}}{\left(2\pi \sigma^2\right)^{1/2}}\sum_{s=0}^{n/2}\left(\frac{\sqrt{2}}{\sigma}\right)^{n-2s}(-1)^s \frac{n!}{(n-2s)!s!} f(n,s,a,\sigma),\label{u+0final}
\end{equation}
where
\begin{equation}
    \ba{rl}
    f(n,s,a,\sigma)&=\int_{-\infty}^{\infty}dx_+\exp\left[-\frac{(x_++a)^{2}}{2\sigma^2}\right](x_+)^{n-2s}\\
    &=(-1)^{-2 s} e^{-\frac{a^2}{4 \sigma ^2}} 2^{n-2 s} | \sigma | ^{n-2 s-1} \times\\
    &\times \left(a | \sigma |  \left((-1)^n-(-1)^{2 s}\right) \Gamma \left(\frac{n}{2}-s+1\right) \,
   _1F_1\left(\frac{n}{2}-s+1;\frac{3}{2};\frac{a^2}{4 \sigma ^2}\right)\right.\\
   &+\left.\sigma ^2 \left((-1)^n+(-1)^{2 s}\right) \Gamma \left(\frac{1}{2} (n-2 s+1)\right) \,
   _1F_1\left(\frac{1}{2} (n-2 s+1);\frac{1}{2};\frac{a^2}{4 \sigma ^2}\right)\right)
   \ea
\end{equation}
and $\Gamma(z)$ is the Gamma function \cite{Arfken2011} and $ _1F_1\left(\frac{1}{2} (n-2 s+1);\frac{1}{2};\frac{a^2}{4 \sigma ^2}\right)$ is the Kummer confluent hypergeometric function \cite{Weisstein2003}. Using Eq. \eqref{u+0final} with the parameters described in Fig.~\ref{fig:corr}, we could compute the probabilities $p_{\text{\tiny TPM}}(0, H_{\text{\tiny HO}} , U_\tau, \rho_1)=|\bra{0}U_\tau^{+}\ket{0}|^2$ and $p_{\text{\tiny TPM}}(\hbar\omega, H_{\text{\tiny HO}} , U_\tau, \rho_1)=|\bra{1}U_\tau^{+}\ket{0}|^2$ presented in the Figure 2.

\subsection{Calculating the probabilities $p(n, H_{\text{\tiny HO}}(\tau))$ in subsection ``OBS protocol, non-commutative operators, and result 2''}
Our goal in this subsection is to compute the probability of the Harmonic oscillator part $H_{\text{\tiny HO}} = \frac{P^2}{2m}\otimes \mathbbm{1}_{\text{\tiny s}} + \frac{m\omega^2X^2}{2}\otimes \mathbbm{1}_{\text{\tiny s}} = \hbar\omega (N + 1/2)\otimes \mathbbm{1}_{\text{\tiny s}}$ of the total energy to have some arbitrary value at time $\tau$. To do so, we first notice that since $\ket{n,\pm}$ are eigenstates of $H_{\text{\tiny HO}}$ at time $0$ with eigenvalues $\hbar\omega(n+1/2)$, then the probability of the system of being $\hbar\omega(n+1/2)$ at time $\tau$ is given by
\begin{equation}
    p(n, H_{\text{\tiny HO}}(\tau)) =|\bra{n,-}U_\tau\ket{\alpha,-}|^2.
\end{equation}
where we considered the fact that, since $\sigma_z$ commutes with $U_\tau$, then 
\begin{equation}
    |\bra{n,+}U_\tau\ket{\alpha,-}|^2=0.
\end{equation}

Our strategy to obtain $|\bra{n,-}U_\tau\ket{\alpha,-}|^2$ is to make a transition to the Heisenberg picture. For that, first notice that from the transition from the Schr{\"o}dinger to the Heisenberg picture, we have that given that $\ket{n,\pm}$ are eigenvectors of $H_{\text{\tiny HO}}$, then $U_\tau^{\dagger}\ket{n,\pm}=\ket{n',\pm}$ are eigenstates of the Heisenberg evolved version $U_\tau^\dagger H_{\text{\tiny HO}}U_\tau$, since \cite{Sakurai1994}
\begin{equation}
    U_\tau^\dagger H_{\text{\tiny HO}}U_\tau \ket{n',\pm}=U_\tau^\dagger H_{\text{\tiny HO}}U_\tau U_\tau^{\dagger}\ket{n,\pm}=\left(n+\frac{1}{2}\right)\hbar\omega U^{\dagger}\ket{n,\pm}=\left(n+\frac{1}{2}\right)\ket{n',\pm}
\end{equation}
Now, let us analyze $\{\ket{n'}\}$. We saw in the section ``Trapped ion case study'', that the Heisenberg time-evolved version of the operator  $X' = X \otimes \mathbbm{1}_{\text{\tiny s}} + a (\mathbbm{1}_{\text{\tiny CM}} \otimes \sigma_z)$
is given by $X'(\tau) = -X \otimes \mathbbm{1}_{\text{\tiny s}} - a (\mathbbm{1}_{\text{\tiny CM}} \otimes \sigma_z)$. Therefore, given that $\sigma_z$ does not change in time for the given evolution, it follows that  $X(\tau)=U_{\tau}^\dagger (X\otimes\mathbbm{1}_{\text{\tiny s}})U_\tau= -X \otimes \mathbbm{1}_{\text{\tiny s}} - 2a (\mathbbm{1}_{\text{\tiny CM}} \otimes \sigma_z)$. Moreover, given that $P(\tau) = -P \otimes \mathbbm{1}_{\text{\tiny s}}$, we have that
\begin{equation}
    H_{\text{\tiny HO}}' \coloneqq  \bra{-}H_{\text{\tiny HO}}(\tau)\ket{-}=\bra{-}\left(\frac{P^2(\tau)}{2m} + \frac{m\omega^2X^2(\tau)}{2}\right) \ket{-}=\frac{P^2(0)}{2m}+\frac{m\omega}{2}X^{'2}
\end{equation}
and 
\begin{equation}
     \bra{+}H_{\text{\tiny HO}}(\tau)\ket{-}= \bra{-}H_{\text{\tiny HO}}(\tau)\ket{+}=0
\end{equation}
where $X'=X-2a$. Moreover, notice that $[X',P]=i\hbar$ and $[X',X']=0$, so that, analogously as in the usual harmonic oscillator, by defining \cite{Sakurai1994}
\begin{equation}
    a'=\sqrt{\frac{m\omega}{2\hbar}}(X'+\frac{iP}{m\omega})
\end{equation}
we have that
\begin{equation}
    [a',a^{'\dagger}]=1,\quad N'=a^{'\dagger}a'=\frac{H_{S}'}{\hbar\omega}-\frac{1}{2},\quad [N',a^{'\dagger}]=a^{'\dagger},\qquad [N',a']=-a'.
\end{equation}
So that 
\begin{equation}
    H_{\text{\tiny HO}}'=\hbar\omega \left(N'+\frac{1}{2}\right)
\end{equation}
and the same algebraic structure as the usual harmonic oscillator can be used to deduce that
\begin{equation}
    N'\ket{n'}=n'\ket{n'},\quad a^{'\dagger}\ket{n'}=\sqrt{n'+1}\ket{n'+1},\quad a^{'}\ket{n'}=\sqrt{n'}\ket{n'-1}.
\end{equation}
As a result,
\begin{equation}
    \ket{n'}=\frac{(a^{'\dagger})^{n}}{\sqrt{n!}}\ket{0'}
\end{equation}
The difference from the usual harmonic oscillator appears in the wave-functions, since, considering that $X'\ket{x}=(x-2a)\ket{x}$, then
\begin{equation}
    \bra{x}a'\ket{0'}=\left(x-2a+x_0^2\partial_{x}\right)\braket{x|0}=0,
\end{equation}
where
$x_0=\sqrt{\frac{\hbar}{m\omega}}$, resulting in 
\begin{equation}
    \braket{x|0}=\frac{1}{\pi^{\frac{1}{4}}\sqrt{x_0}}\exp \left[-\frac{1}{2}\left(\frac{x-2a}{x_0}\right)^2\right]
\end{equation}
From this, we obtain, using a method similar to the usual harmonic oscillator case \cite{Sakurai1994}, that
\begin{equation}
    \braket{x|n'}=\left(\frac{1}{\pi^{\frac{1}{4}}\sqrt{2^{n}n!}}\right)\left(\frac{1}{x_0^{n+1/2}}\right)\left(x-2a-x_0^2 \partial_{x}\right)^{n}\exp\left[-\frac{1}{2}\left(\frac{x-2a}{x_0}\right)^2\right]
\end{equation}
These equations are the same as for the eigenvectors $\ket{n}$ of $\frac{P^2}{2m} + \frac{m\omega^2X^2}{2} = \hbar\omega (N + 1/2)$, but displaced by $2a$ (see section 2.3 of Ref. \cite{Sakurai1994}), so that
\begin{equation}
    \braket{x-2a|n}=\braket{x|n'}.
\end{equation}
Considering that the wave function of the coherent state is defined as \cite{Cohen2020}
\begin{equation}
    \braket{x|\alpha}=\mathrm{e}^{i\theta_\alpha}\left(\frac{m\omega}{\pi\hbar}\right)^{\frac{1}{4}}\exp\left[-\left(\frac{x-\bra{\alpha}X\ket{\alpha}}{2\sigma_X^2}\right)^2+\frac{i\bra{\alpha}P\ket{\alpha}x}{\hbar}\right]
\end{equation}
with $\sigma_X = \sqrt{\hbar/(2m\omega)}$ and $\mathrm{e}^{i\theta_\alpha}=\mathrm{e}^{\alpha^{*2}-\alpha^2}$, then
we have that
\begin{equation}
    \ba{rl}
    \braket{n'|\alpha}&=\int_{-\infty}^{\infty}dx\braket{n'|x}\braket{x|\alpha}=\int_{-\infty}^{\infty}dx\braket{n'|x}\mathrm{e}^{i\theta_\alpha}\left(\frac{m\omega}{\pi\hbar}\right)^{\frac{1}{4}}\exp\left[-\left(\frac{x-\bra{\alpha}X\ket{\alpha}}{2\sigma_X^2}\right)^2+\frac{i\bra{\alpha}P\ket{\alpha}x}{\hbar}\right]\\
    &=\int_{-\infty}^{\infty}dx'\braket{n'|x'+2a}\mathrm{e}^{i\theta_\alpha}\left(\frac{m\omega}{\pi\hbar}\right)^{\frac{1}{4}}\exp\left[-\left(\frac{x'+2a-\bra{\alpha}X\ket{\alpha}}{2\sigma_X^2}\right)^2+\frac{i\bra{\alpha}P\ket{\alpha}(x'+2a)}{\hbar}\right]\\
    &=\mathrm{e}^{i\theta_\alpha+i2a\bra{\alpha}P\ket{\alpha}/\hbar}\int_{-\infty}^{\infty}dx'\braket{n'|x'+2a}\left(\frac{m\omega}{\pi\hbar}\right)^{\frac{1}{4}}\exp\left[-\left(\frac{x'+2a-\bra{\alpha}X\ket{\alpha}}{2\sigma_X^2}\right)^2+\frac{i\bra{\alpha}P\ket{\alpha}x'}{\hbar}\right]\\
    &=\mathrm{e}^{i\theta_\alpha+i2a\bra{\alpha}P\ket{\alpha}/\hbar}\int_{-\infty}^{\infty}dx'\braket{n|x'}\left(\frac{m\omega}{\pi\hbar}\right)^{\frac{1}{4}}\exp\left[-\left(\frac{x'+2a-\bra{\alpha}X\ket{\alpha}}{2\sigma_X^2}\right)^2+\frac{i\bra{\alpha}P\ket{\alpha}x'}{\hbar}\right]\\
    \ea
\end{equation}
Defining 
\begin{equation}
    \alpha'=\sqrt{\frac{m\omega}{2\hbar}}(\bra{\alpha}X\ket{\alpha}-2a)+i\frac{1}{\sqrt{2m\hbar\omega}}\bra{\alpha}P\ket{\alpha},
\end{equation}
we can consider the coherent state $\ket{\alpha'}$, such that
\begin{equation}
    \braket{x|\alpha'}=\mathrm{e}^{i\theta_{\alpha'}}\left(\frac{m\omega}{\pi\hbar}\right)^{\frac{1}{4}}\exp\left[\left(\frac{x-\bra{\alpha'}X\ket{\alpha'}}{2\sigma_X^2}\right)^2+\frac{i\bra{\alpha'}P\ket{\alpha'}x}{\hbar}\right],
\end{equation}
where 
\begin{equation}
    \bra{\alpha'}X\ket{\alpha'}=\sqrt{\frac{2\hbar}{m\omega}}\Re(\alpha')=\bra{\alpha}X\ket{\alpha}-2a, \quad \bra{\alpha'}P\ket{\alpha'}=\sqrt{2m\hbar\omega}\Im(\alpha')=\bra{\alpha}P\ket{\alpha},\quad \mathrm{e}^{i\theta_{\alpha'}}=\mathrm{e}^{\frac{\alpha^{'*2}-\alpha^{'2}}{4}}.
\end{equation}
Therefore, 
\begin{equation}
    \braket{n'|\alpha}=\mathrm{e}^{i(\theta_\alpha-\theta_{\alpha'})+i2a\bra{\alpha}P\ket{\alpha}/\hbar}\int_{-\infty}^{\infty}dx'\braket{n|x'}\braket{x'|\alpha'}=\braket{n|\alpha'}\mathrm{e}^{i(\theta_\alpha-\theta_{\alpha'})-i2a\bra{\alpha}P\ket{\alpha}/\hbar}\label{alphalinha}
\end{equation}
Given that $\ket{\alpha'}$ is a proper coherent state, then it follows that
\begin{equation}
    \ket{\alpha'}=\mathrm{e}^{-|\alpha'|^2/2}\sum_{n=0}^{\infty}\frac{\alpha^{'n}}{\sqrt{n!}}\ket{n}
\end{equation}
and
\begin{equation}
    \braket{n|\alpha'}=\mathrm{e}^{-|\alpha'|^2/2}\frac{\alpha^{'n}}{\sqrt{n!}},
\end{equation}
so that, from Eq. \eqref{alphalinha},
\begin{equation}
    |\braket{n'|\alpha}|^2=\mathrm{e}^{-|\alpha'|^2}\frac{|\alpha'|^{2n}}{n!}.
\end{equation}
As a result, it follows that
\begin{equation}
    |\braket{n',+|\alpha,-}|^2=0,\quad  |\braket{n',-|\alpha,-}|^2=\mathrm{e}^{-|\alpha'|^2}\frac{|\alpha'|^{2n}}{n!}
\end{equation}
The probability that $H_{\text{\tiny HO}}$ have an energy $\hbar\omega(n+1/2)$ at time $\tau$ is thus given by
\begin{equation}
    p(n, H_{\text{\tiny HO}}(\tau)) =|\braket{n',-|\alpha,-}|^2=\mathrm{e}^{-|\alpha'|^2}\frac{|\alpha'|^{2n}}{n!}.
\end{equation}
For values different from $\hbar\omega(n+1/2)$ for any $n$, the probability is null.

\subsection{Example of result 2 considering Trapped Ions}

As mentioned in the main text, in specific scenarios, result 2 can be used to obtain the OBS statistics via a probe, an auxiliary energy operator $H_2$. However, as seen in the derivation of result 2, finding a physically intuitive operator $H_2$ and a state $\ket{v}$ such that result 2 can be directly applied is not always straightforward. Nevertheless, there exist special cases where $H_2$ can be explicitly given without introducing an additional subspace ($\mathcal{H}'$ in result 2). We analyze such cases here.

Consider a system $\Omega$ whose total energy is given by
\begin{equation}
    H = H_1 + H_2,
\end{equation}
where both $H_1$ and $H_2$ act on the same Hilbert space $\mathcal{H}$. The unitary evolution is specified as $U = \exp[-iHt/\hbar]$, and the following commutation relations hold:
\begin{equation}
    [U^{\dagger}H_1U, H_1] \neq 0, \quad [U^{\dagger}H_2U, H_2] = 0.
\end{equation}
Since
\begin{equation}
    [\Delta(H_1, U), H_2] =[\Delta(H_1, U), U^{\dagger}H_2U]= 0,
\end{equation}
it follows that for any initial preparation $\rho=\ket{\delta_i(H_1, U)}\bra{\delta_i(H_1, U)}$ as an eigenstate of $\Delta(H_1, U)$, one can measure $H_2$ at times $0$ and $t$, obtaining with 100\% certainty the energies $E_i$ and $E_i' = E_i - \delta_i(H_1, U)$, where $E_i$ is an eigenvalue of $H_2$. By energy conservation, the variation of $H_1$ must necessarily be $\delta_i(H_1, U)$. This represents a special case of result 2, where $H_2\to H_2\otimes\mathbbm{1}_{\mathcal{H}'}$ and $U'=\mathbbm{1}_{\mathcal{H}'}$. 

An illustrative example of this scenario is provided in Fig.~\ref{fig:result2example}, corresponding to the trapped ion system analyzed in Figs. \ref{fig:corr}, \ref{fig:OBSEXAMPLE}, and \ref{fig:comparison} of the main text. In this case, the unitary evolution is given by
\begin{equation}
    U_\tau=\exp[-iH\tau/\hbar],
\end{equation}
where the Hamiltonian takes the form
\begin{equation}
    H=H_{\text{\tiny HO}} \otimes \mathbbm{1}_s+H_{\text{\tiny e}},
\end{equation}
with
\begin{equation}
    H_{\text{\tiny HO}}=(N+\tfrac{1}{2})\hbar\omega, \quad H_{\text{\tiny e}}= \hbar (\omega_z/2 + \Delta_S k_{\text{\tiny SW}} X/2) \otimes \sigma_z.
\end{equation}
By making the substitutions
\begin{equation}
    H_1 \to H_{\text{\tiny HO}} \otimes \mathbbm{1}_s, \quad H_2 \to (H - H_1) = H{\text{\tiny e}}, \quad U\to U_\tau,
\end{equation}
we recover result 2, demonstrating that measuring $H_2$ at times $0$ and $\tau$ yields the same statistics for the variation of $H_1$ under $U_\tau$ as obtained via the OBS protocol.

\begin{figure}[h]
    \centering
    \includegraphics[width=0.8\linewidth]{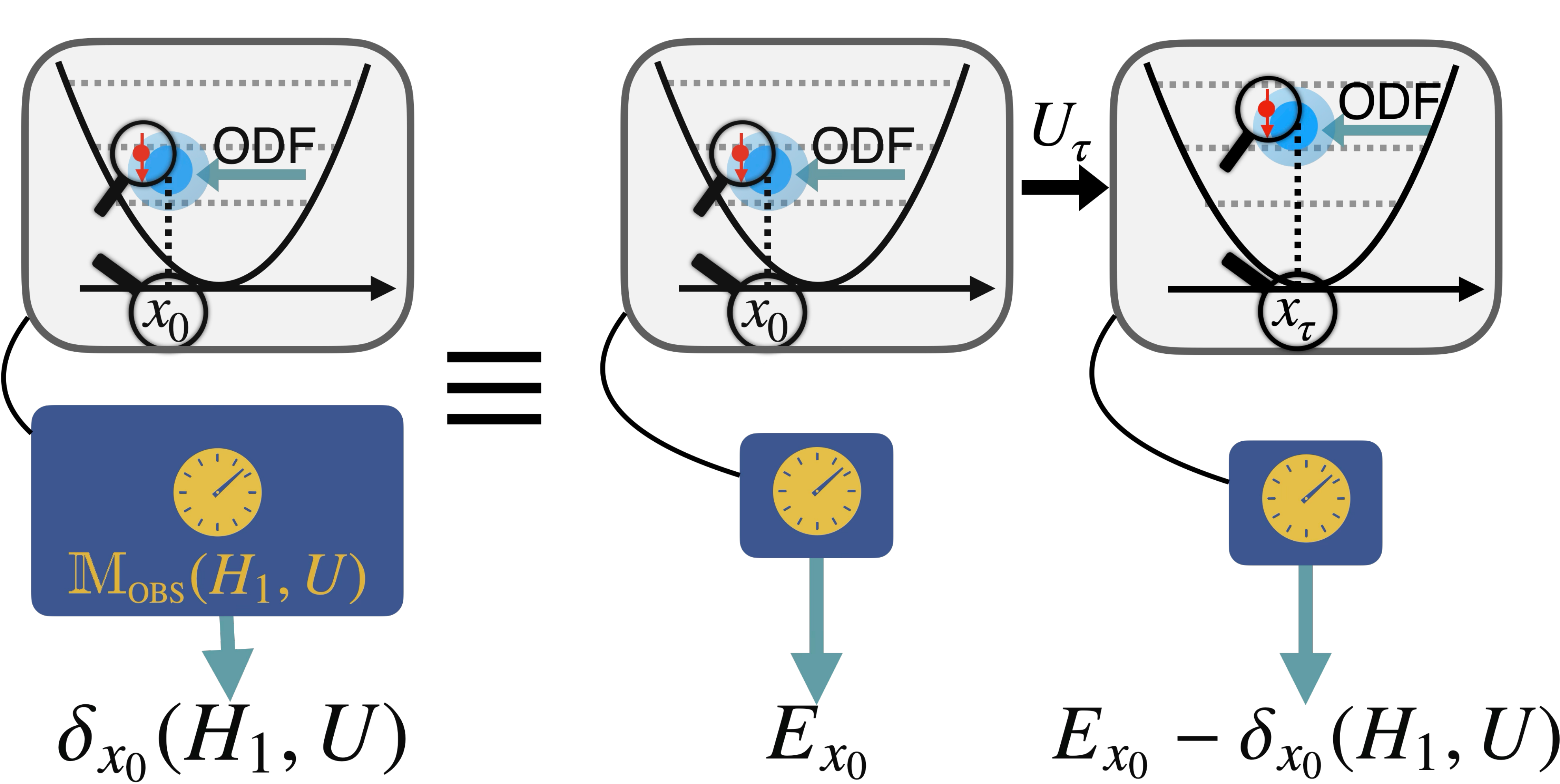}
    \caption{  \textbf{Illustration of the application of result 2 to the trapped Ca$^+$ ion system}. We consider the same system as in Figs. \ref{fig:corr}, \ref{fig:OBSEXAMPLE}, and \ref{fig:comparison} in the main text and define $H_1 = H_{\text{\tiny HO}} \otimes \mathbbm{1}_s$ and $H_2 = H - H_1 = H{\text{\tiny e}}$.  The variation of energy is analyzed over the interval $[0,\tau]$, where $\tau = \pi/\omega$ and the unitary evolution is $U \equiv U_\tau = \exp[-iH\tau/\hbar]$. As detailed in the section ``Trapped ion case study'', the operator of variation of $H_1$ for this interval is $\Delta(H_1, U) = 2m\omega^2 a X \otimes \sigma_z + 2m\omega^2 a^2$, where $a = \hbar \Delta_s k_{\text{\tiny SW}}/ (2m\omega^2)$. The following conditions are satisfied $[U^\dagger H_2 U, H_2] = [U^\dagger H_2 U, \Delta(H_1, U)] = [H_2, \Delta(H_1, U)] = 0$. This scenario represents a special case of result 2. On the left, the OBS protocol is applied to measure $\Delta(H_1, U)$. The system's position $X$ and spin $\sigma_z$ are measured, yielding $x_0$ and $-1$, respectively. This determines the variation of energy as $\delta_{x_0}(H_1, U) = -2m\omega^2 a x_0 + 2m\omega^2 a^2$ which is the eigenvalue of $\Delta(H_1, U)$ associated with the eigenvector $\ket{\delta_{x_0}(H_1, U)} = \ket{x_0, -}$. On the right, result 2 is applied to the same system, approximately prepared in the eigenstate $\ket{\delta_{x_0}(H_1, U)}$. In this case, $\ket{\delta_{x_0}(H_1, U)}$ is an eigenstate of both $H_2$ and $U^\dagger H_2 U$. Measurements of $H_2$ at $t = 0$ and $U^\dagger H_2 U$ at $t = \tau$, via measurements of $X$ and $\sigma_z$, yield the eigenvalues $E_{x_0} = -\hbar (\omega_z/2 + \Delta_S k_{\text{\tiny SW}} x_0/2)$ and  $E_{x_0}'=E_{x_0} - \delta_{x_0}(H_1, U)$, respectively. By the reality condition, the variation of $H_2$ is $-\delta_{x_0}(H_1, U)$. By the conservation of energy (condition 2), the variation of $H_1$ must be $\delta_{x_0}(H_1, U)$, thus showing the equivalence between the measurement results on the left and right side of the figure.}
    \label{fig:result2example}
\end{figure}

In this case, we arrive at a special instance of result 2 where no additional Hilbert space $\mathcal{H}'$ or auxiliary vector $\ket{v}$ is required to obtain the OBS statistics for the variation of $H_1$. However, to fully conform with result 2, we could introduce an auxiliary space $\mathcal{H}'$ by redefining
\begin{equation}
    H_1 \to H_{\text{\tiny HO}} \otimes \mathbbm{1}_s, \quad H_2 \to (H - H_1)\otimes\mathbbm{1}_{\mathcal{H}'} = H{\text{\tiny e}}\otimes\mathbbm{1}_{\mathcal{H}'}.
\end{equation}
Similarly, we set
\begin{equation}
    U\otimes U'\to U_\tau\otimes\mathbbm{1}_{\mathcal{H}'}.
\end{equation}
With this formulation, the same result would be obtained for any normalized state $\ket{v} \in \mathcal{H}'$.

\end{document}